\def\N{\mathbb N}
\def\Z{\mathbb Z}
\def\R{\mathbb R}
\def\Q{\mathbb Q}
\def\C{\mathbb C}
\def\A{\mathcal A}
\def\B{\mathcal B}
\def\pfz{\begin{proof}}
\def\pfk{\end{proof}}
\documentclass[12pt]{amsart}
\usepackage[lmargin=2.5cm, rmargin=2.5cm,bottom=2.5cm,top=2.5cm]{geometry}
\usepackage{amssymb,amsmath,mathdots,amsthm}
\usepackage{graphicx}

\usepackage{subfigure}
\usepackage{comment}
\usepackage{algorithm,algorithmic}
\usepackage{enumerate}

\newtheorem{thm}{Theorem}[section]
\newtheorem{prop}[thm]{Proposition}
\newtheorem{coro}[thm]{Corollary}
\newtheorem{lem}[thm]{Lemma}

\newtheorem{ex}[thm]{Example}
\newtheorem{remark}[thm]{Remark}
\newtheorem{claim}[thm]{Claim}

\title { On the spectra of Pisot-cyclotomic numbers  }
\author{Kevin G. Hare}
\address{Department of Pure Mathematics, University of Waterloo, Waterloo, Ontario, Canada N2L 3G1}

\author{Zuzana Mas\'akov\'a, Tom\'a\v s V\'avra}
\address{Department of Mathematics, FNSPE, Czech Technical University in Prague, Trojanova 13, 120~00 Praha 2, Czech Republic}

\date{\today}

\begin{document}
\begin{abstract}
We investigate the complex spectra
\[ X^\A(\beta)=\left\{\sum_{j=0}^na_j\beta^j : n\in\N,\ a_j\in\A\right\} \]
where $\beta$ is a quadratic or cubic Pisot-cyclotomic number and the alphabet $\A$ is given by $0$ along with a finite collection of roots of unity.
Such spectra are discrete aperiodic structures with crystallographically forbidden symmetries. We discuss in general terms
under which conditions they possess the Delone property required for point sets modeling quasicrystals. We study the corresponding Voronoi tilings and we relate these structures to quasilattices arising from the cut and project method.
\end{abstract}
\maketitle
\allowdisplaybreaks

\section{Introduction}

Since the discovery of quasicrystals~\cite{schechtman}, it has been recognized the need of using irrational numbers
in describing mathematical models of these aperiodic structures with non-crystallographically forbidden symmetries.
The most common example are quasicrystals revealing five-fold (or ten-fold) symmetry, for which a model based on the golden ratio $\tau=\frac12(1+\sqrt5)$ is required, see e.g.~\cite{Icosians}. Quasicrystals with symmetry of order 8 and 12
correspond to other quadratic irrationalities. Only recently, quasicrystals with cubic irrationality have been observed~\cite{coloid}, see also~\cite{steurer7}. Naturally, rotation symmetries are connected to the so-called Pisot-cyclotomic numbers, whose most simple example is the golden ratio, and the quadratic and cubic cases are connected precisely to symmetries of order 5 (10), 7 (14), 8, 9 (18), and 12.

In~\cite{BuFrGaKr}, a convenient tool for studying
quasicrystals has been identified, namely number representation in systems with irrational base. Although such
non-standard number systems have been studied long before~\cite{renyi}, their connection to quasicrystal models
has only been exploited in recent years. In this paper we study originally a number-theoretic object, the so-called
spectrum of a number. By its nature, the spectrum is self-similar and has a rotational symmetry.
We demonstrate the utility of these structures for modelling quasicrystals.

The problem of the so-called spectrum of a real number $\beta>1$ has been introduced by Erd\"os et al.~\cite{Erdos}.
The spectrum of $X^{\A}(\beta)$ for $\A=\{0,1,\dots,r\}\subset\N$ is the set of polynomials with non-negative integer
coefficients in the alphabet $\A$, evaluated in $\beta$,
$X^\A(\beta)=\big\{\sum_{j=0}^na_j\beta^j : n\in\N,\ a_j\in\A\big\}$.
Specific properties are valid if $\beta$ is taken to be a Pisot number, i.e.\ an algebraic integer with conjugates smaller
than $1$ in modulus. In particular, there is only a finite number of distances between consecutive points of $X^\A(\beta)$,
and the distance sequence can be generated by a substitution~\cite{FengWen02}.
A generalization of the spectra to the complex plane was considered by Hejda and Pelantov\'a~\cite{HePe15}, who define the
spectrum of a complex Pisot number $\beta$ and study its Voronoi tessellation for the case of cubic complex Pisot units.

We consider a generalization in another direction. The base $\beta$ remains a real number $>1$, but the alphabet $\A$ of
digits is taken to be a finite set of complex numbers. In particular, we study the cases leading to $X^\A(\beta)$
as subsets of cyclotomic integers, providing thus discrete aperiodic structures with crystallographically forbidden symmetries.
We focus on the cases where $\beta$ is a quadratic and cubic Pisot-cyclotomic number and the alphabet $\A$ is formed
by the origin and the vertices of a regular polygon.

Recall a Pisot number is a real algebraic integer $\beta > 1$ such that
all of $\beta$'s Galois conjugates are strictly less than $1$ in absolute value.
A Pisot-cyclotomic number of order $n$ is a Pisot number $\beta$ such that $\Z[\beta]=\Z[2\cos\frac{2\pi}{n}]$.
Notice that if $\beta$ is a Pisot-cyclotomic number of order $n$, then
    $\beta$ will have degree $\phi(n)/2$, where $\phi$ is the Euler
    totient function.
In the case when $n = 5, 8, 10$ or $12$, this means that $\beta$ is
    quadratic.
In the case when $n = 7, 9, 14$ or $18$ this means that $\beta$ is cubic.
We list the set of all quadratic and cubic Pisot-cyclotomic numbers,
 in Table~\ref{tab:pc}.

\begin{table}[ht]
\begin{center}
\renewcommand{\arraystretch}{1.2}
\begin{tabular}{llll}
\hline
order   & name      & approximate value & minimal polynomial \\
\hline
5 or 10 & $\tau$    &  1.618033989      & $x^2-x-1$               \\
        & $\tau^2$  &  2.618033989      & $x^2-3x+1$              \\
7 or 14 & $\lambda$ &  2.246979604      & $x^3-2 x^2-x+1$         \\
        &           &  4.048917340      & $x^3-3 x^2-4 x-1$       \\
        &           &  5.048917340      & $x^3-6 x^2+5 x-1$       \\
        &           &  20.44264896      & $x^3-20 x^2-9 x-1$      \\
        &           &  21.44264896      & $x^3-23 x^2+34 x-13$    \\
8       & $\delta$  &  2.414213562      & $x^2-2 x-1$             \\
        &           &  3.414213562      & $x^2-4 x+2$             \\
9 or 18 & $\kappa$  &  2.879385242      & $x^3-3 x^2+1$           \\
        &           &  7.290859369      & $x^3-6 x^2-9 x-3$       \\
        &           &  8.290859369      & $x^3-9 x^2+6 x-1$       \\
12      & $\mu$     &  2.732050808      & $x^2-2 x-2$             \\
        &           &  3.732050808      & $x^2-4 x + 1$           \\
\hline\\
\end{tabular}
\caption{Pisot cyclotomic numbers of degree 2 and 3, as found in \cite{BellHare}. Certain of these Pisot-cyclotomic numbers we denote with special symbols,
    as we will be using them in the future.}
\label{tab:pc}
\end{center}
\end{table}

Pisot-cyclotomic numbers of order $n$ play a crucial role in models for quasilattices and quasicrystals.
The standard method used for such models is the cut-and-project method.  In this paper we present an alternate
model that comes from the spectra of a Pisot-cyclotomic number with respect to a specific digit set.
We show that this alternate model shares many of the desired properties with the cut-and-project method.
In some cases the two models coincide exactly, whereas in other we can show that these are two distinctly different models.
Thus, we give an alternative model that may find future applications.

For a subset $\Lambda \subset \C$ to be a quasilattice, one usually requires a number of specific properties.
\begin{enumerate}[(i)]
\item \label{pr:ro} Rotational symmetry:  For a root $\omega$ of unity, one has $\omega \Lambda = \Lambda$.
\item Dilation:  There exists a $\beta \in \R$, $\beta \neq \pm 1$ such that $\beta \Lambda \subset \Lambda$.
\item \label{pr:ud} Uniform discreteness:  There exists $\epsilon_1 > 0$ such that
    for all $z_1, z_2 \in \Lambda$ either $z_1 = z_2$ or $|z_1 - z_2| >     \epsilon_1$.
\item \label{pr:rd} Relative density:  There exists an $\epsilon_2 > 0$ such that
    for all $z \in \C$ we have $B_{\epsilon_2}(z) \cap \Lambda \neq \emptyset$.
\item \label{pr:flc} Finite local complexity:  $\Lambda$ is uniformly discrete, relatively dense, and for each $\epsilon$, the intersection
   of the set $\Lambda - \Lambda$ with the open ball $B_\epsilon(0)$ is finite.
\end{enumerate}

Sets satisfying properties~\eqref{pr:ud} and~\eqref{pr:rd} are called Delone sets. The supremum of parameters $\epsilon_1$ gives the double of the
packing radius of $\Lambda$, and the infimum of all $\epsilon_2$ with the required property is called the covering radius of $\Lambda$; 
    we denote it by $r_c$.
Property (\ref{pr:flc}) in fact means that there are only finitely many local configurations of any given radius.
The closest neighbourhoods of points in a Delone set $\Lambda$ can be described using the so-called Voronoi tiles.
A Voronoi tile of a point $x\in\Lambda$ is the set of points in $\C$ closer to $x$ than to any other point of $\Lambda$.
Voronoi tiles of all points of $\Lambda$ form a tiling of $\C$. If $\Lambda$ has property (\ref{pr:flc}),
the corresponding tiling is composed of copies of only finitely many Voronoi tiles.

All the (not necessarily primitive) $n$-th roots of unity form vertices of a regular $n$-gon in the complex plane,
let us denote by them $\Delta_n=\{\omega^j : j\in\Z\}$. We will consider the following sets,
$$
X^{\A}(\beta) = \Big\{\sum_{j=0}^na_j\beta^j : n\in\N,\ a_j\in\A\Big\}\,,\quad\text{ where }\A=\A_n=\Delta_n\cup\{0\}\,.
$$
This is the spectra of $\beta$ with respect to the alphabet $\A$.
Note that by definition, we trivially have rotational symmetry as $\omega^k X = X$.
We further see that we have dilation, as $\beta X \subset X$.
In fact even strong is true, for any 
    $A := a_{k-1} \beta^{k-1} + \dots + a_0$ with $a_i \in \mathcal{A}$ one can show that 
    $\beta^{k} (X + \frac{A}{\beta-1}) \subset X + \frac{A}{\beta-1}$.
Hence this satisfies dilations with multiple possible centers.
In Section \ref{sec:general} we discuss this set with respect to the other
    desired properties.

We show under what conditions $X^{\A}(\beta)$ gives a quasilattice satisfying properties~\eqref{pr:ro}--\eqref{pr:flc}.
We will compare the corresponding spectrum with the cut-and-project sets defined for these cases. It turns out that sometimes, the acceptance window is given by a polygon, in other situations, we need to consider cut-and-project sets whose acceptance window have fractal boundary. Note that such structures appeared already in~\cite{BaKlSc} in connection to aperiodic tilings.
We also provide some information about the Voronoi tiling of the spectrum.
\section{Uniform discreteness and relative density of the spectrum}
\label{sec:general}

By a variation of the Proof of Theorem 2.9 of Ngai and Wang~\cite{NgaiWang} one can show that the spectrum of a Pisot
number $\beta$ with alphabet $\B$ being a finite subset of an imaginary quadratic extension of $\Q(\beta)$
is always a uniformly discrete set.

\begin{thm}
\label{thm:ud}
Let $\beta$ be a Pisot number and let $\omega$ be a zero of a quadratic polynomial $g(x)\in\big(\Q(\beta)\big)[x]$ with
negative discriminant. If $\B\subset\Q(\omega)$ is a finite alphabet, then the set $X^{\B}(\beta)=
\Big\{\sum_{j=0}^na_j\beta^j : n\in\N,\ a_j\in\B\Big\}$
is uniformly discrete.
\end{thm}

\begin{proof}
Let first $F \subset \Z$ be a finite set of integers.
A well known property of Pisot numbers is that
    $X^{F}(\beta)$ is uniformly discrete.
See for example~\cite{Garsia61}.
It is not hard to see that this is still true if we take
    $F \subset \Q(\beta)$ with $|F| < \infty$.
For if $a_i=\frac{b_0^{(i)}+b_1^{(i)}\beta+\dots+b_{d-1}^{(i)}\beta^{d-1}}{Q}$ with $b_j^{(i)},Q\in\Z$, then one can find an integer alphabet $G$ (consisting of sums of $b_j^{(i)}$) such that
$X^F(\beta)\subseteq\tfrac1Q X^G(\beta).$

Consider a finite set $\B\subset\Q(\omega)$. Since $\omega$ is quadratic over $\Q(\beta)$,
every element in $\B$ can be written as $a + b \omega$ where
    $a, b \in \Q(\beta)$.
Considering $\B_1 = \{a : a + b \omega \in \B\}$ and
            $\B_2 = \{b : a + b \omega \in \B\}$
    we see that
    \[ X^{\B}(\beta) \subset X^{\B_1}(\beta) + \omega X^{\B_2}(\beta)\,. \]
As the polynomial $g$ has negative discriminant, $\omega$ is complex and therefore
the right hand side of the above is a lattice composed of two uniformly discrete
    subsets of $\R$ (one of them rotated by $\omega$). We see that
    the right hand side is uniformly discrete.
Hence so is the left hand side as required.
\end{proof}

We see that Theorem \ref{thm:ud} can be applied in the case of the 
    spectrum of Pisot numbers under consideration. 

\begin{coro}
\label{c:ud}
Let $\beta$ be a Pisot-cyclotomic number of order $n$, let $\omega$ be a primitive $n$-th root of unity. If
$\B\subset\Q(\omega)$ is a finite set, then the spectrum $X^{\B}(\beta)$ is uniformly discrete.
\end{coro}

\begin{proof}
We have $\omega+\omega^{-1}=2\cos\frac{2\pi}{n}$, which implies that $\omega$ is a zero of the quadratic polynomial
$g(x)=x^2-2\cos\frac{2\pi}{n}x+1$. By definition of Pisot-cyclotomic numbers, $2\cos\frac{2\pi}{n}\in\Z[\beta]$,
thus $g$ is a polynomial over $\Q(\beta)$. Its discriminant $4\cos^2-4$ is negative (except for $n=1,2$ that corresponds to the trivial case $\beta\in\Z$). Thus we can apply Theorem~\ref{thm:ud}
to obtain the result.
\end{proof}

\begin{coro}\label{thm:flc}
Let $\beta>1$ be a Pisot-cyclotomic number of order $n$ and let $\B\subset\Q(\beta)$ be a finite set.
If $X^\B(\beta)$ is relatively dense, then it is of finite local complexity.
\end{coro}

\begin{proof}
We need to check whether $X^\B(\beta)-X^\B(\beta)$ has finitely many elements in the ball $B_\epsilon(0)$ of any radius $\epsilon$.
Since $X^\B(\beta)-X^\B(\beta)\subset X^{\B-\B}(\beta)$ and $\B-\B\subset\Q(\omega)$, we can use Corollary~\ref{c:ud} to obtain the statement.
\end{proof}

Under the assumption that $X^\B(\beta)$ is discrete, we can compute all the points of the finite set $X^\B(\beta)\cap B_R(0)$. The following algorithm is based on the observation that if $z\in X^\B(\beta)$
such that $|z|>\frac{\max\{|a|\ :\ a\in\B\}}{\beta-1}$, then $|\beta z+a| > |z|.$
\begin{itemize}
  \item  Set $X_0:=\B$ and $M:=\max\{R,\frac{\max\{|a|\ :\ a\in\B\}}{\beta-1}\}$.
  \item Compute $\tilde X_n : = X_{n-1}\cup (\bigcup_{a\in\B} \beta X_{n-1}+a)$.
  \item Set $X_n := \tilde X_n\cap B_M(0)$.
  \item Repeat until $X_n=X_{n-1}.$
\end{itemize}

%

Let us concentrate on the question about relative density of the spectra in the complex plane. We provide a number of results for general
real base and general complex alphabet.
The next is an adaptation of Theorem 1.1 \cite{HePe15},
    where spectrum of complex numbers with an alphabet of
    integer digits is considered.

\begin{thm}
\label{thm:nrd}
Let $\beta\in\C,|\beta|>1$ and let $\B\subset\C$ be finite. 
If $\#\B<|\beta|^2$, then $X^\B(\beta)$ is not relatively dense.
\end{thm}

\begin{proof}
Denote $M:=\max\{|a|:a\in\B\}$.
For any relatively dense set $\Lambda\subset\C$, we have
\begin{equation}\label{eq:fr}
               \liminf_{r\to\infty} \frac{\#\bigl(\Lambda\cap B_r(0)\bigr)}{r^2} >0\,.
\end{equation}

Set $\Lambda:= X^\B(\beta)$. We will show that such $\Lambda$ does not satisfy~\eqref{eq:fr}.
First notice that for any $r> \frac{M}{|\beta|}$ we have
\[
               \Lambda\cap B_{|\beta| r-M}(0)\subseteq \beta\bigl(\Lambda\cap B_r(0)\bigr)+\B\,,
\]
which then clearly implies that
\begin{equation}\label{eq:frL}
               \# \bigl(\Lambda\cap B_{|\beta| r-M}(0)\bigr)\leq \#\B\cdot\#\bigl(\Lambda\cap B_r(0)\bigr)
.\end{equation}
Indeed, consider $x=\sum_{j=0}^k a_j\beta^j$ with $a_j\in\B$  and such that $|x|\leq|\beta| r-M$.
Then $y:= (x-a_0)/\beta=\sum_{j=1}^k a_j\beta^{j-1}\in \Lambda$ and
$|y|\leq (|x|+|a_0|)/|\beta|\leq (|\beta| r-M+M)/|\beta|=r$.
Since $x=\beta y+a_0$, the inclusion is valid.

By~\eqref{eq:fr}, it is enough to construct a sequence $(r_k)_{k\geq 0}$ such that $r_k\to\infty$ and
\begin{equation}\label{eq:h}
               \lim_{k\to\infty} \frac{\#\bigl(\Lambda\cap B_{r_k}(0)\bigr)}{r_k^2} = 0\,.
\end{equation}
Set $r_{k+1}= |\beta| r_k-M$. We have $r_k=|\beta|^{k}\big(r_0-\frac{M}{|\beta|-1}\big)+\frac{M}{|\beta|-1}$, and so
any choice $r_0>\frac{M}{|\beta|-1}$ ensures that $r_k\to\infty$ and $r_{k+1}/r_k\to|\beta|$.
Then \eqref{eq:frL} gives
$$
\# \bigl(\Lambda\cap B_{r_{k+1}}(0)\bigr)\leq \#\B\cdot\#\bigl(\Lambda\cap B_{r_k}(0)\bigr)\,.
$$
Since $\Lambda= X^\B(\beta)$ always contains the origin, the latter is non-zero and we can write
$$
\frac{\#\bigl(\Lambda\cap B_{r_{k+1}}(0)\bigr)/{r_{k+1}^2}}
{\# \bigl(\Lambda\cap B_{r_k}(0)\bigr)/{r_{k}^2}} \leq \#\B\frac{r_{k}^2}{r_{k+1}^2} \xrightarrow{k\to\infty}
               \frac{\#\B}{|\beta|^2}<1\,,
$$
which proves~\eqref{eq:h}. Thus $X^\B(\beta)$ is not relatively dense.
\end{proof}

The condition $\#\B\geq |\beta|^2$ for relative density in Theorem~\ref{thm:nrd} is necessary, but not sufficient.
We further
provide necessary and sufficient conditions for relative density of the spectrum $X^{\B}(\beta)$ in terms of representation of
complex numbers in base $\beta$  with digits in $\B$.

Given a base $\gamma\in\C,|\gamma|>1$ and a finite alphabet $\B\subset\C$, we will say that $\C$ is $(\gamma,\B)$-representable, if for any $z\in\C$,
there exists a convergent series
$$
z=\sum_{i\leq l}a_i\gamma^{i}\,,\quad l\in\N,\ a_i\in\B\,,
$$
which is said to be a $(\gamma,\B)$-representation of $z$. For the given representation of $z$, we define its integer and fractional parts,
$$
\mathrm{inp}(z):=\sum_{i=0}^{l}a_i\gamma^i\quad\text{ and }\quad\mathrm{frp}(z):=\sum_{i=-\infty}^{-1} a_i\gamma^i\,.
$$

For $\gamma\in\C, |\gamma|<1$ define the set $K:=K(\gamma,\B)\subset\C$ as follows:
\begin{equation}\label{eq:K}
K(\gamma,\B) := \left\{\sum_{i=0}^{+\infty}a_i\gamma^{i}: a_i\in\B\right\}\,.
\end{equation}

\begin{remark}
It can be shown that $K$ is a compact set. In fact, it is the attractor of the iterated function system (IFS) with contraction $\gamma$ and offsets $a\in\B$.
In particular, $K=K(\gamma,\B)$ is the unique non-empty compact set satisfying
$$
K = \bigcup_{a \in \B} (\gamma K + a) \,.
$$
\end{remark}

\begin{prop}\label{p:rep-rd}
Let $\beta\in\C, |\beta|>1$ and let $\B\subset\C$ be finite.
If $\C$ is $(\beta,\B)$-representable, then $X^\B(\beta)$ is relatively dense.
Moreover, for its covering radius, we have $r_c\leq M/(|\beta|-1)$, where $M=\max\{|a|:a\in\B\}$.
\end{prop}

\pfz
By assumption, for any $z\in\C$ we have an expansion $z=\sum_{i=-\infty}^l a_i\beta^{i}$, $a_i\in\B$. 
Let $\tilde{z} = \mathrm{inp}(z) = 
\sum_{i= 0}^l a_i\beta^i$.
Then $\tilde{z} \in X^\B(\beta)$ and 
    $|z-\tilde{z}|=|\sum_{i=-\infty}^{-1} a_i\beta^i|\leq M/(|\beta|-1)$.
\pfk

In order to present an opposite statement, we need an additional requirement. Recall that we say a set $\Lambda\subset\C$ is discrete,
if for every $\epsilon>0$, the intersection $\Lambda\cap B_{\epsilon}(0)$ is finite.

\begin{prop}\label{p:rd-rep}
Let $\beta\in\C,|\beta|>1$ and let $\B\subset\C$ be finite.
If $X^\B(\beta)$ is relatively dense and discrete, then $\C$ is $(\beta,\B)$-representable.
\end{prop}

\pfz
Since $X^\B(\beta)$ is relatively dense, we can find, for a given $z\in\C$, an element $y\in X^\B(\beta)$ such that $|z-y|\leq r_c$, where $r_c$ is the covering radius
of $X^\B(\beta)$. Similarly, we find a sequence $(y_n)_{n\in\N}\subset X^\B(\beta)$ such that for each $n\in\N$, we have $|\beta^nz-y_n|\leq r_c$.
Obviously, $\lim_{n\to\infty} y_n/{\beta^{n}} = z$. We have to prove the existence of a series representing $z$.

Since $y_n\in X^\B(\beta)$, we can write
$$
\frac{y_n}{\beta^n} = \sum_{i=-n}^{k(n)}a_i^{(n)}\beta^i\,.
$$
Note that $\mathrm{frp}(y_n/\beta^n)\in K(1/\beta,\B)$ and $K$ is compact, thus we can find a subsequence $(w_n)_{n\in\N}$ of $(y_n/\beta^n)_{n\in\N}$ such that
$\mathrm{frp}(w_n)_{n\in\N}$ is convergent with $\lim_{n\to\infty} \mathrm{frp}(w_n)= w=\sum_{i=-\infty}^{-1}b_i\beta^i\in K$.

Also the sequence
$\mathrm{inp}(w_n)=w_n-\mathrm{frp}(w_n)\in X^\B(\beta)$ is convergent. Since $X^\B(\beta)$ is discrete, $\mathrm{inp}(w_n)$ is eventually constant,
i.e. $\mathrm{inp}(w_n)=\sum_{i= 0}^Kb_i\beta^i$ for $n\geq N_0$.

Altogether, $\sum_{i=-\infty}^Kb_i\beta^i$ is a $(\beta,\B)$-representation of $z$.
\pfk

\begin{prop}\label{prop:rd-K}
Let $\beta\in\C, |\beta|>1$ and let $\B\subset\C$ be finite. Suppose that $X^\B(\beta)$ is discrete.
Then $\C$ is $(\beta,\B)$-representable if and only if $0\in \mathrm{int}\big(K(1/\beta,\B)\big)$.
\end{prop}

\pfz
First suppose that $0\in \mathrm{int}\big(K(1/\beta,\B)\big)$, and consider $z\in\C$. Necessarily, there exists $n\in\N$ such that $z/\beta^n\in K$. Therefore
$z/\beta^n=\sum_{i=-\infty}^{0}a_i\beta^i$ for some $a_i\in\B$, whence $z$ has a $(\beta,\B)$-representation $z=\sum_{i=-\infty}^{n}a_{i-n}\beta^{i}$.

Suppose on the other hand that every $z\in\C$ has a $(\beta,\B)$-representation.
If $z\in B_1(0)$, then for its fixed representation we have
$$
\mathrm{inp}(z)=z-\mathrm{frp}(z)\in B_C(0)\cap X^\B(\beta)\,,
$$
where $C=1+\frac{M}{|\beta|-1}$, $M:=\max\{|a| : a\in\A\}$. Therefore the set $\{ \mathrm{inp}(z) : z\in B_1(0)\}\subset B_C(0)$ is finite and its elements can be picked from
$\{\sum_{i=0}^L a_i\beta^i : a_i\in\B \}$ for some fixed $L\in \N$. 
This proves that all elements in $\beta^{-L}B_C(0)$ have a representation 
    with the integer part being $0$.
That is $\beta^{-L}B_C(0)\subset K(1/\beta,\B)$, i.e. $0\in \mathrm{int}\big(K(1/\beta,\B)\big)$.
\pfk


\begin{thm}\label{thm:ekvivalence}
Let $\beta\in\C, |\beta|>1$ and $\mathcal B\subset\C$ be finite. Suppose that $X^\B(\beta)$ is discrete. Then the following statements are equivalent.
\begin{enumerate}
  \item $X^\B(\beta)$ is relatively dense.
  \item $\C$ is $(\beta,\B)$-representable.
  \item $0\in\mathrm{int} (K(1/\beta,\B))$.
  \item $\sum_{i=-N}^0a_i\beta^{i}\in\mathrm{int} (K(1/\beta,\B))$ for every $N\in\N$ and $a_i\in\B$.
\end{enumerate}
\end{thm}

\begin{proof}
The equivalences (1),(2), and (3) follow from Propositions~\ref{p:rep-rd}, \ref{p:rd-rep}, and~\ref{prop:rd-K}.

Let $x = \sum_{i=0}^Na_i\beta^{-i}$. If (3) holds, then we have
$$
x\in\mathrm{int}\big(x+\beta^{-N-1}K(1/\beta,\B)\big)\subset K(1/\beta,\B).
$$
The implication (4)$\implies$(3) is trivial.
\end{proof}

\section{Spectrum with polygonal alphabet}

Let us focus on the specific example of the spectrum, namely when $\beta$ is a Pisot-cyclotomic number of order $n$ and the alphabet of digits
is formed by the vertices of the regular $n$-gon and the origin.
In particular, $\A:=\Delta_n\cup\{0\}$. Recall we have denoted $\Delta_n=\{\omega^j:j=1,\dots,n\}$ with $\omega=\exp{\frac{2\pi i}{n}}$. Since $\A\subset\Q(\omega)$,
it follows from Theorem~\ref{thm:ud} that the spectrum $X^{\A}(\beta)$ is uniformly discrete.
Based on Theorem~\ref{thm:nrd}, the only options for relatively dense spectrum among the quadratic and cubic Pisot-cyclotomic numbers from Table~\ref{tab:pc} are
$\tau$ (with $n=5$ and $n=10$),
$\tau^2$ (with $n=10$),
$\lambda$ (with $n=7$ and $n=14$),
$\delta$ (with $n=8$),
$\kappa$ (with $n=9$ and $n=18$), and
$\mu$ (with $n=12$).

In order to prove/disprove relative density, we use Propositions~\ref{p:rep-rd} and~\ref{p:rd-rep}, and the results of Herreros~\cite{herreros}. In his thesis, he considers
representations of complex numbers in a real base $\gamma$ with polygonal alphabet. In particular, he proves the following statement, which we rephrase with our notation.

\begin{table}[ht]
\renewcommand{\arraystretch}{1.3}
\begin{center}
\begin{tabular}{cccccc}
\hline
order $n$  & base   & approx.\ value & $1+2\cos\frac{2\pi}{n}$ & $1+\cos\frac{\pi}{n}+(\cos\frac{\pi}{n})^2$ &relative density  \\
\hline
5       & $\tau$   & 1.618033989    & 1.618033989  && YES       \\
10      & $\tau$   & 1.618033989    & 2.618033989  && YES       \\
        & $\tau^2$ & 2.618033989    & 2.618033989  && YES       \\
7       & $\lambda$& 2.246979604    & 2.246979604  && YES       \\
14      & $\lambda$& 2.246979604    & 2.801937736   && YES       \\
8       & $\delta$ & 2.414213562    & 2.414213562   && YES       \\
9       & $\kappa$ & 2.879385242    &               & 2.822714843 & NO   \\
18      & $\kappa$ & 2.879385242    & 2.879385242   && YES       \\
12      & $\mu$    & 2.732050808    & 2.732050808   && YES       \\
\hline\\
\end{tabular}
\caption{Pisot cyclotomic numbers \& relative density}
\label{tab:rd}
\end{center}
\end{table}

\begin{prop}\label{p:herreros}
Let $\gamma>1$ and let $\A:=\Delta_n\cup\{0\}$. If $\gamma\leq 1+2\cos\frac{2\pi}{n}$, then $\C$ is $(\gamma,\A)$-representable.
On the other hand, if $\gamma>1+\cos\frac{\pi}{n}+(\cos\frac{\pi}{n})^2$ for $n$ odd, or $\gamma>2+\cos\frac{2\pi}{n}$ for $n$ even, then
$\C$ is not $(\gamma,\A)$-representable.
\end{prop}

It is important to notice that for our quadratic and cubic Pisot-cyclotomic numbers one has
$$
\begin{aligned}
\tau &=1+2\cos\frac{2\pi}{5},\\
\tau^2&=1+2\cos\frac{2\pi}{10},
\end{aligned}
\qquad
\begin{aligned}
\delta&=1+2\cos\frac{2\pi}{8},\\
\mu&=1+2\cos\frac{2\pi}{12},
\end{aligned}
\qquad
\begin{aligned}
\lambda&=1+2\cos\frac{2\pi}{7},\\
\kappa&=1+2\cos\frac{2\pi}{18}.
\end{aligned}
$$
This allows one to combine easily Proposition~\ref{p:rep-rd} with the statement of Herreros to obtain that
all cases of the spectra are relatively dense, except the case that $\beta=\kappa$ and $\A=\Delta_9\cup\{0\}$. Here, one must check
that
$$
\kappa>1+\cos\frac{\pi}{9}+(\cos\frac{\pi}{9})^2\,,
$$
and thus Proposition~\ref{p:rd-rep} gives that the spectrum is not relatively dense.
The results are summarized in Table~\ref{tab:rd}.

In order to put these results in context of Proposition \ref{prop:rd-K}, see Figure~\ref{fig:rd} where the set $K(\beta,\A)$ is drawn for the considered cases.

\begin{figure}[ht]
\subfigure[$\tau$ with $n=5$]{\includegraphics[scale=0.17,angle=270]{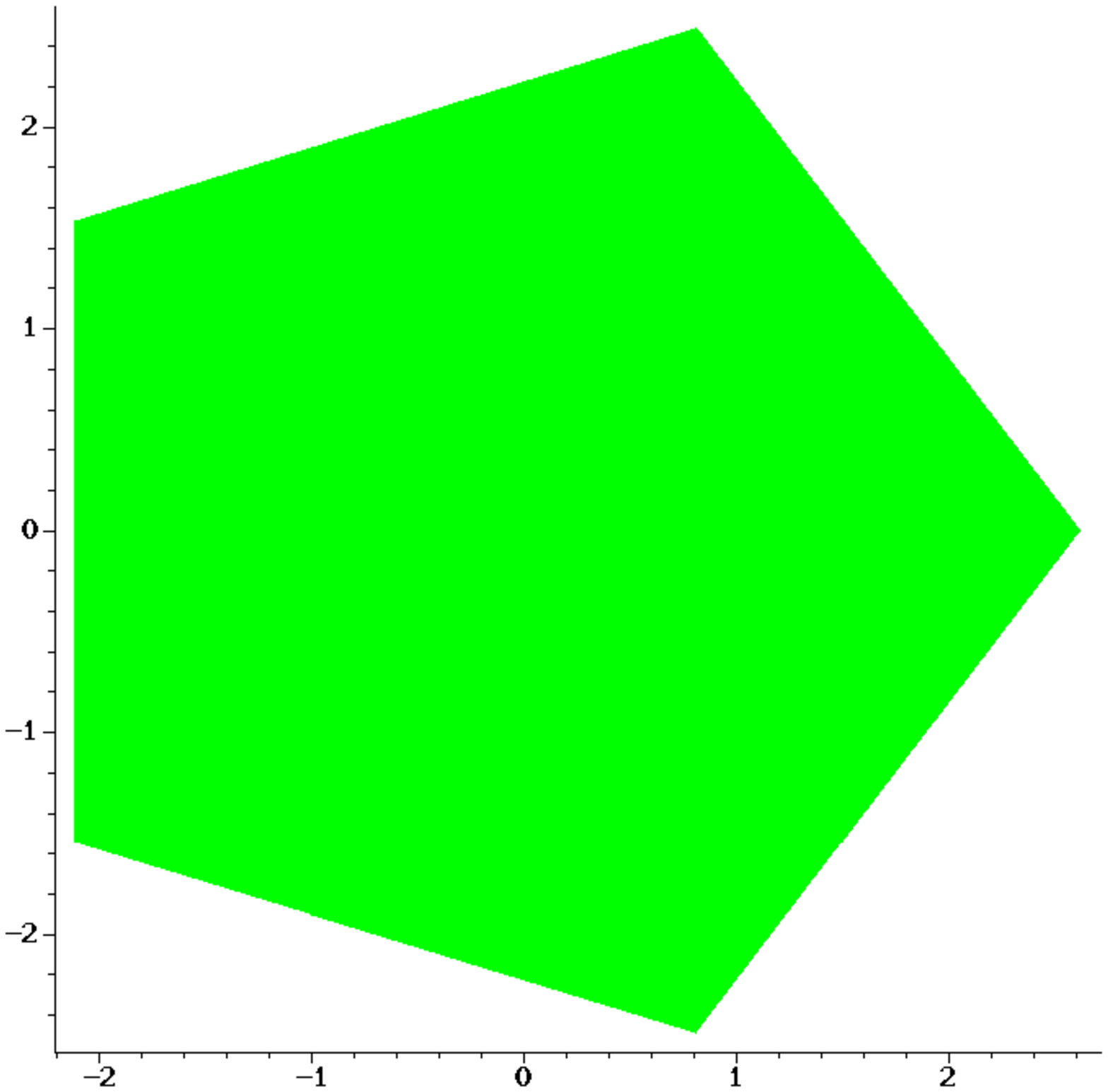}}
\subfigure[$\tau$ with $n=10$]{\includegraphics[scale=0.17,angle=270]{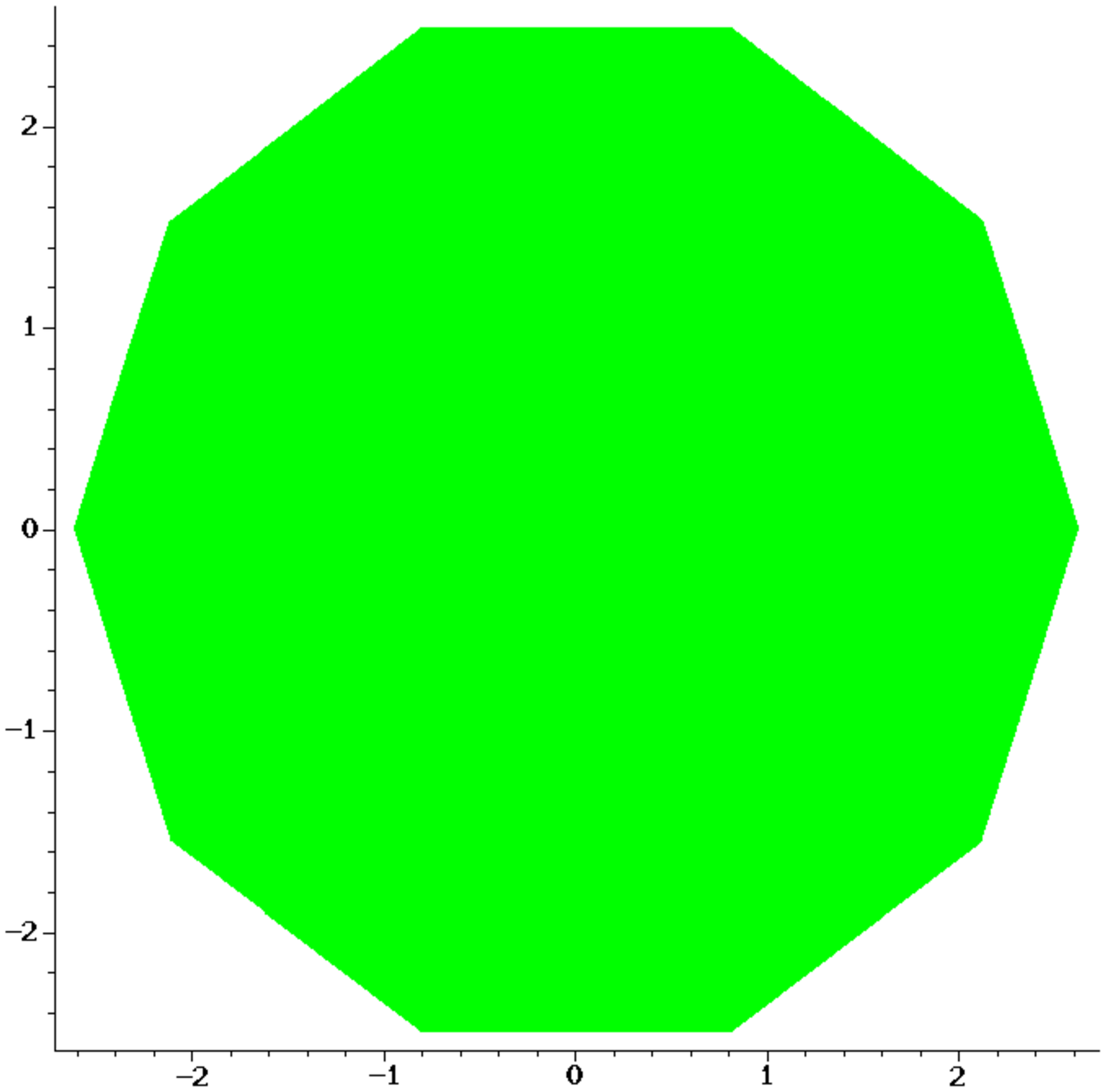}}
\subfigure[$\tau^2$ with $n=10$]{\includegraphics[scale=0.17,angle=270]{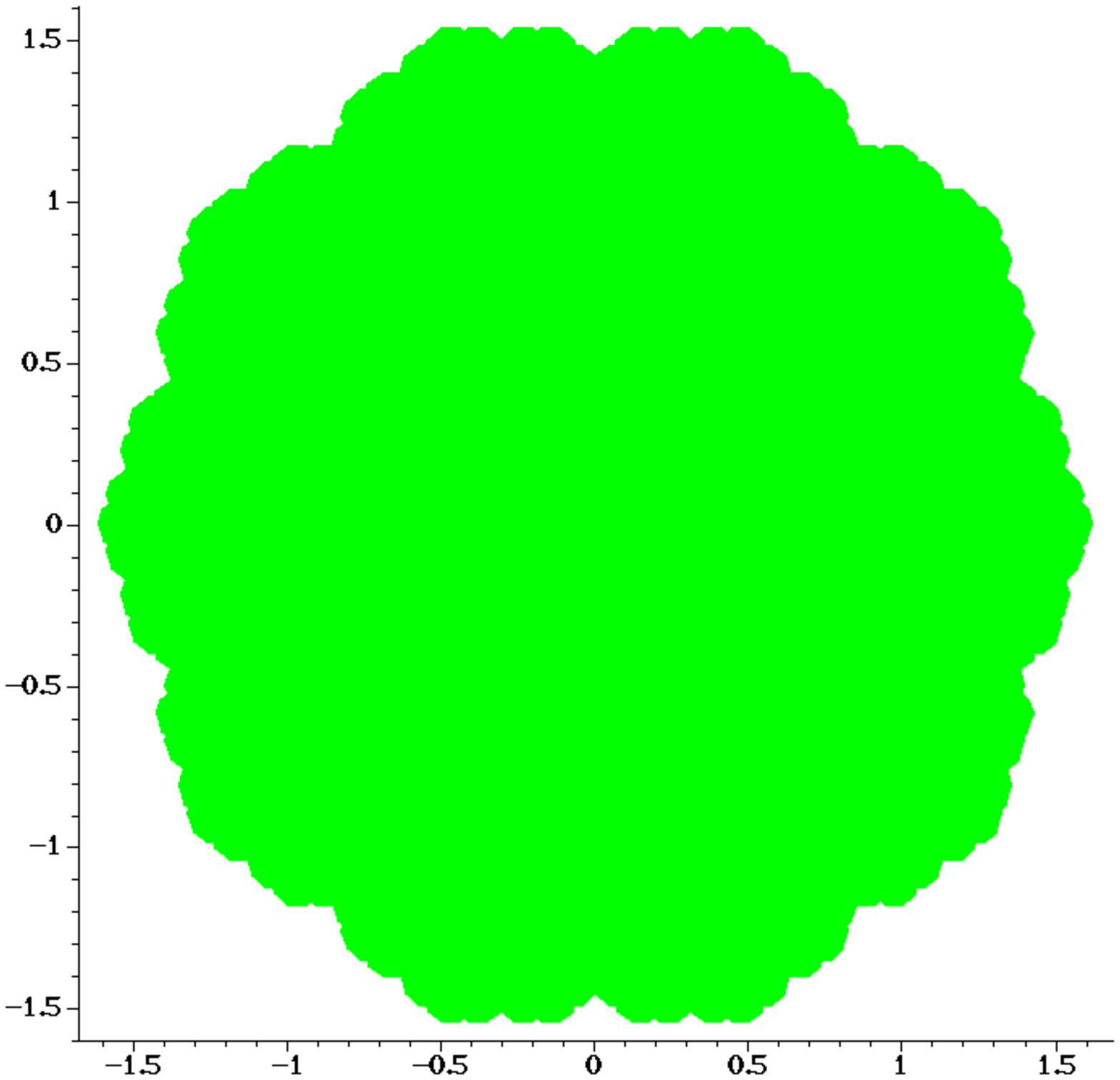}}
\\
\subfigure[$\lambda$ with $n = 7$]{\includegraphics[scale=0.17,angle=270]{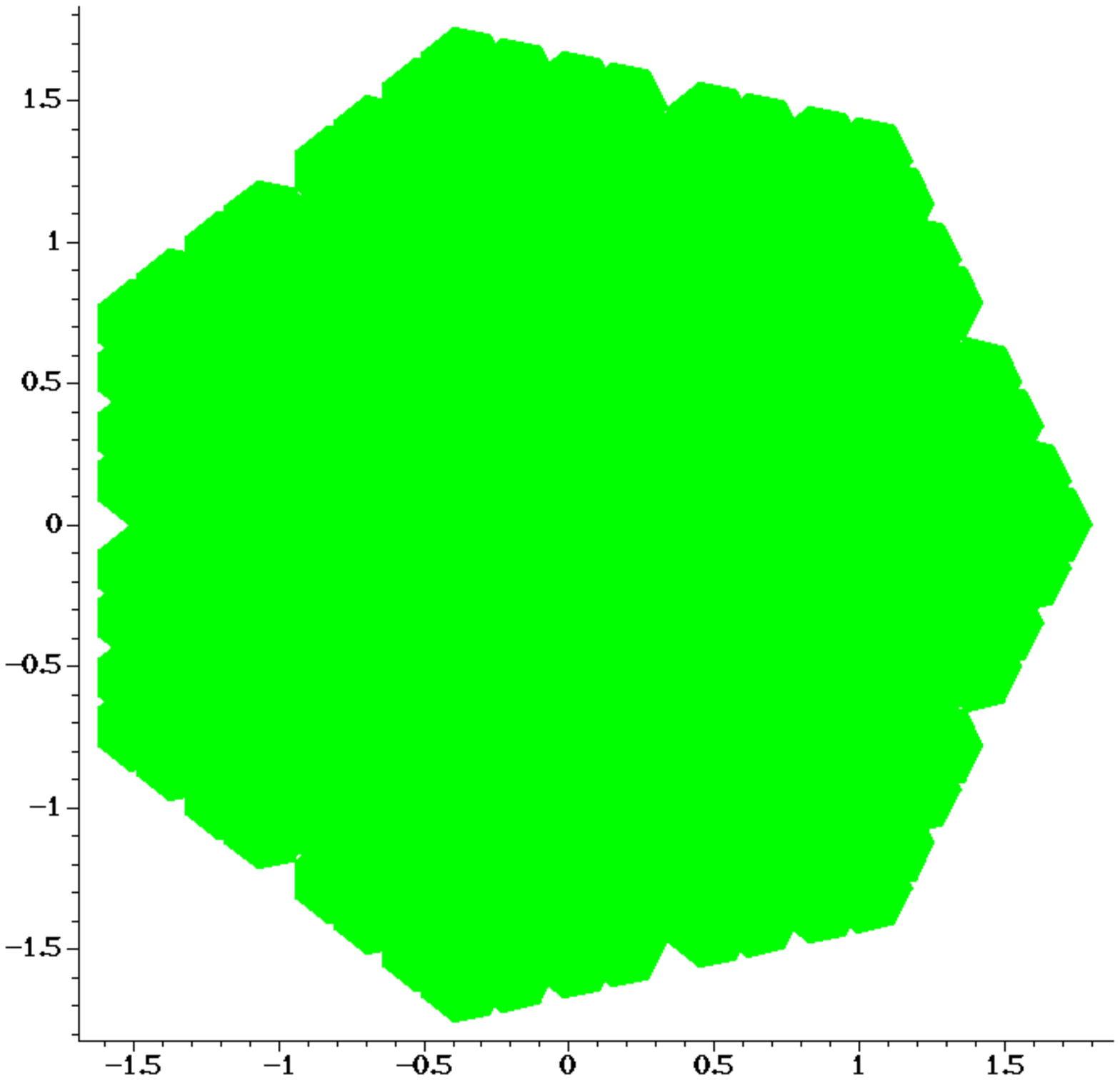}}
\subfigure[$\lambda$ with $n = 14$]{\includegraphics[scale=0.17,angle=270]{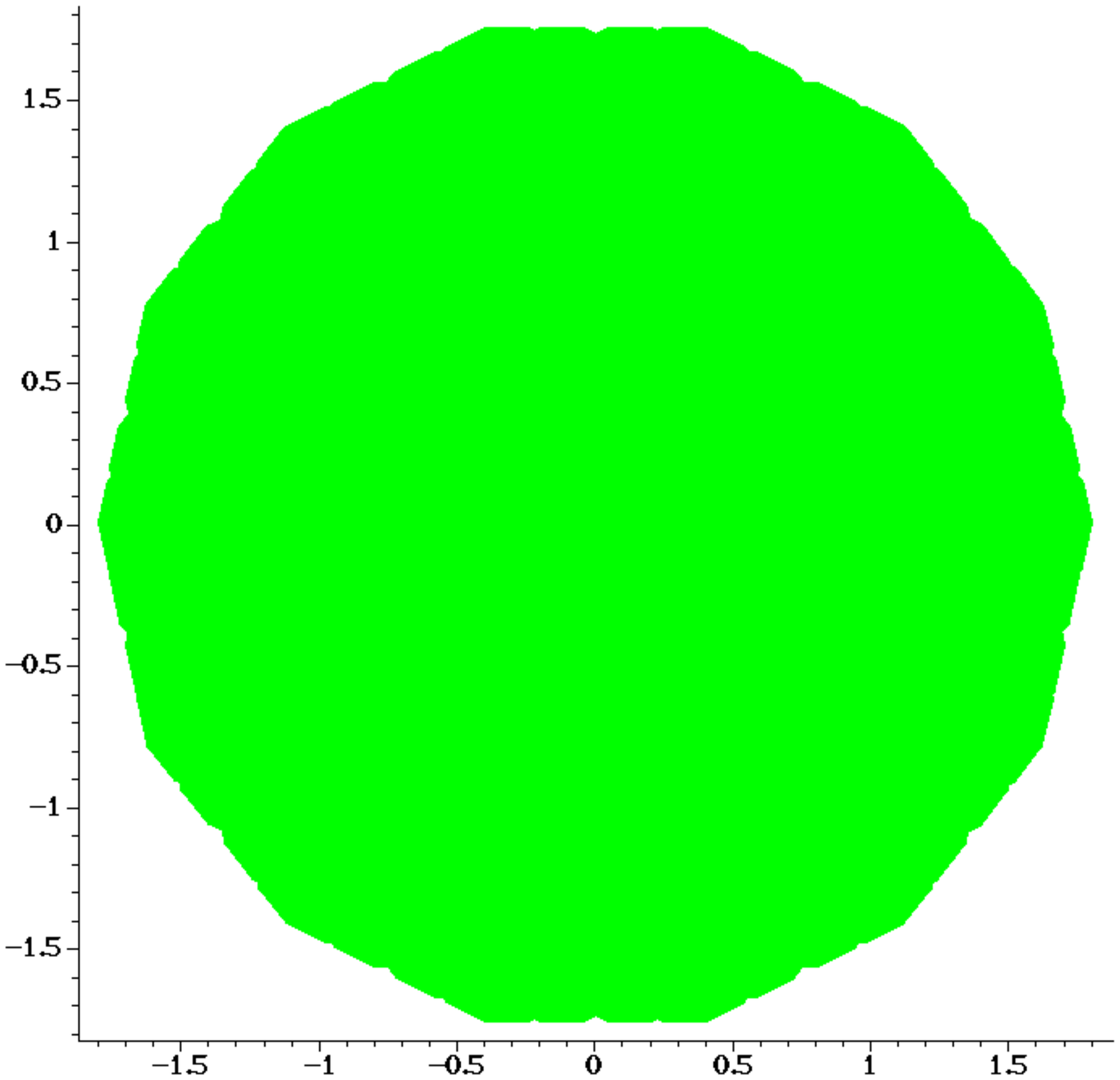}}
\subfigure[$\delta$ with $n = 8$]{\includegraphics[scale=0.17,angle=270]{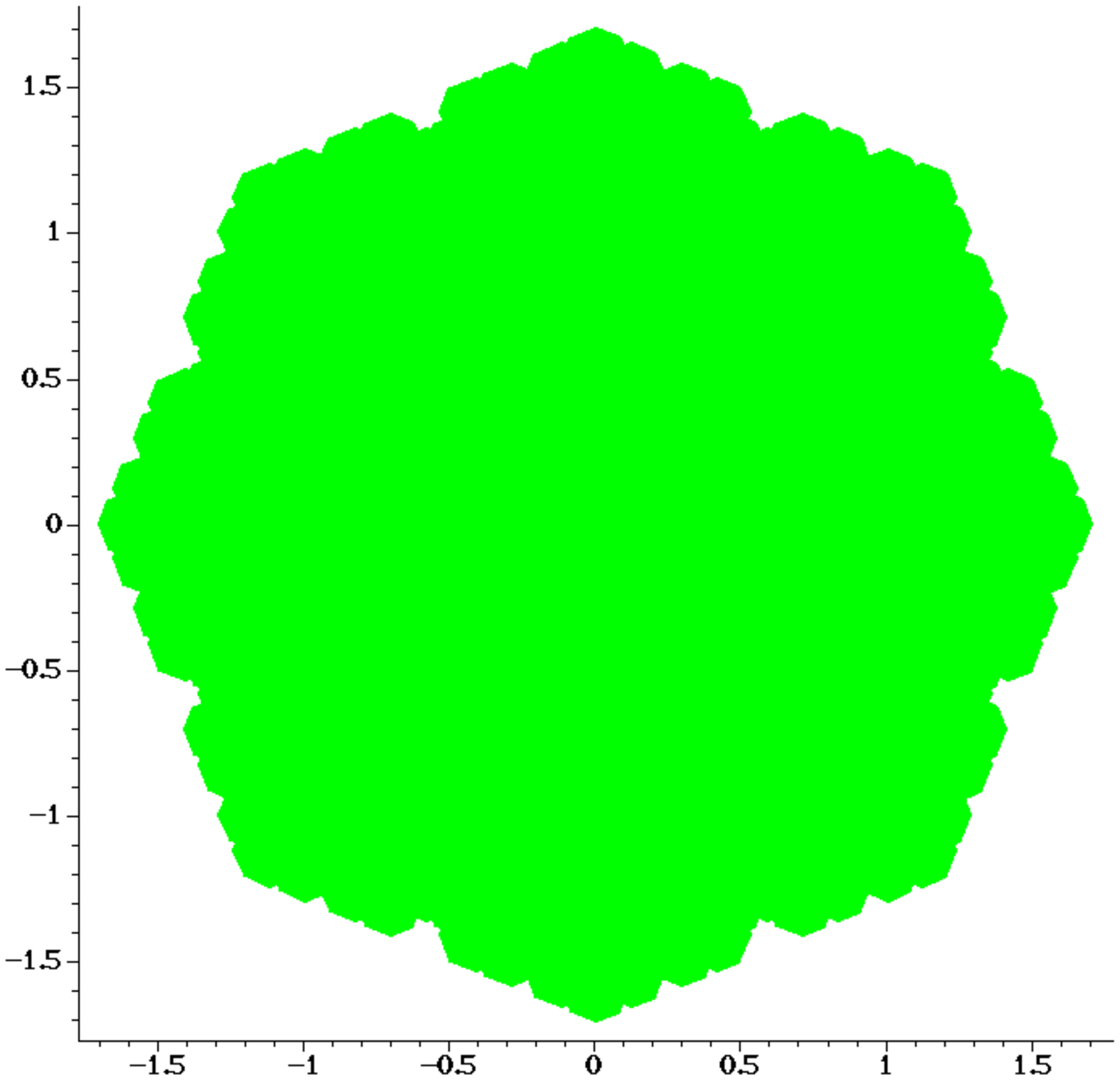}}
\\
\subfigure[$\kappa$ with $n = 9$]{\includegraphics[scale=0.17,angle=270]{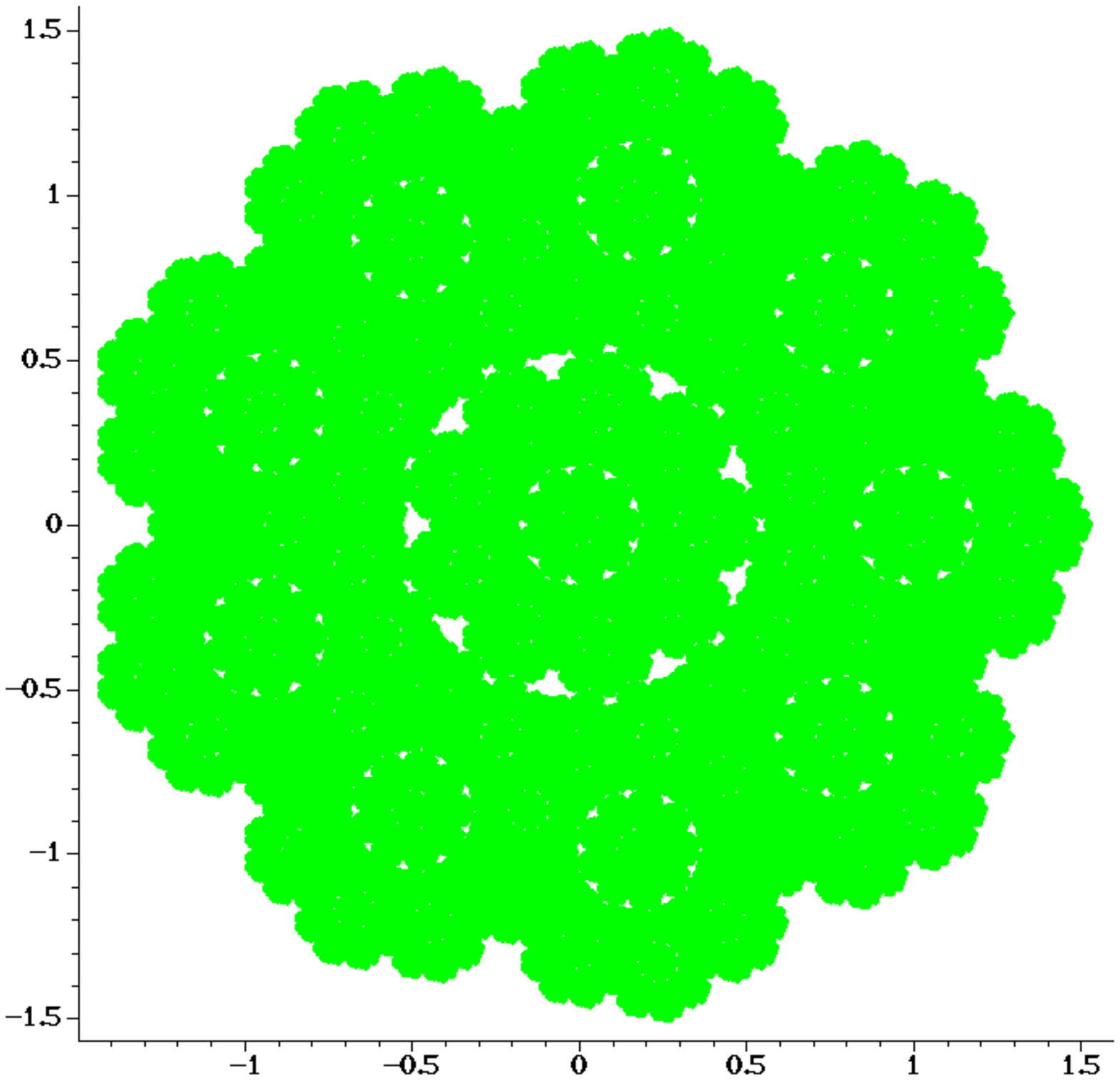}}
\subfigure[$\kappa$ with $n = 18$]{\includegraphics[scale=0.17,angle=270]{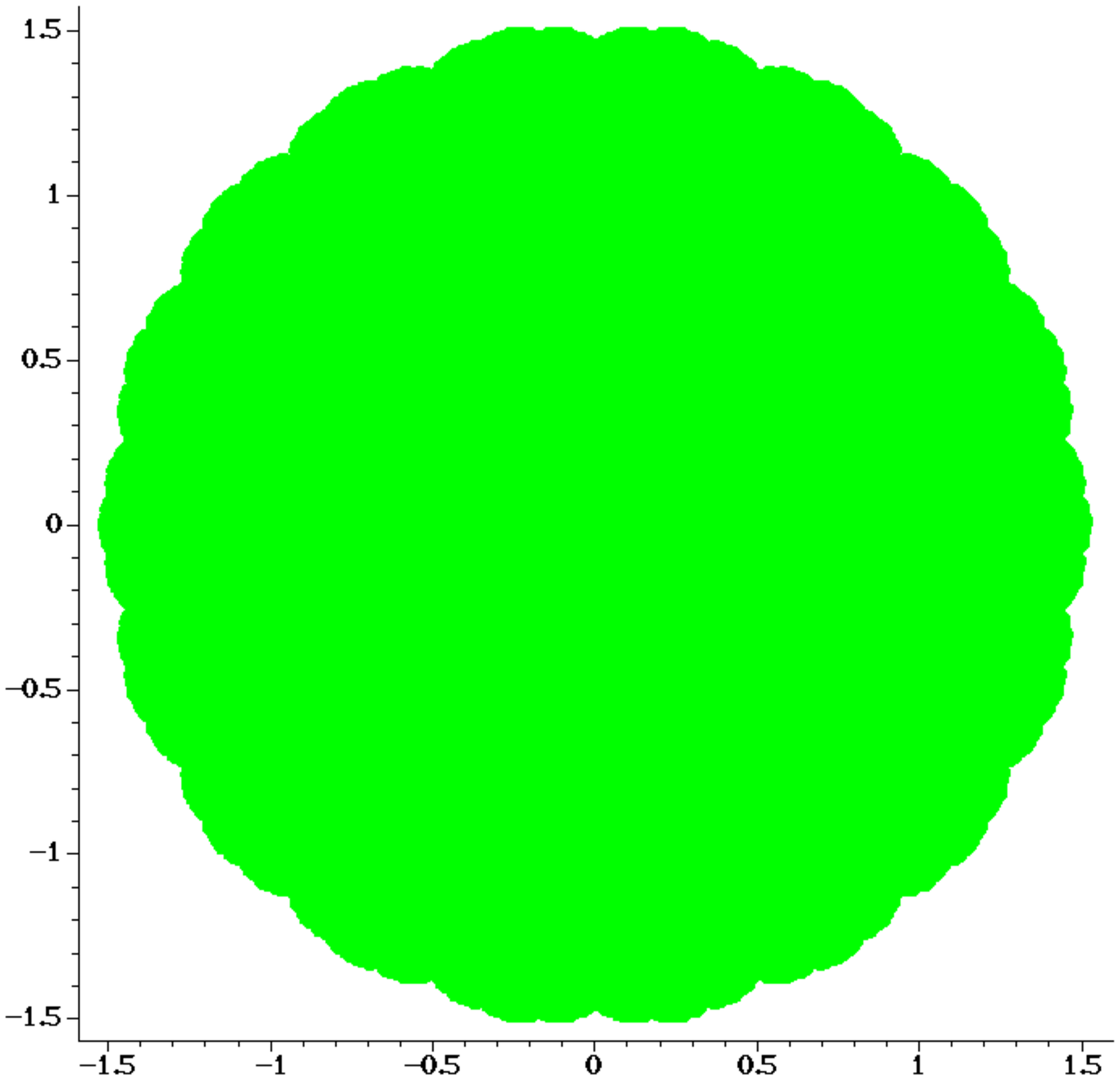}}
\subfigure[$\mu$ with $n = 12$]{\includegraphics[scale=0.17,angle=270]{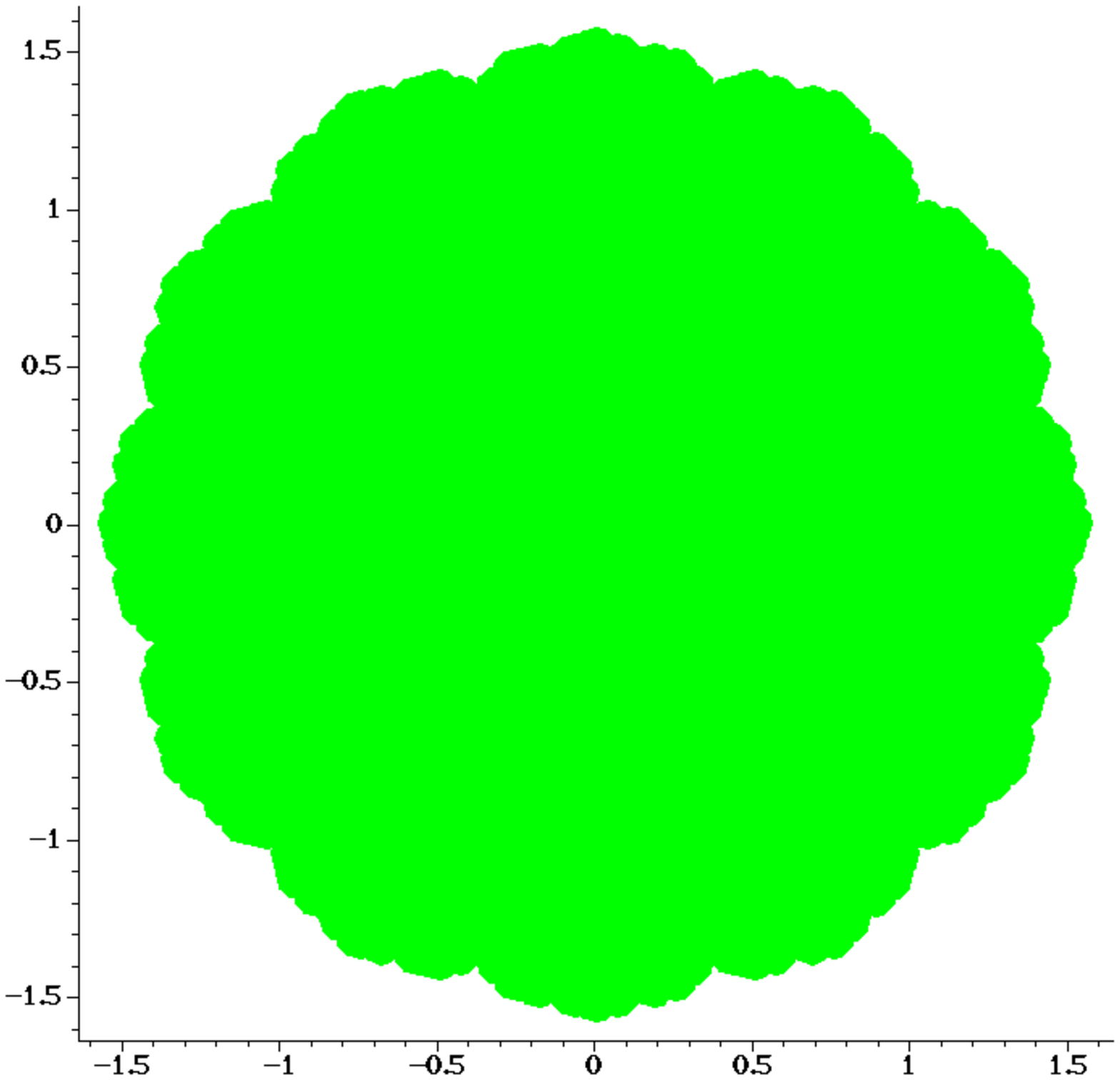}}
\caption{The attractor of the IFS with contraction $\beta^{-1}$ and offsets $\Delta_n\cup\{0\}$ for the cases listed in Table~\ref{tab:rd}.}
\label{fig:rd}
\end{figure}

\section{Spectra and cut-and-project sets}\label{sec:cap}

Before starting our discussion on the connection of the spectra to cut-and-project sets,
let us first review some of the basic details concerning the cut-and-project method.
Usually, it is defined by projection of a section of a higher dimensional lattice
to a suitably chosen subspace, see~\cite{Moody}. The cut-and-project scheme is given
by two subspaces $V_1$, $V_2$ of a $(d+e)$-dimensional euclidean space $V$, such that $\dim V_1=d$, $\dim V_2=e$,
$V_1\oplus V_2=V$. Denote by $\pi_1$ (resp. $\pi_2$) the projection of $V$
to $V_1$ along $V_{2}$ (resp. $V_2$ along $V_1$). Given a lattice $L\subset V$, we require that
$\pi_1$ restricted to $L$ is injective, and that $\pi_2(L)$ is dense in $V_2$.
Then, one has a mapping $\star : M:=\pi_1(L) \to M^\star:=\pi_2(L)$.
If $\Omega$ is a bounded subset of $V_2$ with non-empty interior, then we define
$\Sigma(\Omega)=\{x\in M : x^\star \in \Omega\}$.

In special cases leading to non-crystallographic symmetries,
this definition can be recast in an algebraic way, in particular, with the use
of Pisot-cyclotomic numbers~\cite{Icosians}. We present here the definition of a cut-and-project set in this
restricted frame, suitable for our purposes.

Let $\beta$ be a Pisot-cyclotomic number of order $n$, and let $\omega$ be an $n$-th root of unity.
If $\beta$ is quadratic, there exists an automorphism $\sigma$ of
    $\Z[\omega]$ that takes $\beta$ to its conjugate $\beta'$.
For $\beta$ cubic, there exist two different automorphisms $\sigma_1$
    and $\sigma_2$ of $\Z[\omega]$ that take $\beta$ to its
    two conjugates $\beta_1$ and $\beta_2$, respectively.

It is worth observing that as $\Z[\beta]=\Z[2\cos\frac{2\pi}{n}]$ and $\cos \frac{2\pi}{n}$ is totally real,
then so is $\beta$. It can be easily shown that
$\Z[\beta]+\Z[\beta]\omega = \Z[\omega]$. In particular, $\Z[\omega]$ is equal to the ring of integers in the cyclotomic field $\Q(\omega)$.
Set $M = \Z[\omega]$.
 It holds that $\omega M= M$ and that $M$ is closed under Galois conjugation.
 Given a bounded set $\Omega$ ($\Omega\subset\C$ for quadratic $\beta$, or  $\Omega \subset \C \times \C$ for cubic $\beta$)
 with non-empty interior, we define the cut-and-project set $\Sigma(\Omega)$ by
$$
\Sigma(\Omega) = \{z\in M : \sigma(z) \in\Omega\}\subset\C\,.
$$
for quadratic $\beta$ and
$$
\Sigma(\Omega) = \{z\in M : (\sigma_1(z), \sigma_2(z))  \in\Omega\}\subset\C\,.
$$
for cubic $\beta$. The set $\Omega$ is called the acceptance window.
As the sets $\Sigma(\Omega)$ fall within the frame of model sets studied in~\cite{Moody}, one can derive that such
cut-and-project are uniformly discrete relatively dense sets of finite local complexity.


Sometimes, one considers cut-and-project sets where the acceptance window $\Omega$ satisfies some additional properties.
In our frame, taking $\Omega$ symmetric under rotation of order $n$, i.e.
$\omega^k\Omega=\Omega$, we also have $\omega^k\Sigma(\Omega)=\Sigma(\Omega)$, which means that $\Sigma(\Omega)$ has $n$-fold symmetry.

It is useful to consider also restriction of the above defined cut-and-project sets to one dimension,
in this case presented as the cut-and-project sequence
$$
\Sigma(I)=\{x\in \Z[\beta] : \sigma(x)\in I\}\subset\R\,,
$$
$\sigma(x)$ amounts to the Galois conjugation in the field $\Q(\beta)$. The set $I$ is usually taken to be a bounded non-degenerated interval.
Such sequences have been well described, in particular, they are easy to generate, and allow to generate more-dimensional cases,
as we shall use further in this section.

Let us study whether the spectrum $X^{\A}(\beta)$ can be related to a cut-and-project set in the above setting. Obviously, $X^{\A}(\beta)$ a subset
of the set of cyclotomic integers $M$. It is natural to ask whether there exists a reasonable $\Omega$ such that $X^{\A}(\beta)= \Sigma(\Omega)$.
We consider the Galois conjugation of elements of $X^\A(\beta)$.

\begin{thm}
\label{thm:conj=ifs}
Let $\beta$ be a Pisot cyclotomic number of order $n$ and let $\A=\Delta_n\cup\{0\}$. Suppose that $\sigma$ is a
    Galois conjugation on $\Q(\omega)$ such that $|\sigma(\beta)| < 1$.
Then
    \[K(\sigma(\beta), \A) = \mathrm{cl}(\sigma(X^{\A}(\beta))).\]
Moreover, if $K(\sigma(\beta),\A)$ contains $0$ in its interior, then
\[\sigma(X^{\A}(\beta)) \subseteq {\rm int }\big(K(\sigma(\beta), \A)\big) .\]
\end{thm}

\begin{proof}
Since $\sigma(\A)=\A$, elements of $\sigma(X^{\A}(\beta))$ are of the form $\sum_{i=0}^{n}a_i\big(\sigma(\beta)\big)^i$, for $n\in\N$, $a_i\in\A$.
Setting $0=a_{n+1}=a_{n+2}=\cdots$, we have $\sum_{i=0}^{+\infty}a_i\big(\sigma(\beta)\big)^i\in K\big(\sigma(\beta),\A\big)$.
This implies
$K(\sigma(\beta), \A) \supset \mathrm{cl}(\sigma(X^{\A}(\beta)))$.
For the opposite inclusion, take $z\in K(\sigma(\beta), \A)$, i.e. $z=\sum_{i=0}^{+\infty}b_i\big(\sigma(\beta)\big)^i$. Obviously $z$ is the limit of
$w_n=\sum_{i=0}^{n}b_i\big(\sigma(\beta)\big)^i$.

For the latter inclusion, it suffices to see the equivalence in Theorem~\ref{thm:ekvivalence}.
\end{proof}

\begin{coro}
Since $K(\sigma(\beta), \A)$ is the attractor of the corresponding IFS, one can derive some information about the shape  $\sigma(X^{\A}(\beta))$.
In particular
\begin{itemize}
\item If $1/2 < \sigma(\beta) < 1$ then
    $\mathrm{cl}(\sigma(X^{\A}(\beta)))$ is a polygon.
\item If $-1 < \sigma(\beta) < -1/2$ and $n$ is even then
    $\mathrm{cl}(\sigma(X^{\A}(\beta)))$ is a polygon.
\item If $-1 < \sigma(\beta) < -1/\sqrt{2}$ and $n$ is odd then
    $\mathrm{cl}(\sigma(X^{\A}(\beta)))$ is a polygon.
\end{itemize}
\end{coro}


The shapes of $K(\sigma(\beta), \A)$ for all the quadratic resp.\ cubic cases are displayed in Figures~\ref{fig:missingquadr} and~\ref{fig:missingcub}.
 Note that if $\beta$ is a quadratic unit, we have $\sigma(\beta)=\pm\beta^{-1}$. Thus $K(\sigma(\beta),\A_n)=K(1/\beta,\A_n)$ when either $\sigma(\beta) = 1/\beta$ and/or $n$ is even (which means $\A_n=-\A_n$). In fact, the cases where this is not true are $n=5$, $\beta=\tau$ and $n=12$, $\beta=\mu$ (compare Figures~\ref{fig:rd} and~\ref{fig:missingquadr}).

From Theorem~\ref{thm:conj=ifs}, one can derive a suitable choice of the acceptance window for comparing the spectra with cut-and-project sets.

\begin{remark}\label{rem:cap}
Let $\beta$ be a quadratic or a cubic Pisot-cyclotomic number of order $n$ and let $\A=\Delta_n\cup\{0\}$.
Then  $X^{\A}(\beta) \subseteq \Sigma(\Omega)$, where
$$
\Omega =
\begin{cases}
K(\sigma(\beta), \A) & \text{for quadratic $\beta$,}\\
K(\sigma_1(\beta), \A) \times K(\sigma_2(\beta), \A) & \text{for cubic $\beta$}.
\end{cases}
$$
In the quadratic case, we denote $\sigma$ the Galois conjugation on $\Q(\omega)$ such that $|\sigma(\beta)| < 1$.
In the cubic case, we have similarly $|\sigma_1(\beta)| < 1$, $|\sigma_2(\beta)| < 1$.
\end{remark}


\begin{figure}
\subfigure[$\lambda$ with $n = 7$ for $\sigma_1$]{\includegraphics[scale=0.13]{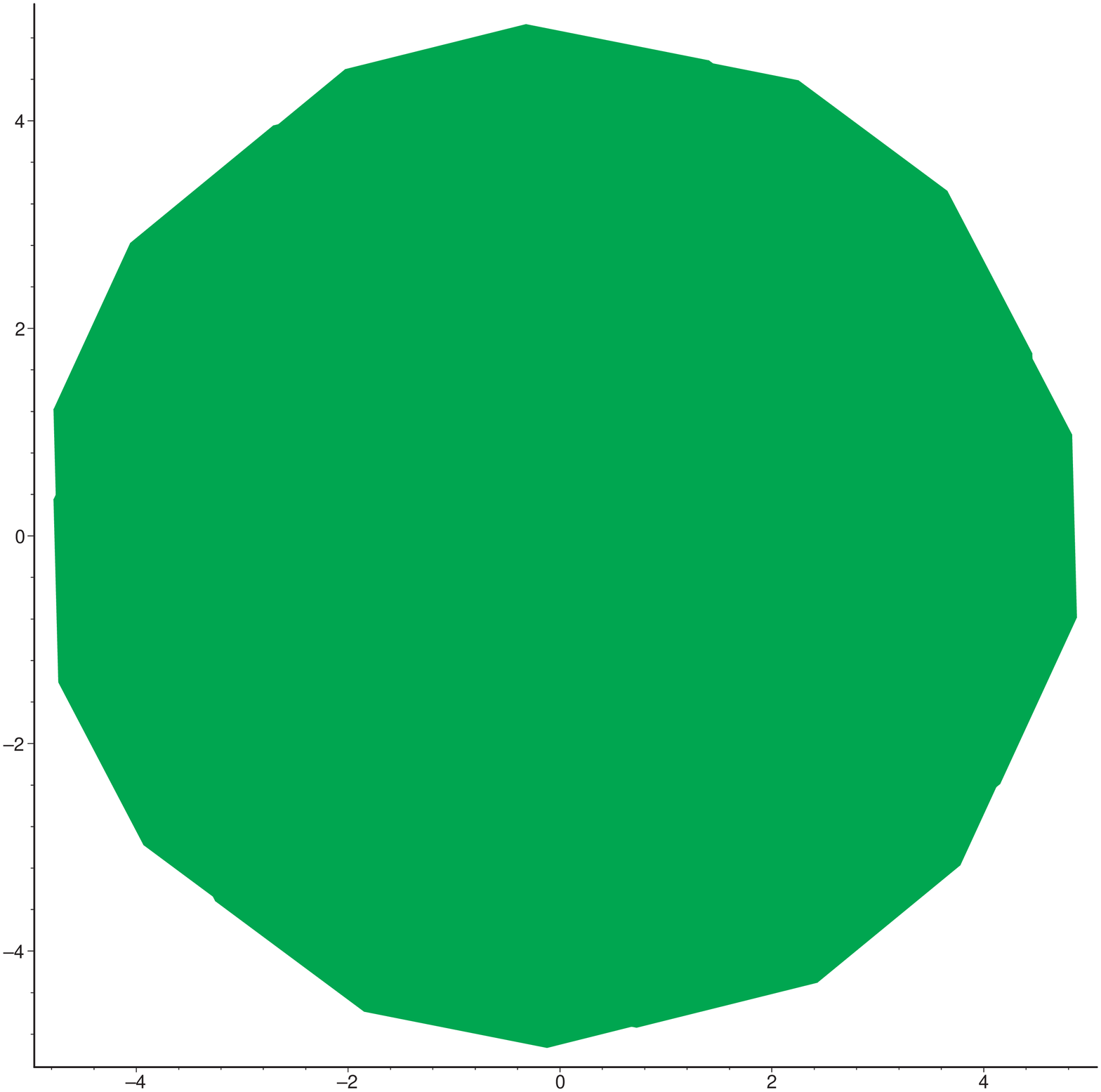}}\qquad\qquad
\subfigure[$\lambda$ with $n = 14$ for $\sigma_1$]{\includegraphics[scale=0.13]{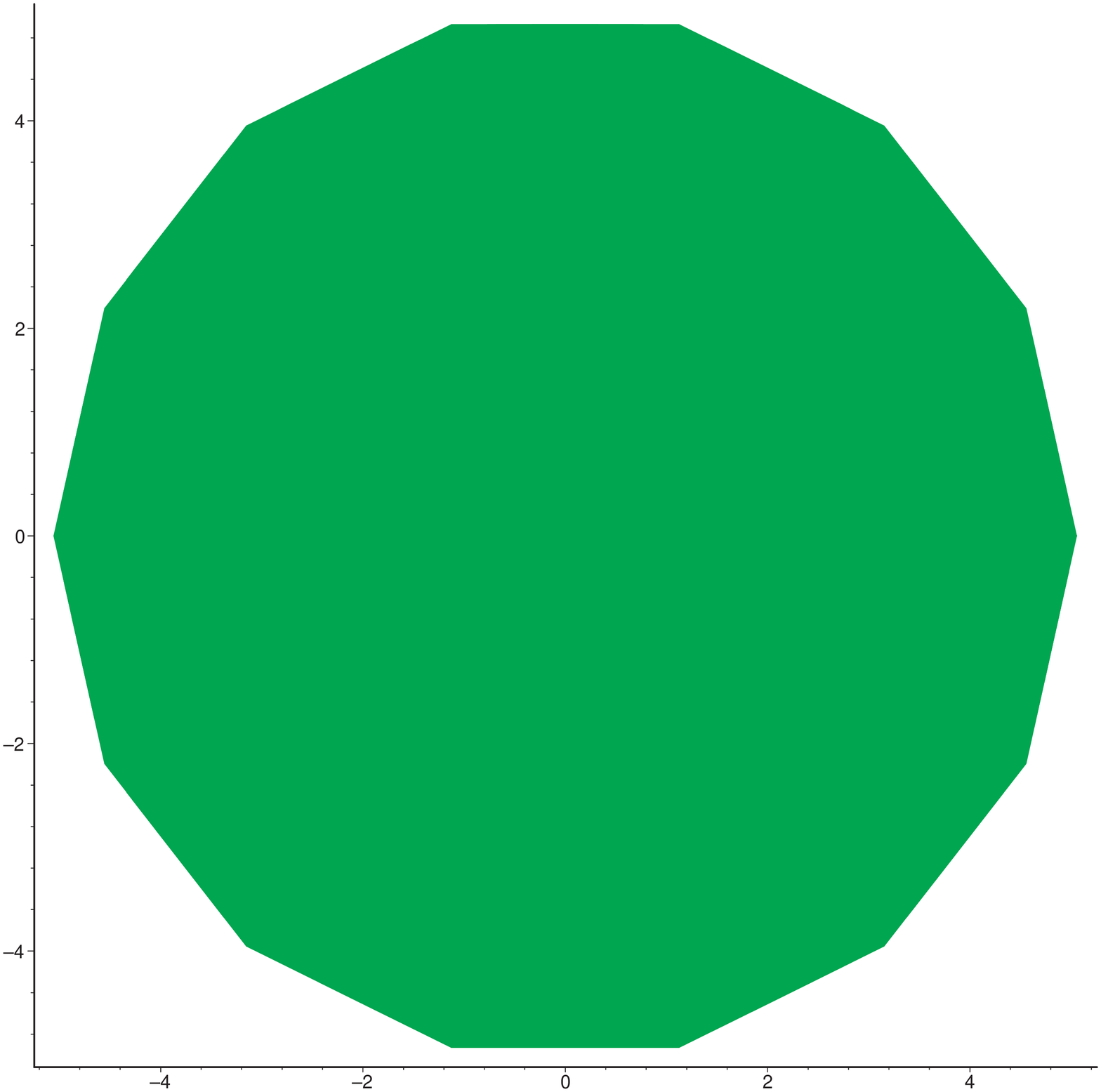}}\qquad\qquad
\subfigure[$\kappa$ with $n = 18$ for $\sigma_1$]{\includegraphics[scale=0.13]{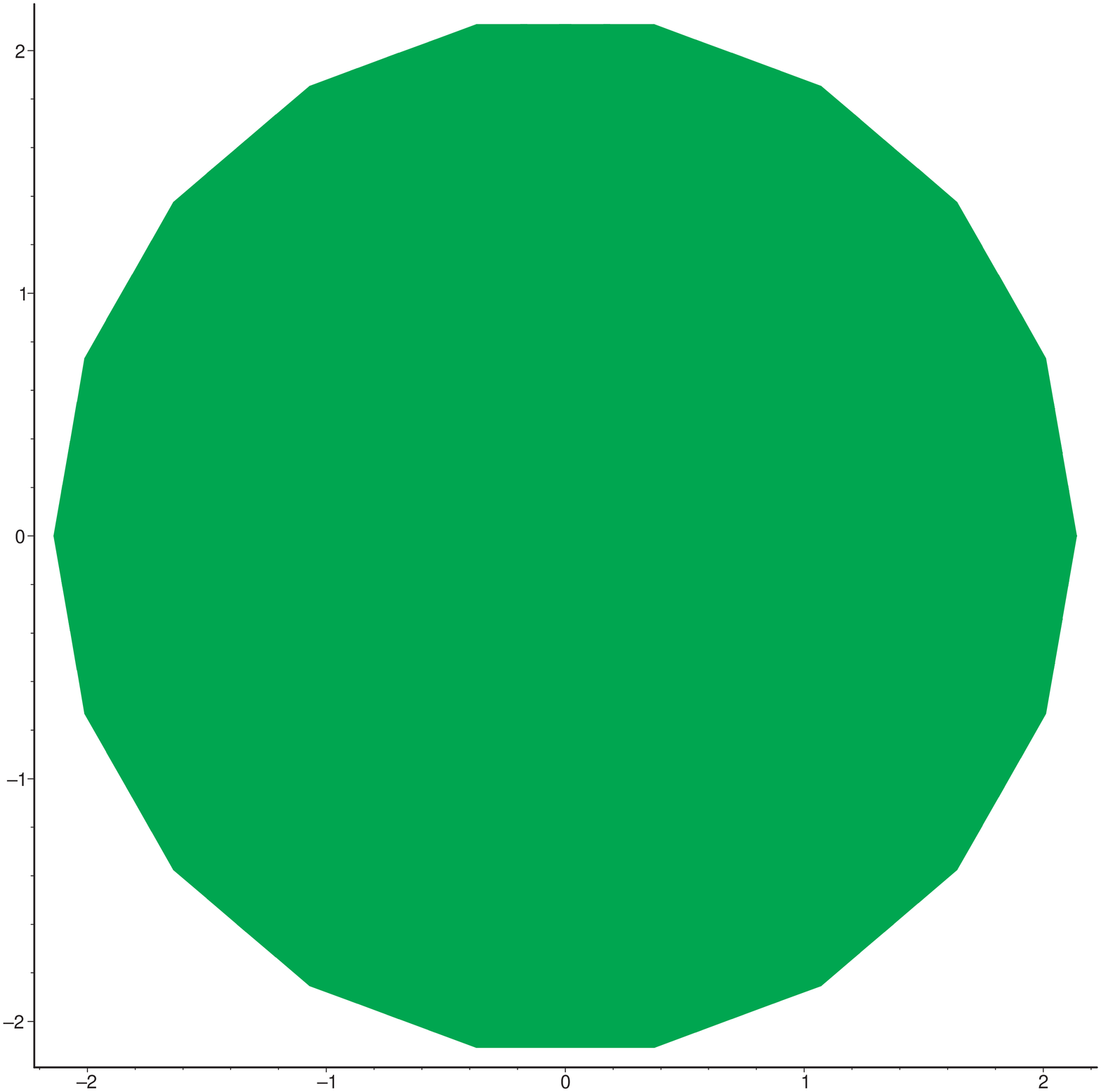}}
\\
\subfigure[$\lambda$ with $n = 7$ for $\sigma_2$]{\includegraphics[scale=0.13]{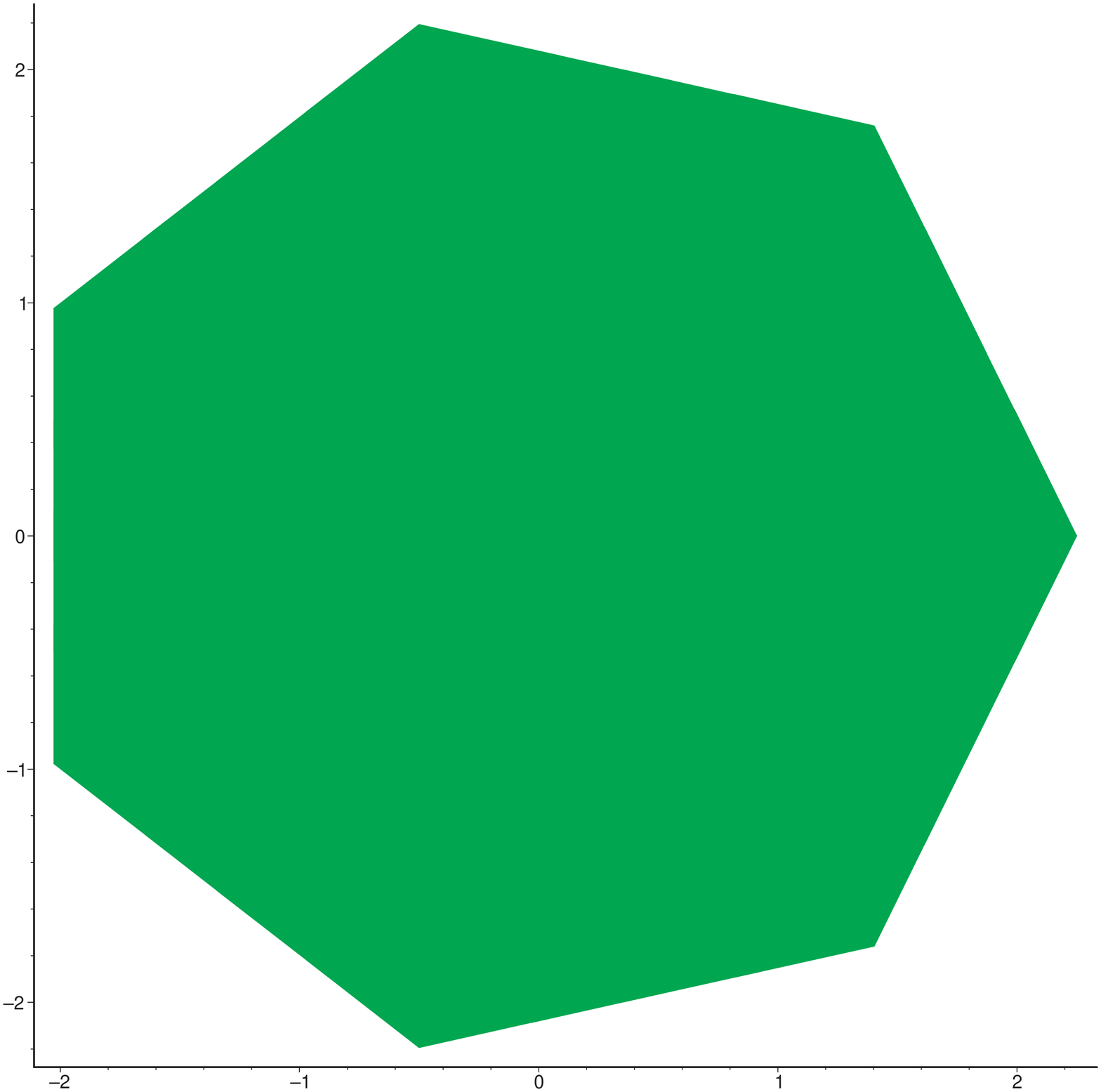}}\qquad\qquad
\subfigure[$\lambda$ with $n = 14$ for $\sigma_2$]{\includegraphics[scale=0.13]{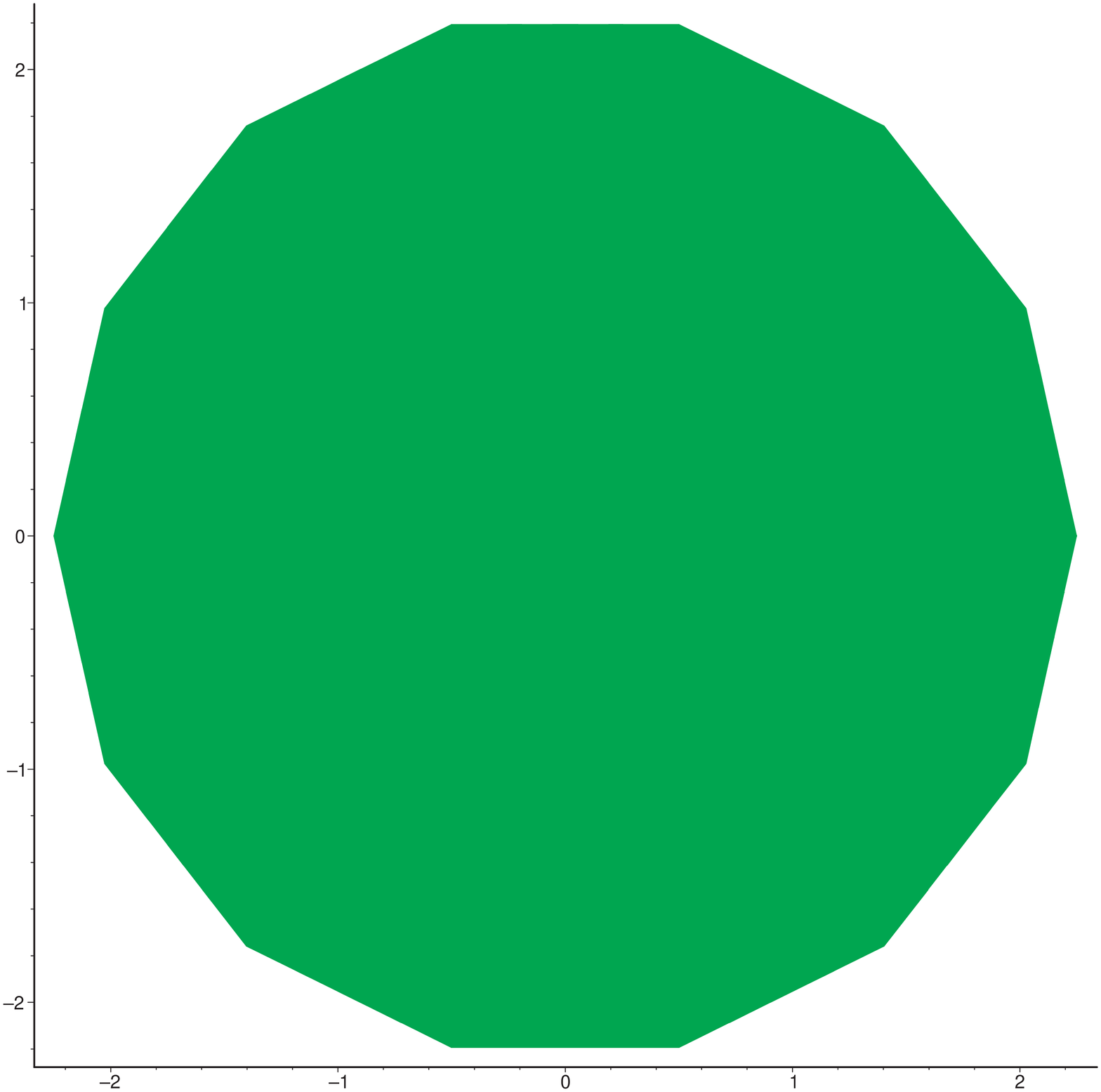}}\qquad\qquad
\subfigure[$\kappa$ with $n = 18$ for $\sigma_2$]{\includegraphics[scale=0.13]{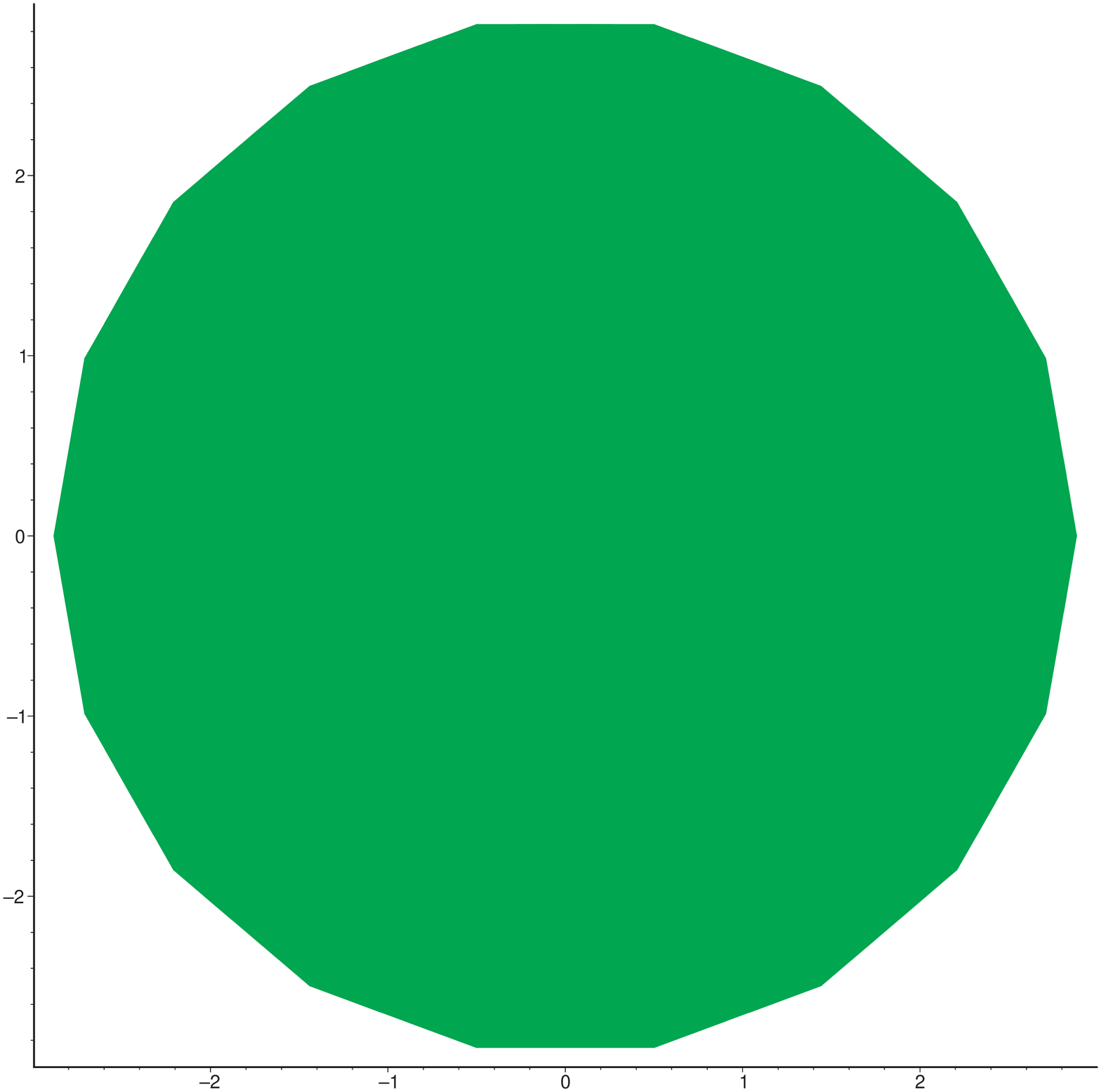}}
\caption{The attractors $K(\sigma(\beta), \A)$ for the spectra of cubic
Pisot-cyclotomic numbers. For each spectrum, two attractors apply, since there are two non-trivial automorphisms to consider.}
\label{fig:missingcub}
\end{figure}

In general, it is not easy to decide, whether all points of the cut-and-project set mentioned in Remark~\ref{rem:cap} are included in the spectrum,
in other words, whether equality holds in $X^{\A}(\beta) \subseteq \Sigma(\Omega)$. One way to verify this is to consider a criterion, which reduces the question to compare these sets only in a close neighbourhood of the origin. This task can be performed in few steps, taking into account that we have an algorithm to generate the spectra and cut-and-project sets.
In order to decide whether there are some missing points (i.e. points in $\Sigma(\Omega)\setminus X^\A(\beta)$)
we start with the following lemma.

\begin{lem}\label{l:podminkaS}
Let $\beta>1$ be a unit and let $S\subset K(\sigma(\beta),\A)$ satisfy
\begin{equation}\label{eq:podminkaS}
S\subseteq\bigcup_{a\in\A}(\sigma(\beta) S+a).
\end{equation}
Then $(\Sigma(S)\cap \mathrm{cl}(B_{1/(\beta-1)}(0)))\subseteq (X^\A(\beta)\cap \mathrm{cl}(B_{1/(\beta-1)}(0)))$ implies $\Sigma(S)\subseteq X^\A(\beta)$.
\end{lem}

\begin{proof}
Let $x\in\Sigma(S)$, i.e. $\sigma(x)\in S.$ By the assumption there exists $a_0\in\A$ such that $\tfrac{\sigma(x)-a_0}{\sigma(\beta)}\in S.$
Thus $y:=\frac{x-\sigma^{-1}(a_0)}{\beta}\in\Sigma(S).$ Note that the assumption that $\beta$ is a unit is used here to ensure that $y\in\Z[\omega].$
If $y\in X^\A(\beta)$, then also $x=\beta y+\sigma^{-1}(a_0)\in X^\A(\beta)$. We have that $|y|<|x|$ unless $x\leq 1/(\beta-1)$. Thus, by
iterating this procedure, we eventually obtain a point in $\mathrm{cl}(B_{1/(\beta-1)}(0))$.
\end{proof}

Let us explain the consequences of this lemma in the quadratic case. The cubic case can be treated accordingly.
Setting $S=K(\sigma(\beta),\A)$, the condition~\eqref{eq:podminkaS} is satisfied, even with equality. This has the following consequence.

\begin{coro}
\label{thm:spectra=cnp}
Let $\Omega = K(\sigma(\beta),\A)$. If $\beta$ is an algebraic unit, then
the following are equivalent.
\begin{enumerate}
\item $\Sigma(\Omega) \cap \mathrm{cl}(B_{1/(\beta-1)}(0)) =
       X^{\A}(\beta) \cap \mathrm{cl}(B_{1/(\beta-1)}(0))$
\item $\Sigma(\Omega) = X^{\A}(\beta)$.
\end{enumerate}
\end{coro}

\pfz
The inclusion $\Sigma(\Omega) \supseteq X^{\A}(\beta)$ is satisfied by Remark~\ref{rem:cap}, for the other inclusion, use Lemma~\ref{l:podminkaS}.
\pfk

In the cases $(\tau,10)$, $(\tau^2,10)$, and $(\delta,8)$ we can apply Lemma~\ref{l:podminkaS} also to $S=\mathrm{int}(K(\sigma(\beta),\A))$ to be able to compare the sets $X^\A(\beta)$ and $\Sigma(S)$ in those cases.
We only need to show that such $S$ satisfies  condition~\eqref{eq:podminkaS}. This is easy to prove in case $(\tau,10)$ as $S$ is the open decagon. In cases $(\tau^2,10)$ and $(\delta,8)$, things are more difficult, and we provide only a sketch of the proof for $(\tau^2,10)$.

Let $\mathcal F(M)=\bigcup_{a\in\A_{10}}\sigma(\beta)M+a$ for any $M\subseteq \C$. 
Let $D_0$ be the closed decagon with vertices $\tfrac1{\cos(\theta/2)}R(\theta,\gamma)\exp(\tfrac{2\pi i}{10}k)$ with $R(\theta,\gamma)$ as in \eqref{rtheta}.
For $n\in\N$ we define $D_n = \mathcal F(D_{n-1})$. Then $D_0\subset D_1$, thus
$\lim_{n\rightarrow +\infty} D_n = K(\sigma(\beta),\A)$, where the limit is with respect to the Hausdorff metric in the space of compact sets.

One can show that $D_0\subset \mathrm{int}(D_3)$, thus we can find an open set $E_0$ such that $D_0\subset E_0\subset D_3$. Define $E_n = \mathcal F(E_{n-1})$ and $E = \bigcup_{n\in\N} E_n$. We see that $E$ is an open set in $K(\sigma(\beta),\A)$ whose closure is $K$. Moreover, it holds that $\mathcal F(E) = E$, i.e. $E$ satisfies condition~\eqref{eq:podminkaS}.

It is also possible, but rather technical to show that the Hausdorff distance of $\partial{K}$ and $\partial D_n$ tends to zero which implies that $E = \mathrm{int}(K)$. As a byproduct we also obtain that $K$ is simply connected as each $D_n$ is simply connected.

A similar argumentation can be used to show that $S=\mathrm{int}(K(\sigma(\delta),\A_8))$ satisfies~\eqref{eq:podminkaS}.
\begin{coro}
\label{thm:spectra=cnpint}
Let $\Omega = \mathrm{int}(K(\sigma(\beta),\A))$ for $\A=\A_{10}$ and $\beta=\tau,\tau^2$, and for $\A=\A_8$ and $\beta=\delta$. Then
the following are equivalent.
\begin{enumerate}
\item $\Sigma(\Omega) \cap \mathrm{cl}(B_{1/(\beta-1)}(0)) \subseteq
       X^{\A}(\beta) \cap \mathrm{cl}(B_{1/(\beta-1)}(0))$
\item $\Sigma(\Omega) \subseteq X^{\A}(\beta)$.
\end{enumerate}
\end{coro}

Corollaries~\ref{thm:spectra=cnp} and~\ref{thm:spectra=cnpint} give us a technique to find all missing points. First step is to find all missing points $x$ such that
$|x|\leq \frac1{\beta-1}$. Next we check which of the points $\beta x+a$, $a\in\A$, are missing, and we repeat this procedure.

Let us explain the method of finding the missing points of small absolute value. We need to compare the spectrum $X^\A(\beta)$
and the cut-and-project set $\Sigma(\Omega)$ in the ball $B_{1/(\beta-1)}(0)$.
Finding all points of the spectrum within a bounded ball around the origin is straightforward using the techniques explained before.
Let us focus on the method of generating all points of $\Sigma(\Omega)$ within a bounded region. For that, observe the trivial property of
cut-and-project sets that
$$
\Omega\subseteq \tilde{\Omega} \quad\Rightarrow\quad \Sigma(\Omega)\subseteq \Sigma(\tilde{\Omega})\,.
$$
Thus, we choose a suitable $\tilde{\Omega}$ for which the computation of $\Sigma(\tilde{\Omega})$ is simple, and then we check which
of its elements $x$ belong also to $\Sigma(\Omega)$ by verifying the condition on its conjugates.
We set $\tilde{\Omega}$ to be a parallelogram along the roots of unity, so that the cut-and-project set is a Cartesian product
of two one-dimensional cut-and-project sequences $\Sigma(I)$.

Consider $\tilde{\Omega}=I+\sigma(\omega) I$,
where $I$ is a suitable interval. Then
$$
\Sigma(\Omega) \subset \Sigma(\tilde{\Omega}) =\Sigma\big( I+\sigma(\omega) I\big) = \Sigma(I) + \omega \Sigma(I)\,.
$$

By the above corollaries, we only need to consider points in $\Sigma(\tilde{\Omega})\cap \mathrm{cl}\big(B_{1/(\beta-1)}(0)\big)$, which amounts to finding points in
$\Sigma(I)\cap J$, where $J$ is an interval such that $J+\omega J\supset\mathrm{cl}\big( B_{1/(\beta-1)}(0)\big)$. Then
$$
\Sigma(\tilde{\Omega})\cap \mathrm{cl}\big(B_{1/(\beta-1)}(0)\big) \subset \big(\Sigma(I)\cap J\big)+\omega\big(\Sigma(I)\cap J\big)\,.
$$

In order to find all elements of $\Sigma(I)\cap J$, realize that these are of the form $x=a+b\beta$, $a,b\in\Z$, such that
$$
a+b\beta\in J \quad\text{and}\quad a+b\beta'\in I\,.
$$
This is a system of two linear inequalities for two unknowns $a,b$, which can be easily solved.

Let us find explicitly the suitable intervals $I$, $J$ for the considered cases.
We have that $\Omega$ is contained in $B_R(0)$ with $R=\sum_{i=0}^{+\infty}|\sigma(\beta)|^i=\tfrac{1}{1-|\sigma(\beta)|}$.
Furthermore, $B_R(0)$ is contained in the parallelogram $I+\sigma(\omega)I$ with $I = [\frac{-R}{\sin\theta'},\frac{R}{\sin\theta'}]$
where $\theta' = \arg(\sigma(\omega))$. Similarly, we can set $J = [-\tfrac{1}{\beta-1}\tfrac1{\sin\theta},\tfrac{1}{\beta-1}\tfrac1{\sin\theta}]$ where $\theta = \arg(\omega)$.
The appropriate choices of $\sigma(\omega)$ are shown in Table~\ref{tab:intervals} together with intervals $I$ and $J$. Figure~\ref{fig:missingquadr} illustrates the acceptance windows with the missing points marked.



\begin{table}
\renewcommand{\arraystretch}{1.3}
\begin{tabular}{llcccl}
\hline
order   & $\beta$    & $ \sigma(\omega)$&$s$            & $t$  & missing points      \\
\hline
5       & $\tau$   & $\exp(\frac{4\pi}5)=\omega^2$    & $4.45406$     & $1.70130$ & interior \\%
10      & $\tau$  & $\exp(\frac{6\pi}{10})=\omega^3$ & $2.75276$     & $2.75276$ & boundary \\
        & $\tau^2$  & $\exp(\frac{6\pi}{10})=\omega^3$ & $1.70130$     & $1.05146$ & boundary \\
8       & $\delta$ &  $\exp(\frac{6\pi}8)=\omega^3$   & $2.41421$     & $1$       & none \\
12      & $\mu$    &  $\exp(\frac{10\pi}{12})=\omega^5$& $7.46410$    & $1.15470$ & interior\\
\hline\\
\end{tabular}
\caption{Determination of intervals $I=[-s,s]$ and $J=[-t,t]$ for computation of the cut-and-project set $\Omega$.}
\label{tab:intervals}
\end{table}

\begin{figure}
\subfigure[$\tau$ with $n=5$]{\includegraphics[scale=0.17]{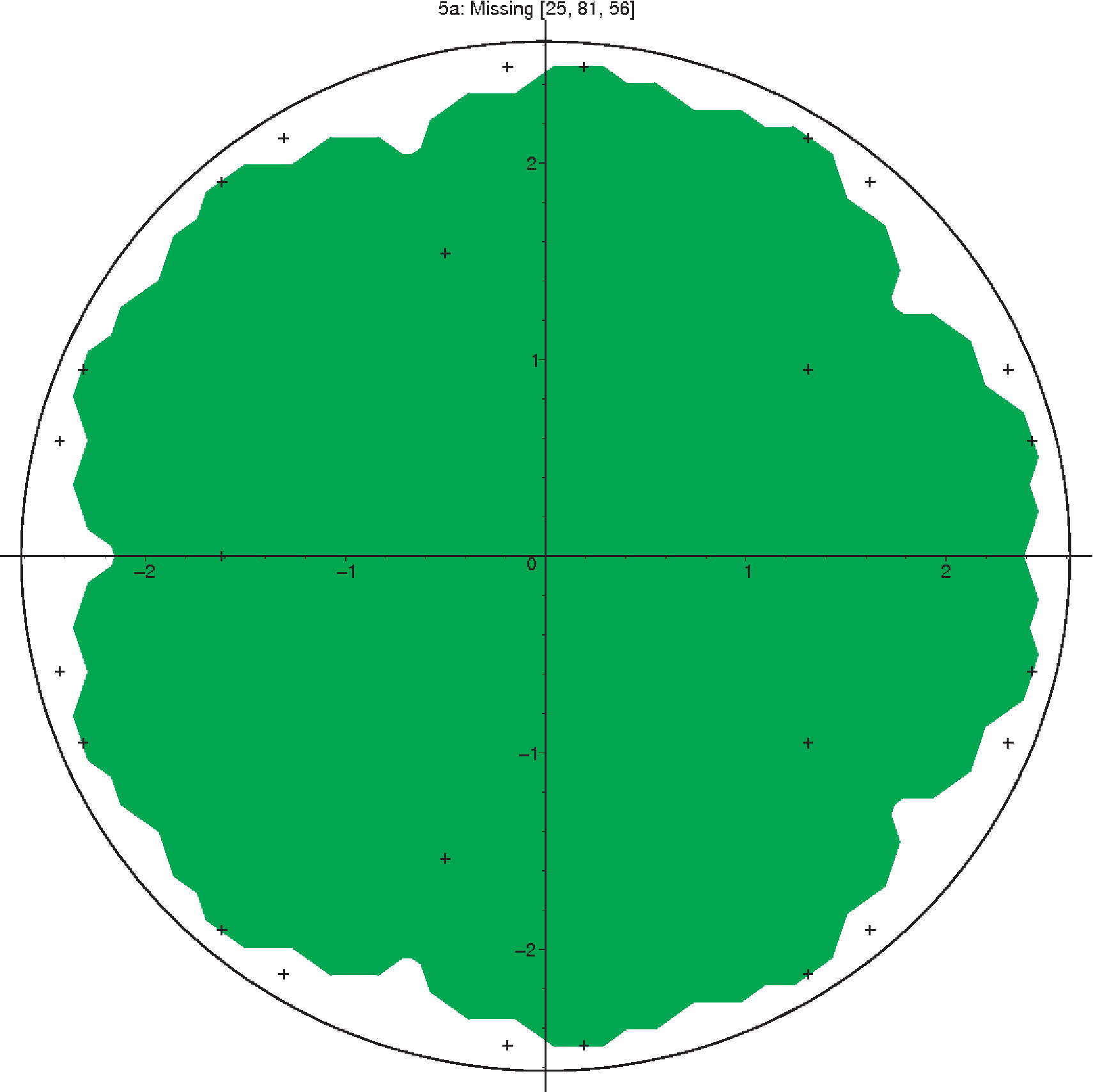}}
\subfigure[$\tau$ with $n=10$]{\includegraphics[scale=0.17]{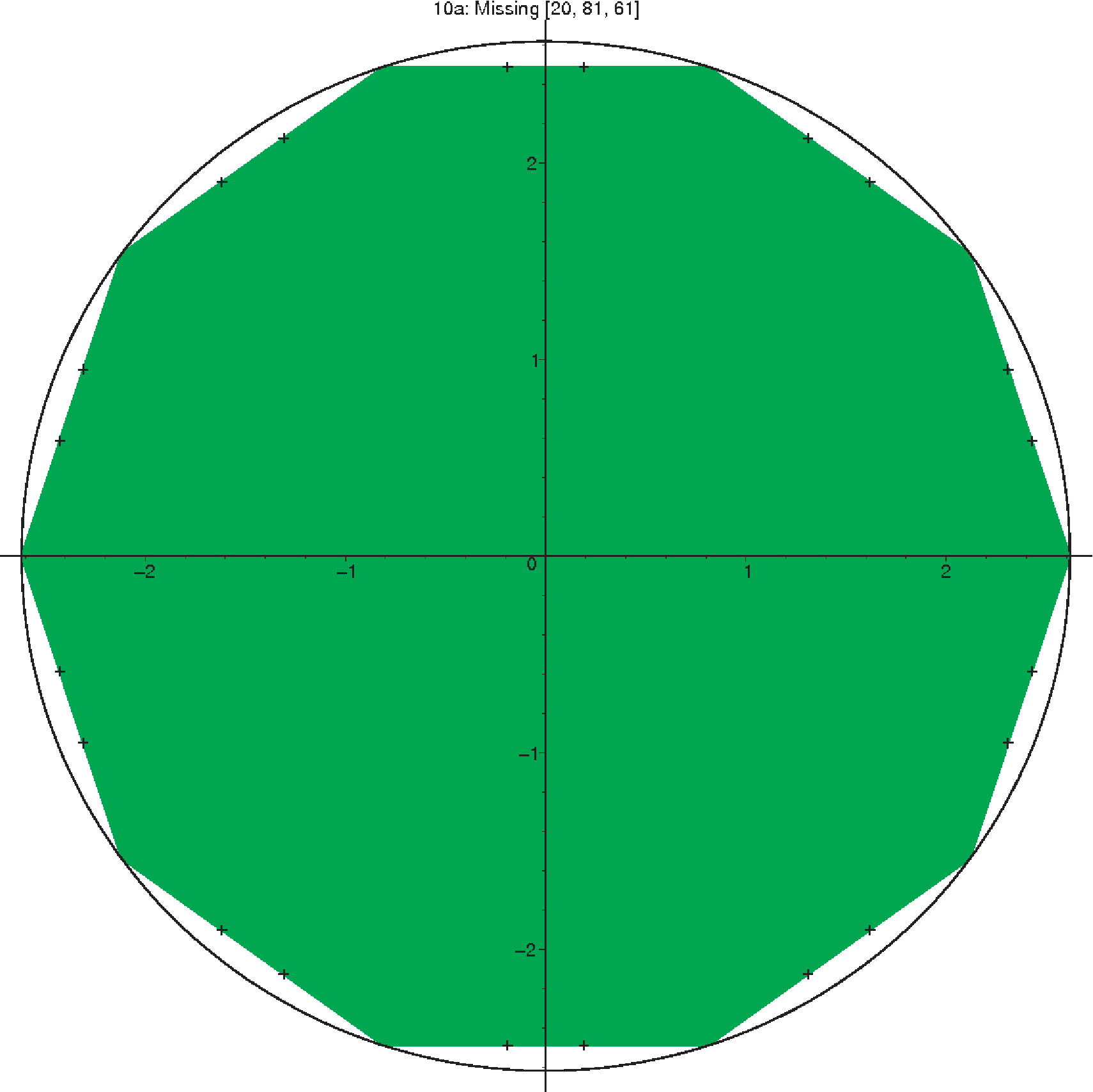}}
\subfigure[$\tau^2$ with $n=10$]{\includegraphics[scale=0.17]{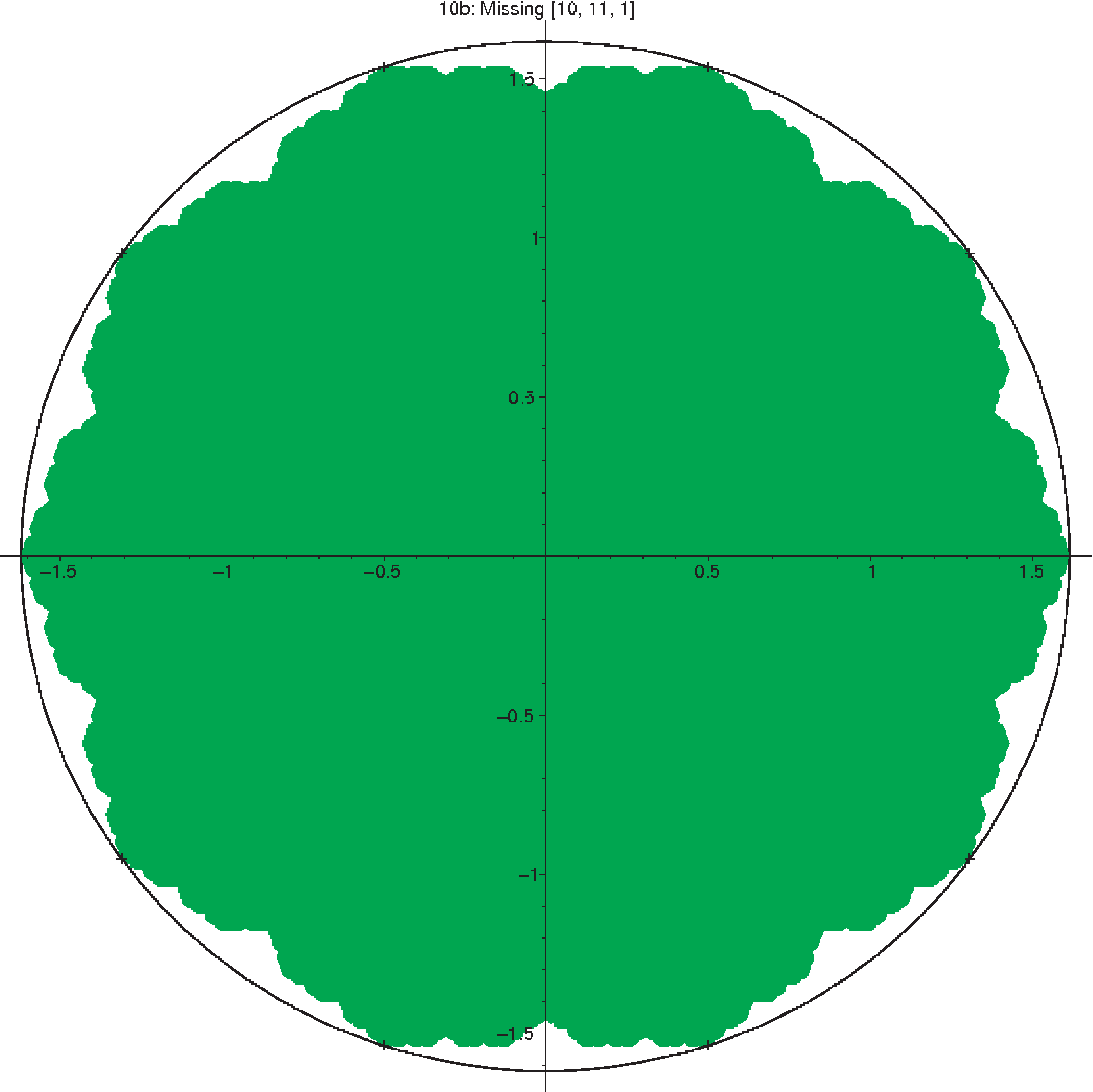}}
\\
\subfigure[$\delta$ with $n = 8$]{\includegraphics[scale=0.17]{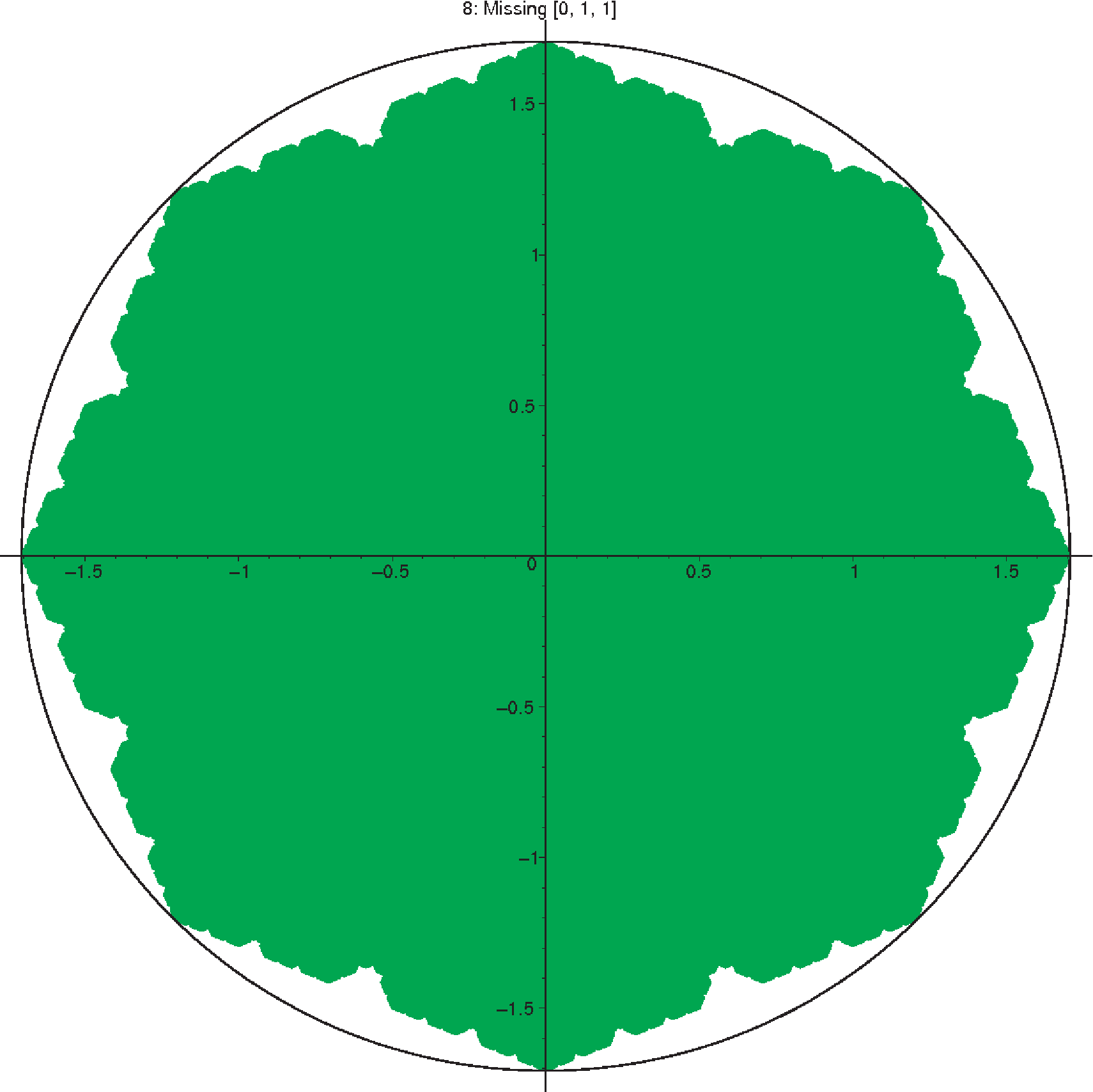}}
\subfigure[$\mu$ with $n = 12$]{\includegraphics[scale=0.17]{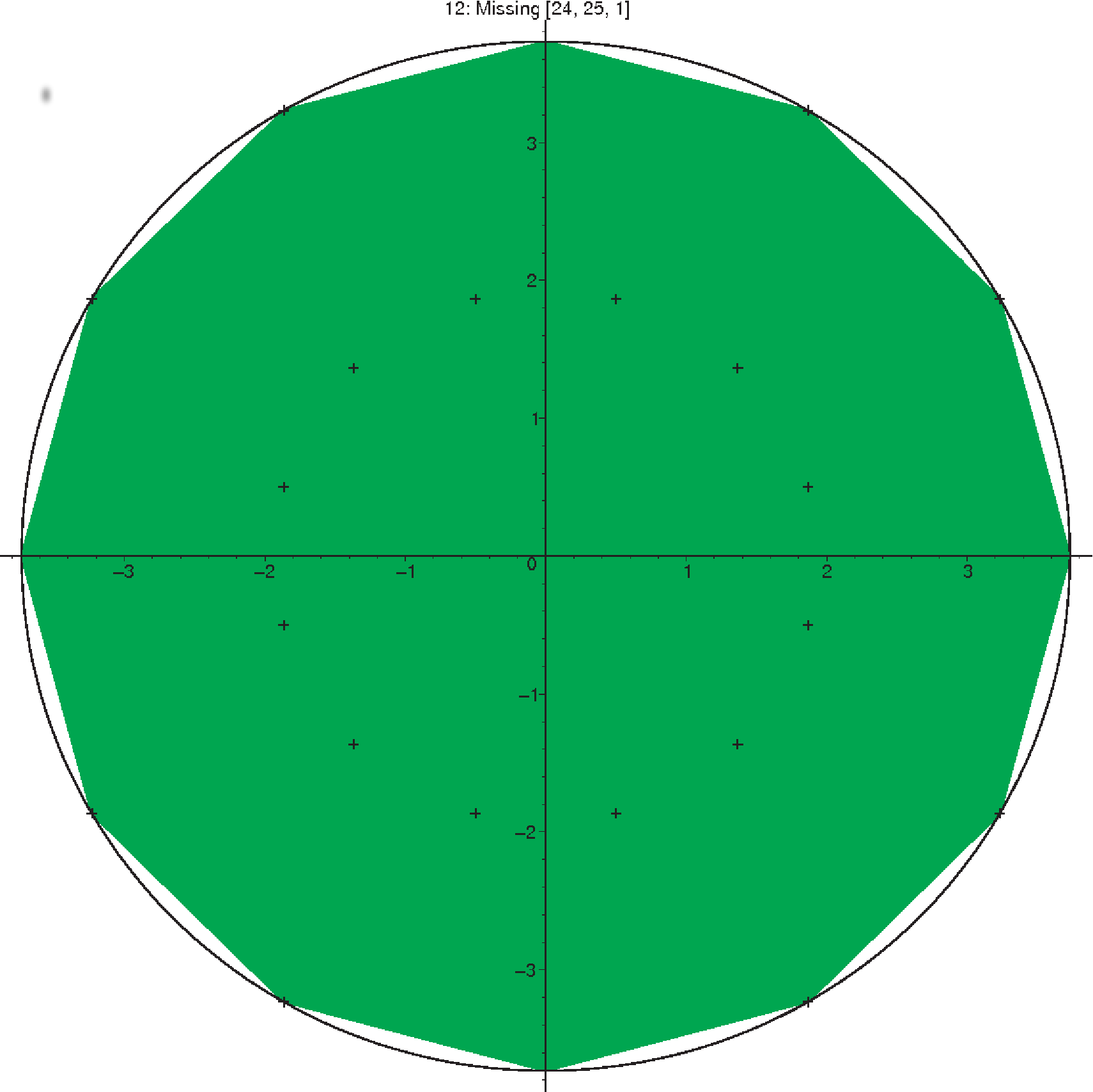}}
\caption{The acceptance windows with the missing points marked, in cases corresponding to the spectra of quadratic
Pisot-cyclotomic numbers.}
\label{fig:missingquadr}
\end{figure}


Let us briefly explain how to decide whether given $x\in\Z[\omega]$ belongs to $\Sigma\big(K(\sigma(\beta,\A))\big)$. This amounts to checking whether $\sigma(x)\in K(\sigma(\beta,\A))$, i.e. whether there is a sequence of digits $a_i\in\A$ such that $\sigma(x) = \sum_{i=0}^{+\infty}a_i\sigma(\beta)^i$. For $a\in\A$ define $T_a(z) := \frac{z-a}{\sigma(\beta)}$. It is reasonable to consider only those digits $a$ such that $T_a(\sigma(x))\in \mathrm{cl}(B_{R}(0))$.
We claim that under the assumption that $\beta^{-1}$ is an algebraic integer, if
  $b_0,b_1,\dots,b_k\in\A$, $k\in\N$, are such that $T_{b_i}T_{b_{i-1}}\cdots T_{b_0}(\sigma(x))\in \mathrm{cl}(B_{R}(0))$, then $T_{b_k}T_{b_{k-1}}\cdots T_{b_0}(x)$ takes only finitely many values. This fact
can be shown in the same fashion as Lemma~9 of~\cite{BaMaPeVa}.

As a consequence, either there is no infinite sequence of digits $(b_i)_{i\geq 0}$ such that $T_{b_i}T_{b_{i-1}}\cdots T_{b_0}(\sigma(x))\in \mathrm{cl}(B_{R}(0))$, for each $i$, and hence $\sigma(x)\notin K(\sigma(\beta),\A)$. Or, it happens that
$$
T_{c_m}\cdots T_{c_n}\cdots T_{c_1}T_{c_0}(x) = T_{c_n}\cdots T_{c_1}T_{c_0}(x),
$$
and then we have that $\sigma(x)$ belongs to $K(\sigma(\beta),\A)$, and, moreover, it has an eventually periodic representation in powers of $\sigma(\beta)$ with preperiod $c_0c_1\dots c_n$ and period $c_{n+1}\dots c_m$. Therefore checking if $\sigma(x)\in K(\sigma(\beta),\A)$ can be done in a finite time even if $K(\sigma(\beta),\A)$ has a fractal boundary.

For the case of symmetry of order 12, we consider the Pisot-cyclotomic number $\mu=1+\sqrt3$, which is not a unit. Hence we cannot apply Lemma~\ref{l:podminkaS}, and consequently the corollaries. Nevertheless, we can decide that $\Sigma(\Omega) \nsubseteq X^{\A}(\beta)$ for $\Omega={\rm int}\big(K(\sigma(\mu,\A_{12}))\big)$, since already in the ball $B_{1/(\beta-1)}(0)$ there are missing points, as it is seen from Figure~\ref{fig:missingquadr}e.

Similar procedure as in the quadratic case can be performed for the cases of cubic spectra, where both considered numbers are units. It turns out that none of the relatively dense cubic spectra corresponds to a cut-and-project set with simply connected acceptance windows, since there are many missing points. This is the reason we do not include the details of the computation.


\section{Spectra and their Voronoi tilings}

Let us focus on the Voronoi tiling of the spectrum. Such a tiling is important, for it provides a natural definition of `neighbours'
of points in an aperiodic Delone set.
Formally, the Voronoi tile of a point $x$ in a Delone set $\Lambda\subset\C$ is defined as
$$
V(x) = \{z\in\C: \text{ for any }y\in\Lambda,\ |x-z|\leq |y-z| \}\,.
$$
It is not difficult to see that only the closest points of $\Lambda$ influence the shape of the tile $V(x)$. In fact,
one needs to look only to the distance $2r_c$, where $r_c$ is the covering radius of $\Lambda$. The smallest finite subset $F\subset\Lambda$
such that
$$
V(x) = \{z\in\C: \text{ for any }y\in F,\ |x-z|\leq |y-z| \}
$$
can be considered as the set of neighbours.

Consider a central point $x \in \Lambda$ and additional points
    $z_i \in \Lambda$ which surround $x$.
Between each $z_i$ and $x$ there is a line dividing $\C$ such that each point on this line
    has equal distance from $z_i$ and $x$.
These lines form the boundary of the tile centered at $x$. One can naturally define the neighbours of $x$ in $\Lambda$
as those points $z_i$, whose tiles $V(z_i)$ share a line segment with $V(x)$.

Denote the vertices of the tile $V(x)$ by  $v_1, v_2, \dots$ where
    each $v_i$ depends only on $x$, $z_i$ and $z_{i+1}$.
As $v_i$ is the circumcenter of the triangle $x, z_i$ and $z_{i+1}$ we
    have that $r=|x-v_i| = |z_i - v_i| = |z_{i+1} - v_i|$. Note that the disc of radius
    $r$ centered at $x$ contains no elements of $\Lambda$ in its interior. Therefore $r\leq r_c$.

The maximal distance $|v_i - x|$  will be called the {\em radius of the tile} $V(x)$ and denoted by $r_x$.
It is not difficult to realize that taking the supremum of radii of tiles over all points of $\Lambda$,
we obtain the covering radius $r_c$ of $\Lambda$,
\begin{equation}\label{eq:rc=rx}
r_c=\sup\{ r_x : x\in\Lambda\}\,.
\end{equation}
 If the set $\Lambda$ has finite local complexity, the
supremum is achieved, i.e.\ there exists a tile whose radius is equal to $r_c$.

An upper bound on the covering radius of the spectra was given in Proposition~\ref{p:rep-rd}.
In this section, we will find the exact value for $r_c$ for all our cases of the spectra of Pisot-cyclotomic numbers
which are Delone. This will further allow us to provide a bound on the number of tiles in the corresponding tiling.
In several cases, a more thorough inspection of the problem will allow us to determine the number of tiles precisely, see Section~\ref{sec:lc}.

To find a lower bound on $r_c$, simply compute the spectrum in some large neighbourhood of the origin, together with its Voronoi tiling,
and find the maximum radii of these tiles. Computing an upper bound is more challenging, but still feasible.
We will find an upper bound on $r_c$ by inspecting the radius of certain tiles in the tiling created by the (finite) set $X_{n}^\A(\beta)$.
Obviously, there are unbounded tiles at the boundaries, however, we will check only tiles in a suitable region $A_n$ around the origin.

Let $X_n^\A(\beta)$ be a subset of $X^\A(\beta)$ that is constructible by polynomials up to degree $n$,
$$
X_n^\A(\beta) =  \Big\{\sum_{j=0}^na_j\beta^j : a_j\in\A\Big\}\,.
$$
Obviously, $X_{n+1}^\A(\beta)$ can be constructed by a union of translates of copies of $X_n^\A(\beta)$, namely
$$
X_{n+1}^\A(\beta) =  \bigcup_{a\in\A}\big(a\beta^{n+1} + X_n^\A(\beta)\big)\,.
$$
For $n\in\N$, consider regions $A_n$, satisfying
\begin{equation}\label{eq:A-cr}
A_{n+1} \subseteq \bigcup_{a\in\A} \big(a\beta^{n+1}+A_n\big)\,.
\end{equation}
It can be shown that we can set $A_n=\beta^n B_R(0)$, where $R=R(\theta, \gamma)$ is given by the formula
\begin{equation}\label{rtheta}
   R(\theta, \gamma) = 
   \frac{\gamma^2\cos(\theta) + \gamma\sqrt{1-\gamma^2\sin^2(\theta)}}{\gamma^2-1}\,,
\end{equation}
see Figure \ref{fig:rad} and Table \ref{tab:covrad}.

\begin{figure}
\subfigure[$\tau$ with $n=5$]{\includegraphics[scale=0.17]{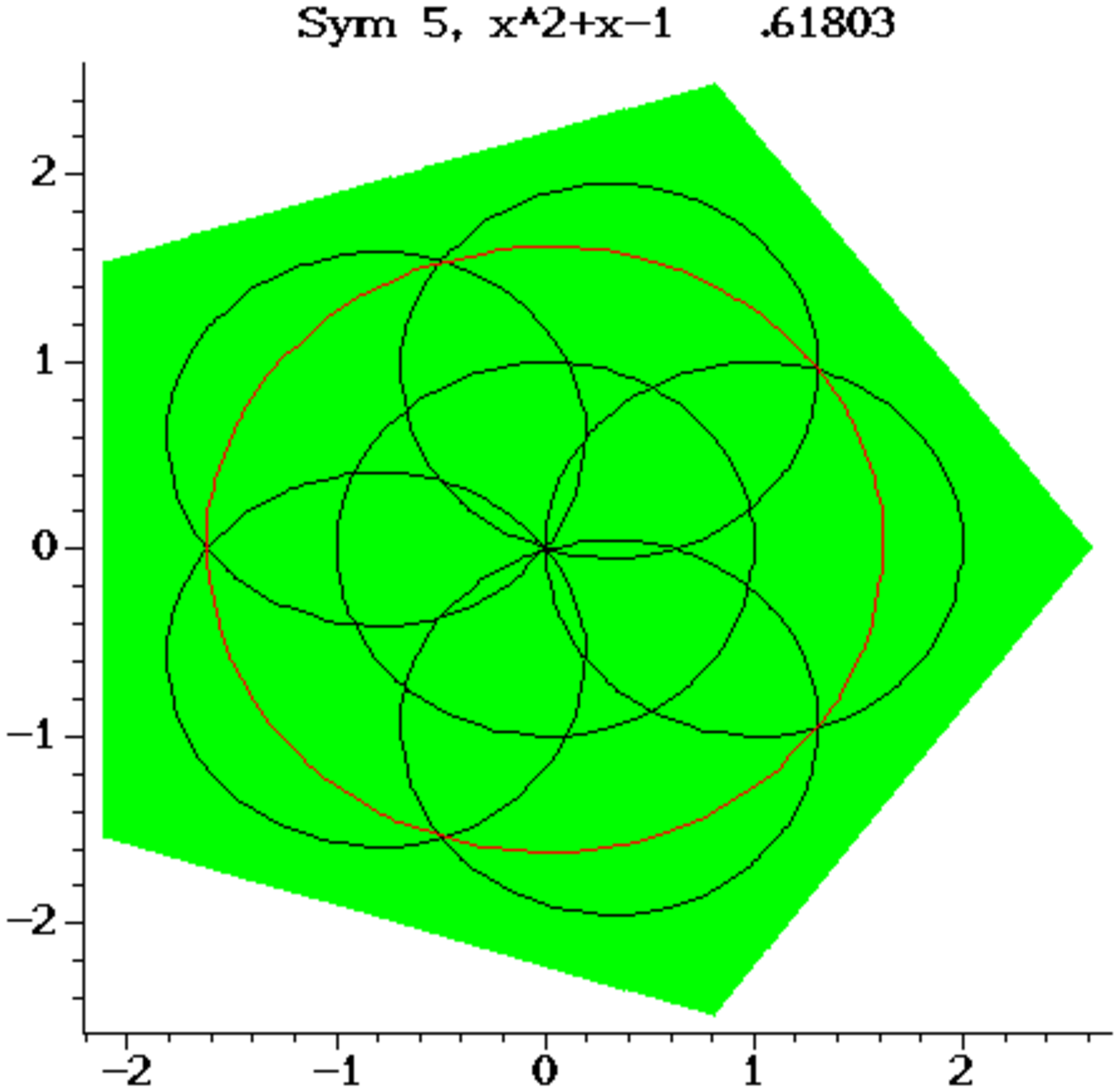}}
\subfigure[$\tau$ with $n=10$]{\includegraphics[scale=0.17]{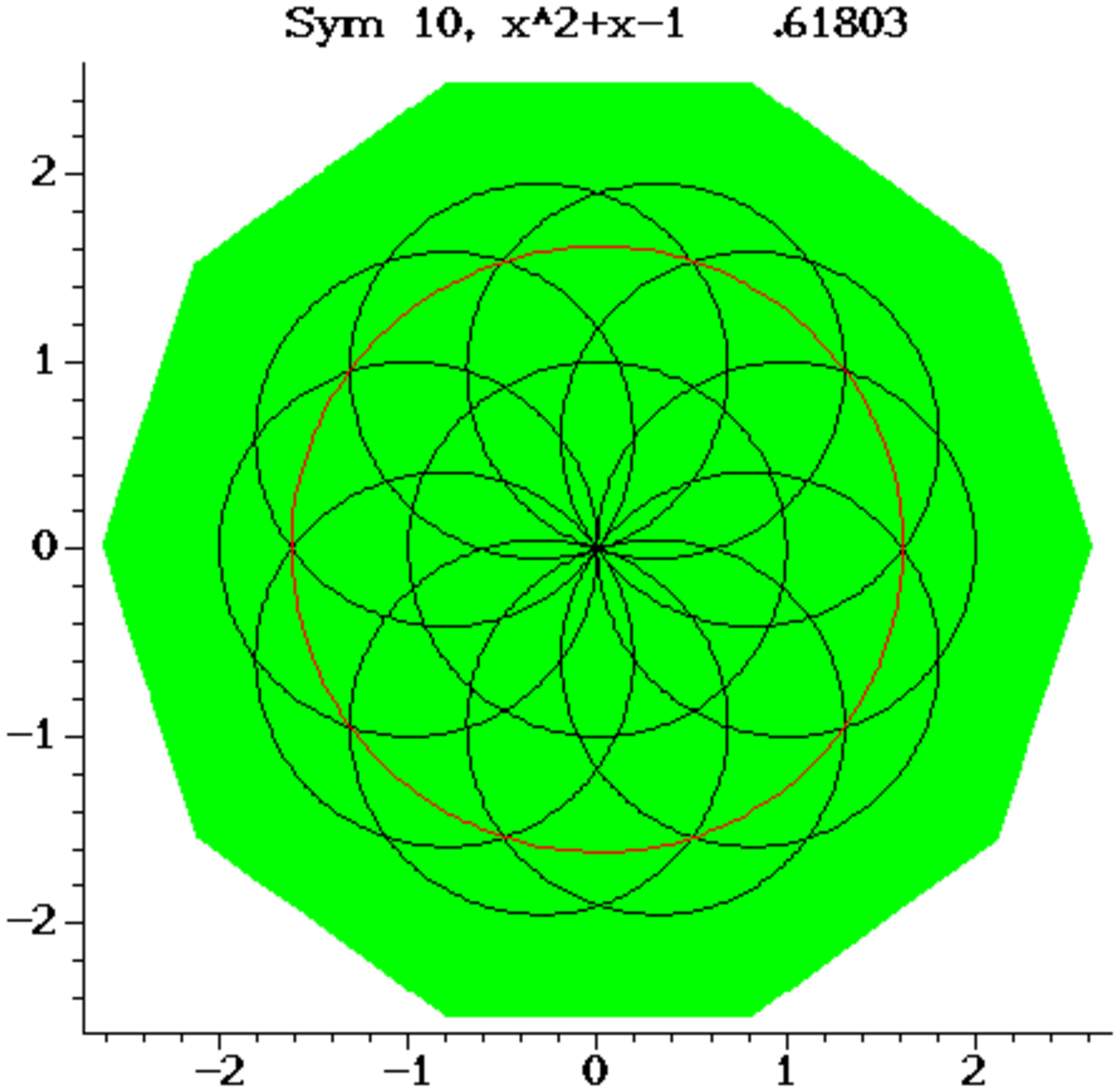}}
\subfigure[$\tau^2$ with $n=10$]{\includegraphics[scale=0.17]{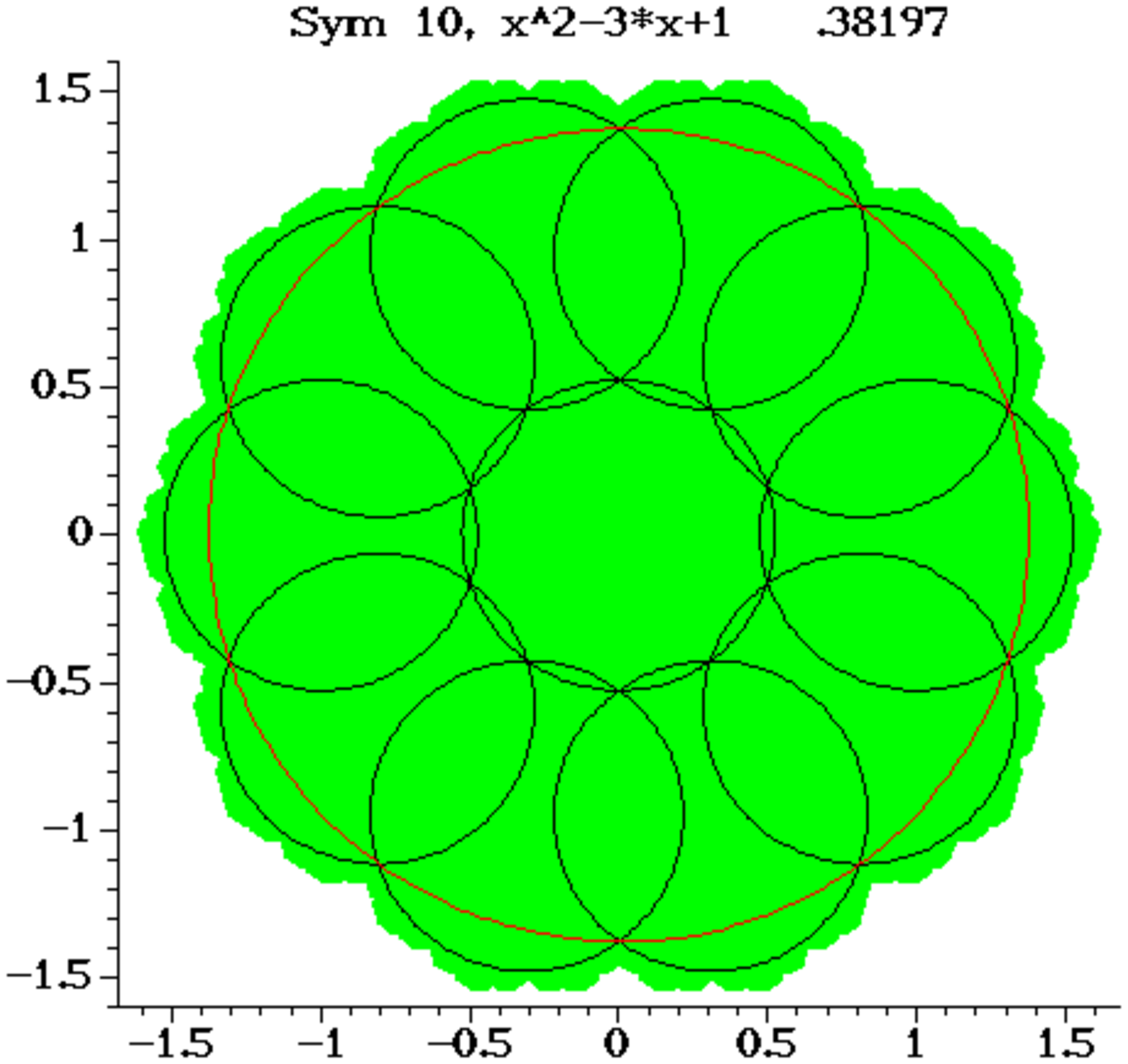}}
\\
\subfigure[$\lambda$ with $n = 7$]{\includegraphics[scale=0.17]{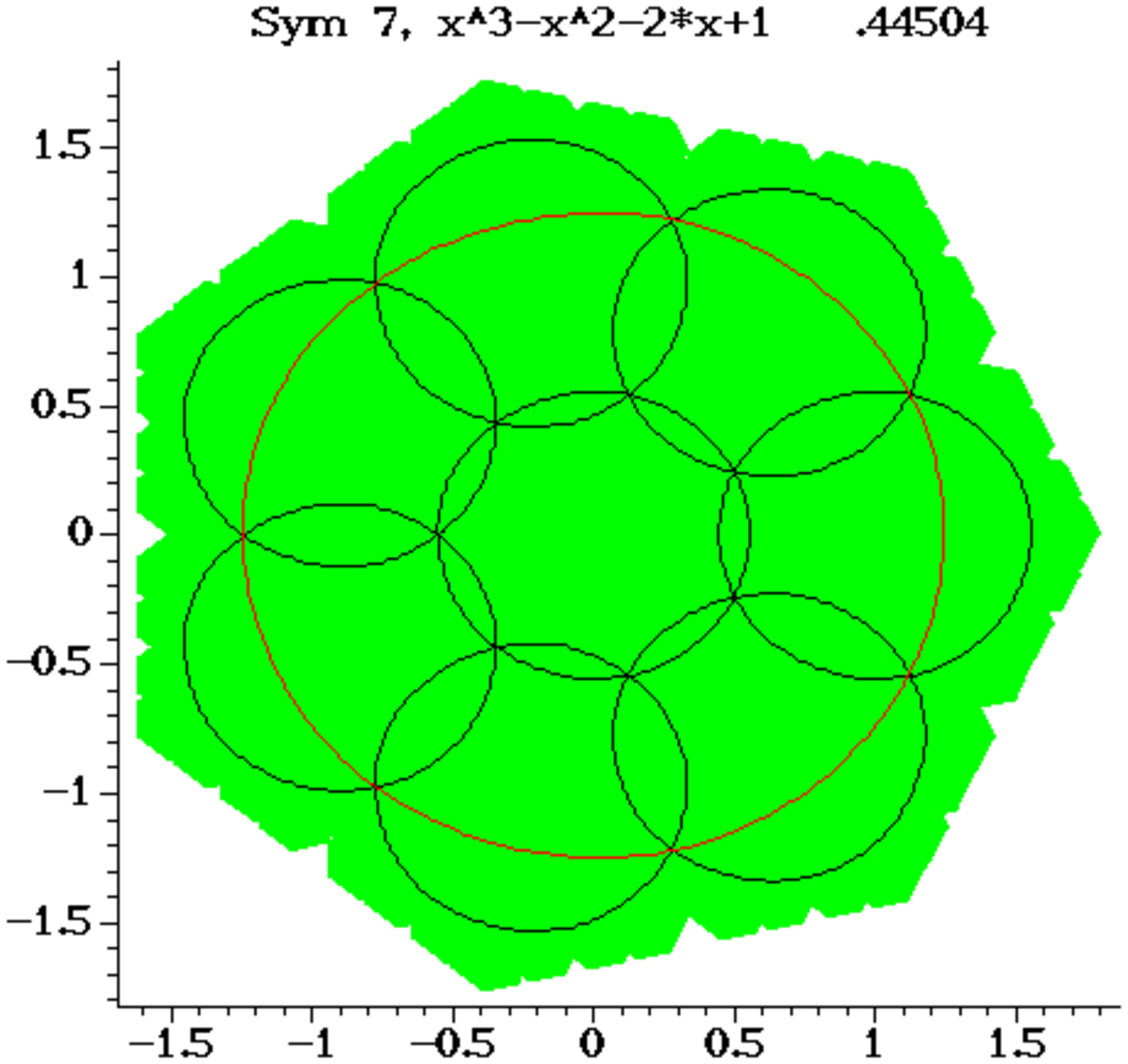}}
\subfigure[$\lambda$ with $n = 14$]{\includegraphics[scale=0.17]{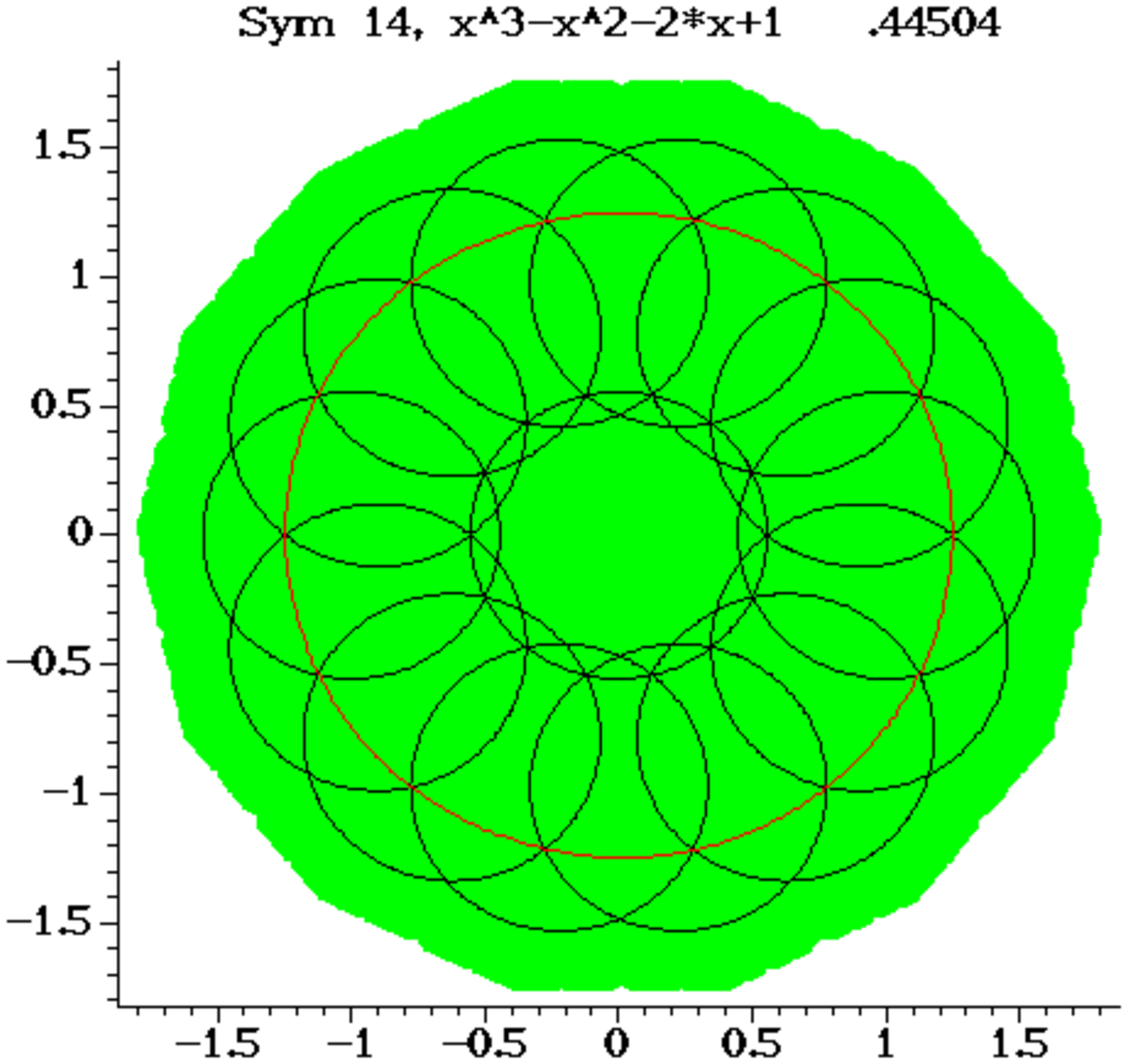}}
\subfigure[$\delta$ with $n = 8$]{\includegraphics[scale=0.17]{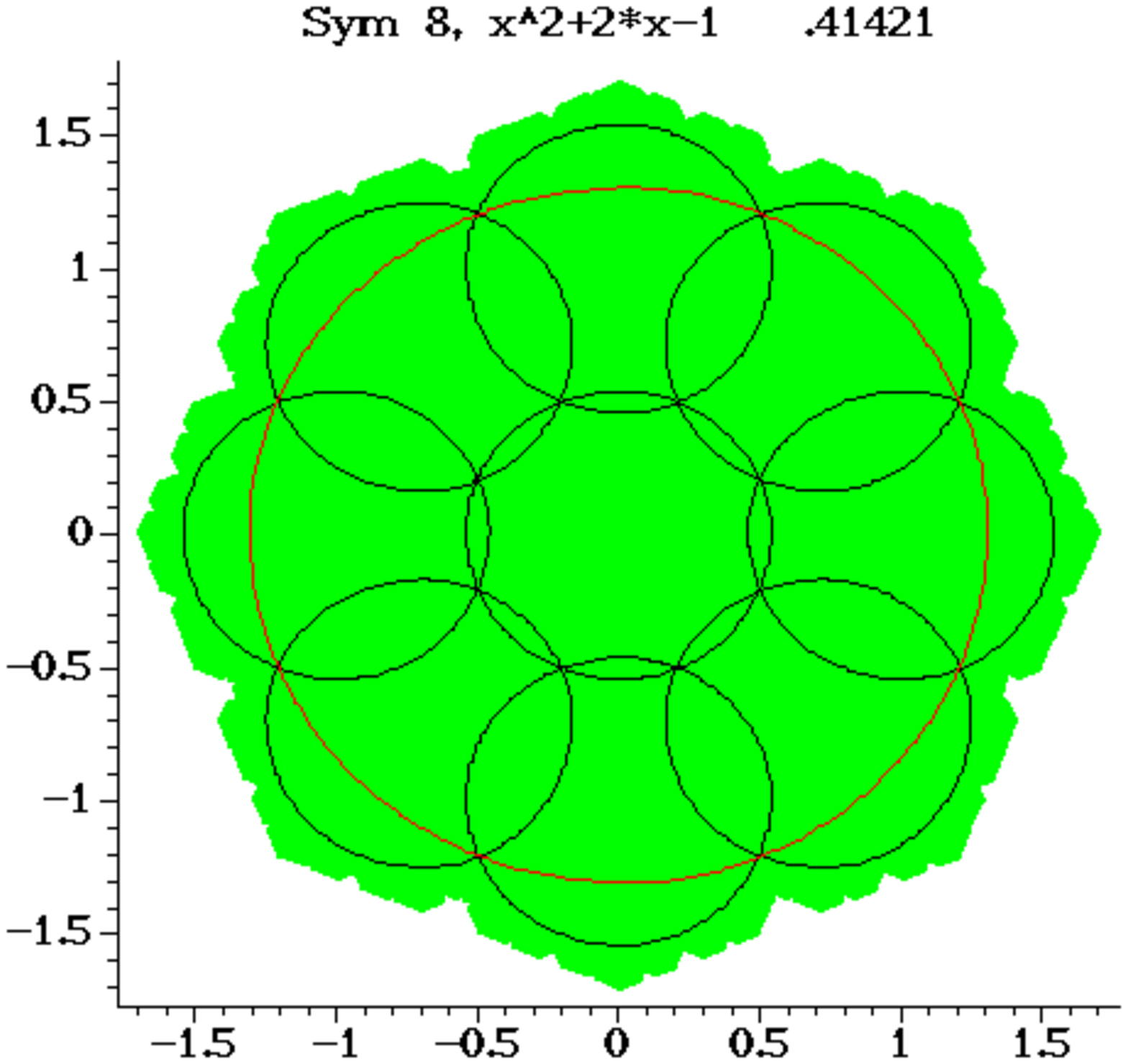}}
\\
\subfigure[$\kappa$ with $n = 18$]{\includegraphics[scale=0.17]{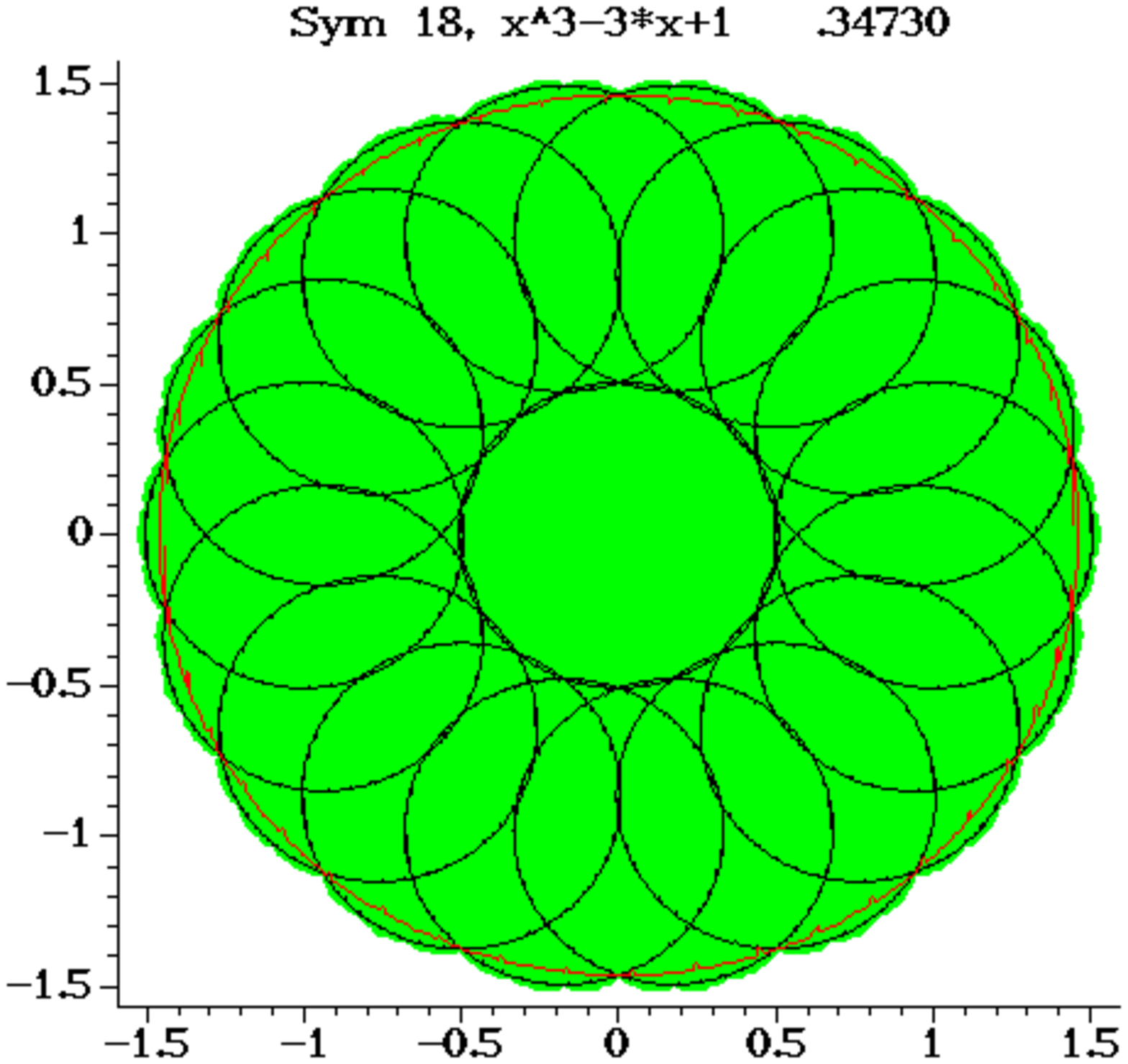}}
\subfigure[$\mu$ with $n = 12$]{\includegraphics[scale=0.17]{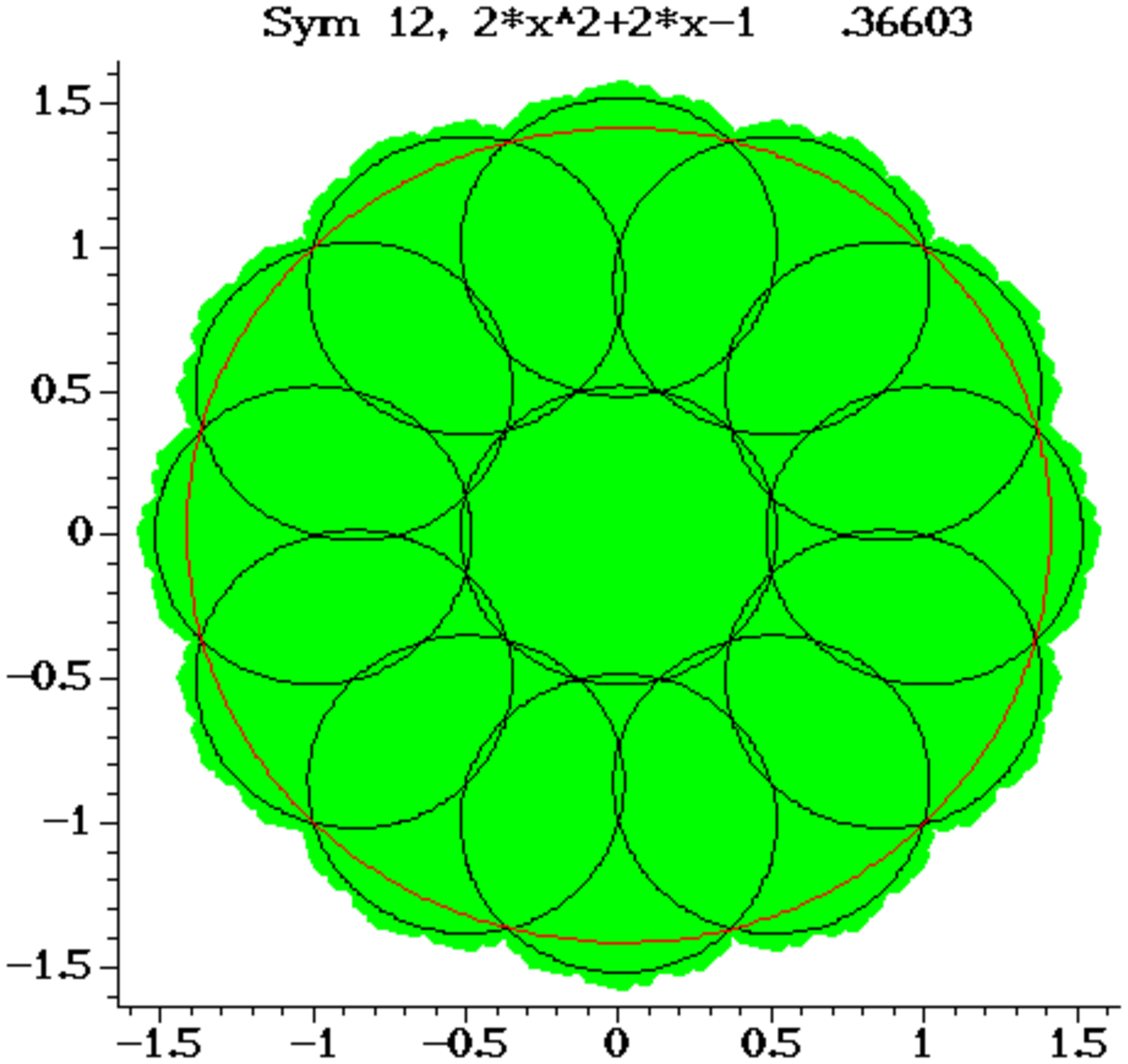}}
\caption{Ball of radius $R(\theta, \gamma)$ (red), the images of the ball of radius $R(\theta, \gamma)$ under the map $y \to y/\beta + a_i$ (black) and the attractor for this map (green).}
\label{fig:rad}
\end{figure}

Such a set has the following advantage for our computation.
Take a point $x\in X_{n+1}^\A(\beta)$ which belongs to $A_{n+1}$. 
Then there exists an $a\in\A$ such that $x \in (a\beta^{n+1} + A_n)\cap (a\beta^{n+1}+X_n^\A(\beta))$.
Therefore one has a corresponding point
$y\in X_{n}^\A(\beta)\cap A_{n}$ such that $x=a\beta^n+y$.

Consider the Voronoi tile of the point $y$ created by points $z_1,\dots,z_k$ of $X_{n}^\A(\beta)\cap A_n$.
Then the Voronoi tile of the point $x=y+a\beta^{n+1}$ is created by points $z_1+a\beta^{n+1},\dots,z_k+a\beta^{n+1}$, which lie in
$X_{n+1}^\A(\beta)\cap A_{n+1}$, and, perhaps, some additional points of $X_{n+1}^\A(\beta)$.
The Voronoi tile $V_x$ of $x$ with respect to $X_{n+1}^\A(\beta)$ is therefore cut from the Voronoi tile $V_y$ of $y$ with respect to $X_{n}^\A(\beta)$.
Consequently, the radius of $V_x$ is smaller or equal to the radius of $V_y$, which, in turn, is bounded from above by the maximum of radii of
tiles in the tiling of $X_{n}^\A(\beta)$ centered in points of $A_n$. We denote this maximum by $\Delta_n$.
We have thus demonstrated the following statement.

\begin{coro}
Let $\Delta_n$ be the maximum of the radii of tiles in the tiling of $X_n^\A(\beta)$ centered at
    points in $X_n \cap A_n$. Then $(\Delta_n)_{n\geq n_0}$ is a decreasing sequence such that
    $\Delta_n \to r_c$ as $n$ tends to infinity.

Moreover, as $X^\A(\beta)$ is of finite local complexity, there exists
    some $N_0$ such that $\Delta_n=r_c$ for all $n\geq N_0$.
\end{coro}

By inspecting tiles in the tiling of $X_{n}^\A(\beta)$, $n\geq 1$, one obtains a sequence of upper estimates
of the covering radius of $X^\A(\beta)$. In our computation, we were successful, after a few steps, in obtaining
an upper bound on $r_c$ which coincides with the lower bound obtained as explained earlier. Thus, we have derived the exact value of the
covering radius for all our cases. The values of $r_c$ of the spectra of quadratic and cubic Pisot-cyclotomic numbers are,
together with $n$ such that $r_c=\Delta_n$, displayed in Table~\ref{tab:covrad}.

\begin{table}
\renewcommand{\arraystretch}{1.3}
\begin{tabular}{llllllll}
\hline
Symmetry & $\beta$   & $\theta$          &  $R(\theta,\beta)$ & $r_c$ &  $n$\\
\hline
5        & $\tau$    & $\frac{2\pi}{10}$ &  1.6180339895 & .7639320250  & 6 \\
10       & $\tau$    & $\frac{2\pi}{10}$ &  1.6180339895 & .6498393940  & 3 \\
         & $\tau^2$  & $\frac{2\pi}{10}$ &  1.3763819202 & 1.051462225  & 1\\
7        & $\lambda$ & $\frac{2\pi}{14}$ &  1.2469796034 & 1.109916265  & 1 \\
14       & $\lambda$ & $\frac{2\pi}{14}$ &  1.2469796034 & 1.025716864  & 1\\
8        & $\delta$  & $\frac{2\pi}{16}$ &  1.3065629649 & 1.082392201  & 1 \\
18       & $\kappa$  & $\frac{2\pi}{36}$ &  1.4619022000 & 1.015426612  & 1 \\
12       & $\mu$     & $\frac{2\pi}{24}$ &  1.4142135622 & 1.035276182  & 1 \\
\hline\\
\end{tabular}
\caption{Pisot cyclotomic numbers, radius of interior ball for computation of the covering radius $r_c$ of the corresponding spectrum, the computed value of $r_c$. The last column shows the iteration needed to obtain the neighbourhood of the origin determining its Voronoi tile.}
\label{tab:covrad}
\end{table}


\begin{remark}
It can be checked that, in fact, the covering radius of the spectrum is in all our cases equal to the
radius $r_0$ of the Voronoi tile of the origin.
\end{remark}

\section{Local configurations}\label{sec:lc}

In a Delone set with finite local complexity, it is interesting to determine the number of local configurations
of a given size. In the most restricted form of this question, we are interested in the number of different Voronoi tiles.
Since a tile is given by a configuration of points within the distance at most $2r_c$,
we immediately have that the Voronoi tiling is composed of translates of at most $2^{|X^{\A-\A}(\beta) \cap B_{2 r_c}(0) |}$
   tiles.
Clearly, this estimate is not reasonable. The bounds on the number of prototiles in the Voronoi tiling of the spectra can be improved by inspecting the local configurations of a given size. This method can be quite computationally expensive, but it yields also the shapes of the appearing tiles. At least in certain cases, we provide a complete description of the tiles in the Voronoi tiling, see Table~\ref{tab:lc}.

Consider $x \in X^{\A}(\beta)$ and define its local configuration as
    \[ lc(x) = \left\{y - x \,:\, y \in X^{\A}(\beta), |y-x| < \frac{2}{\beta-1} \right\}. \]
We will say that two local configurations are equivalent if one is a rotation,
    or reflections of the second.

There are two key ideas about local configurations that are important:

\begin{claim}
The Voronoi tile $V(x)$ of $x$ is completely determined by the points in $lc(x)$.
\end{claim}

This observation follows from~\eqref{eq:rc=rx} and the upper bound on the covering radius $r_c$ obtained in Proposition~\ref{p:rep-rd}.

\begin{claim}
If $x=\beta z + a_i$ for some $z\in X^{\A}(\beta)$ and $a_i\in\A$, then the local configuration of $x$ can be completely determined
from the local configuration of $z$.
\end{claim}

\pfz
To see the second fact, we need to show that if $y-x$ is in $lc(x)$ of $x=\beta z + a_i$, then we can write $y=\beta w+a_j$
for some $a_j\in\A$ and $w\in X^{\A}(\beta)$ such that $w-z$ belongs to $lc(z)$. This is easy to verify:
As $y\in X^{\A}(\beta)$, there exist $w\in X^{\A}(\beta)$, $a_j\in\A$ such that $y=\beta w+a_j$.
We have
$$
\frac{2}{\beta-1}>|y-x|=|\beta(w-z)+a_j-a_i|\geq \beta|w-z| - |a_j-a_i| > \beta|w-z| - 2\,,
$$
from which we can derive that
$$
|w-z|<\frac{2}{\beta-1}\,.
$$
Therefore $w\in lc(z)$, as required.
\pfk

We observe that $\max |a_i - a_j| < 2$ for $\A_5$ and $\A_7$, and hence one could consider local configurations of smaller radius than
$\frac{2}{\beta-1}$ in these two cases.

It follows that local configurations of all points $x\in X^{\A}(\beta)$ are in this sense descendants of the local configuration of the origin.
This suggests the following algorithm for computation of all local configurations in the spectrum $x\in X^{\A}(\beta)$.

\begin{itemize}
\item
Start with $lc(0)=X^{\A}(\beta)\cap B_0(\frac{2}{\beta-1})$.

\item
From $lc(z)$, compute for $a_i\in A$ the local configuration $lc(\beta z+a_i)$  as a set of those points $y$
in $\bigcup_{a\in\A}\beta lc(z) + a$, that are in distance $\frac{2}{\beta-1}$ from $\beta z+a_i$.

\item
Repeat this process until finding the complete family of local configurations in the spectrum. The procedure must stop, since $X^{\A}(\beta)$
is of finite local complexity.
\end{itemize}

\begin{ex}
For example for $\beta = \tau^2$ and $\A = \A_{10}$ the local configuration
    of $0$ is $\{0, 1, \omega, \omega^2, \dots, \omega^9\}$.
In fact, up to equivalents (rotation/reflection), there are 20 different local configurations.
The algorithm described above gives rise to the diagram illustrated in Figure \ref{fig:lc}.
These 20 local configurations generate only five distinct tiles.

\begin{figure}
\includegraphics[scale=0.4]{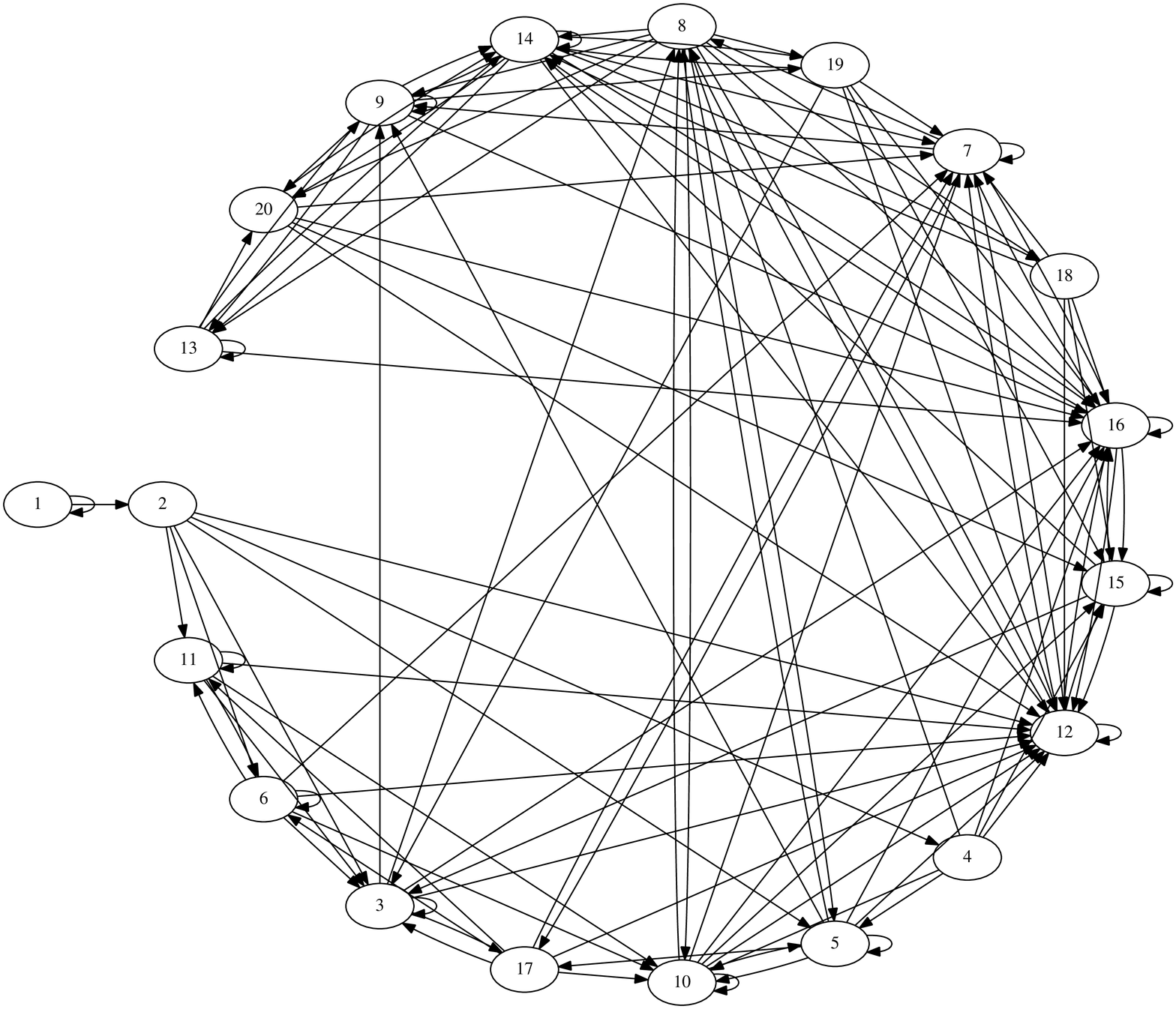}
\caption{Local configuration diagram for $\beta=\tau^2$ with $\A_{10}$}
\label{fig:lc}
\end{figure}
\end{ex}

Table~\ref{tab:lc} shows the results of computations of local configurations in the considered cases of
spectra of Pisot-cyclotomic numbers. Together with the number of local configurations we provide also the number
of non-equivalent Voronoi tiles that these configurations generate. It turns out that only symmetry 5, 10, and 8
gives rise to `reasonable' tilings with low number of different prototiles.

\begin{table}
\renewcommand{\arraystretch}{1.2}
\begin{tabular}{llll}
\hline
Symmetric  & $\beta$ & Number of configurations & Number of tiles \\
\hline
5        & $\tau$    & $7823$         & $12$ \\
10       & $\tau$    & $3818$         & $5$  \\
10       & $\tau^2$  & $20$           & $5$  \\
7        & $\lambda$ & $\geq 279$     & $\geq 201$ \\
14       & $\lambda$ & $\geq 815$     & $\geq 189$ \\
8        & $\delta$  & $26$           & $5$        \\
18       & $\kappa$  & $\geq 881$     & $\geq 154$ \\
12       & $\mu$    & $\geq 1002$     & $\geq 104$ \\
\hline\\
\end{tabular}
\caption{Pisot cyclotomic numbers \& the number of tiles}\label{tab:lc}
\end{table}

\section{Voronoi tilings for the quadratic and cubic cases}

Let us summarize the knowledge about the spectra in individual cases of quadratic and cubic Pisot-cyclotomic numbers.
We will consider only the cases which lead a relatively dense set, as given in Table~\ref{tab:rd}.

In Section~\ref{sec:lc}, we have provided a method for finding prototiles forming the Voronoi tiling of the spectra. In cases of
symmetry 12, 7, 14 and 18, we have shown that the corresponding tilings are rather complex, with the number of shapes
of prototiles exceeding a hundred. On the other hand, for symmetry 5 ($\tau$), 10 ($\tau$, $\tau^2$) and 8 ($\delta$), our method
yielded a complete list of prototiles. Moreover, the spectra $X^{\A_{10}}(\tau)$, $X^{\A_{10}}(\tau^2)$ and $X^{\A_{8}}(\delta)$
also have the interesting property, that they correspond to a cut-and-project set with simply connected acceptance window,
Therefore, one can consider easily the relation of the shape of the Voronoi cell of $x$ to the position of the image of $x$
in the acceptance window. These are quadratic cases, where, by Theorem~\ref{thm:conj=ifs},
$\Omega\subseteq {\rm int}\big(K(\sigma(\beta),\A)\big)$, $\sigma$ denoting the unique automorphism of $\Q(\omega)$
with non-trivial action on $\Q(\beta)$.

Suppose that the Voronoi cell of $x$ is given by the neighbours $z_1,z_2,\dots,z_k$ of $x$. We have
$\sigma(x),\sigma(z_i)\in \Omega$, $i=1,\dots,k$. All points $y=x+t\in X^{\A}(\beta)$ such that
$\sigma(y),\sigma(z_i)+\sigma(t)\in\Omega$, $i=1,\dots,k$, will either have the same type of Voronoi cell, or their Voronoi
cell will be just a `cut' of $V(x)$.
Therefore the acceptance window $\Omega$ will be divided into regions corresponding
to prototiles, such that for any $\sigma(y)$ in a given region, $y$ has the corresponding tile as its Voronoi cell.
These regions are bounded by intersections of the boundary of several suitably shifted copies of $\Omega$.
This phenomenon is illustrated in Figures~\ref{fig:tiles-accw10a},~\ref{fig:tiles-accw10b}, and~\ref{fig:tiles-accw8}.
On the other hand, Figure~\ref{fig:tiles-accw5a} shows that in case there are some missing points in the acceptance window, the situation is much more complicated.

\subsection{Case $\boldsymbol{n=5}$, $\boldsymbol{\beta=\tau}$}

 Consider $K(\sigma(\tau), \A_5)$, see Figure \ref{fig:missingquadr}.
    In Section~\ref{sec:cap}, we have seen that the spectrum with pentagonal alphabet does not correspond to a cut-and-project set with simply connected acceptance window, since there are points in $\Sigma(\Omega)\setminus X^{\A_5}(\tau)$, where $\Omega=\mathrm{int}(K(\sigma(\tau),A_{5}))$. Nevertheless, the tiling of this case of spectra can be described, see Figure~\ref{fig:tiles-accw5a}. Surprisingly, there are only 12 prototiles.

\begin{figure}[ht]
\subfigure[$\sigma(X^{\A_{5}}(\tau))$]{\includegraphics[scale=0.25]
    {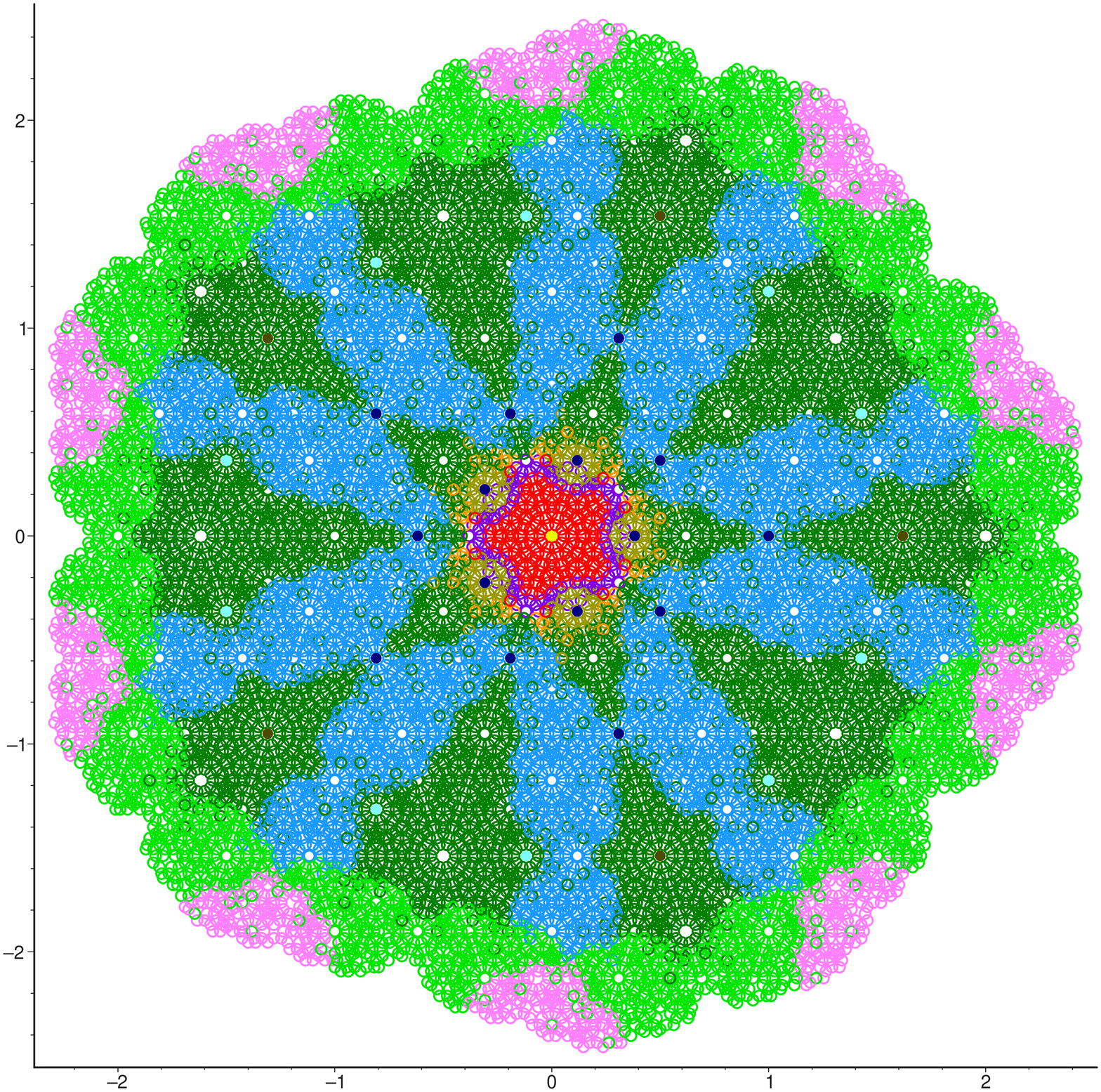}}
\subfigure[Tiling of $X^{\A_{5}}(\tau)$]{\includegraphics[scale=0.25]
    {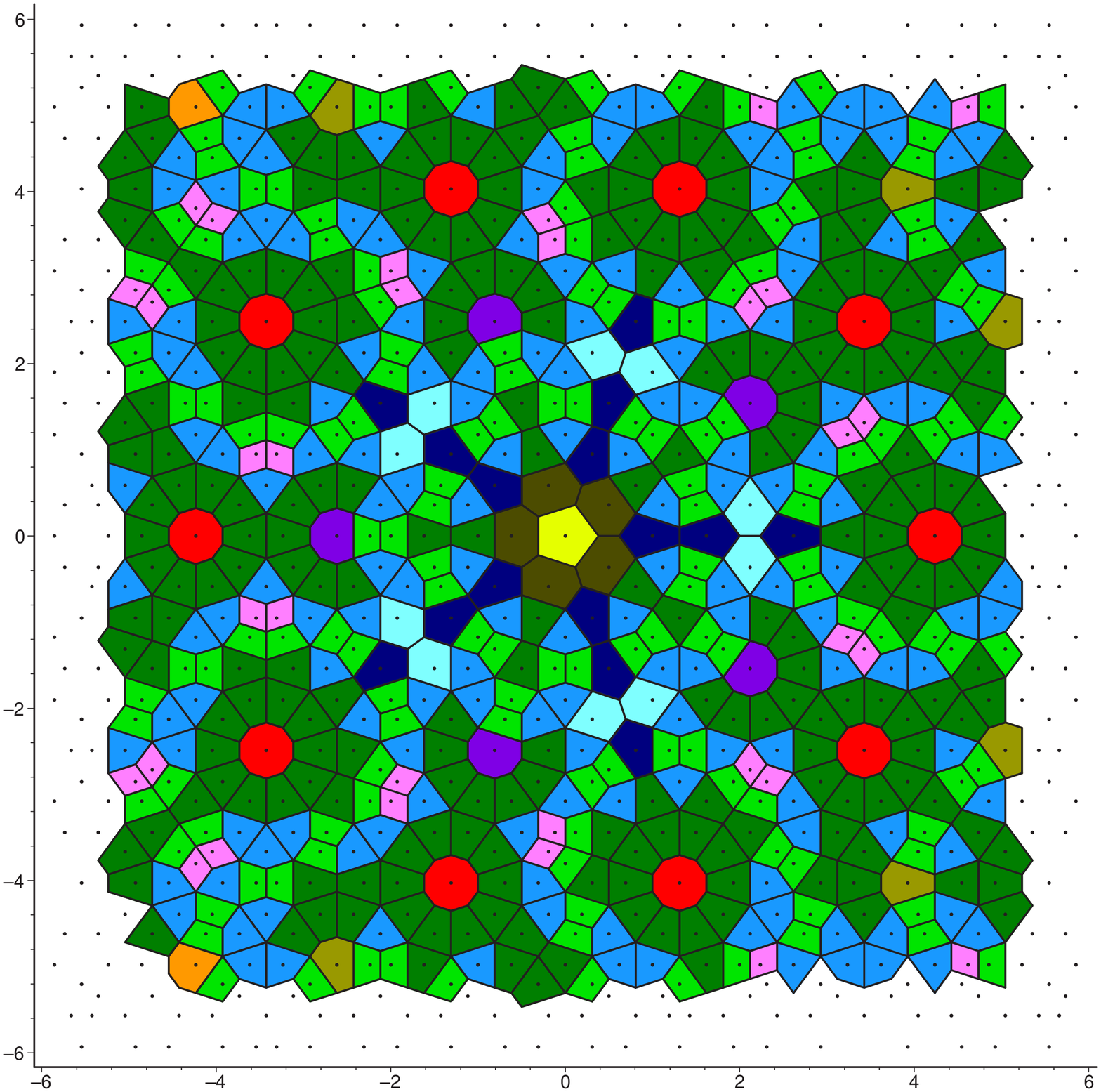}}
\subfigure[The 12 prototiles of $X^{\A_{5}}(\tau)$]{\includegraphics[scale=0.5]
    {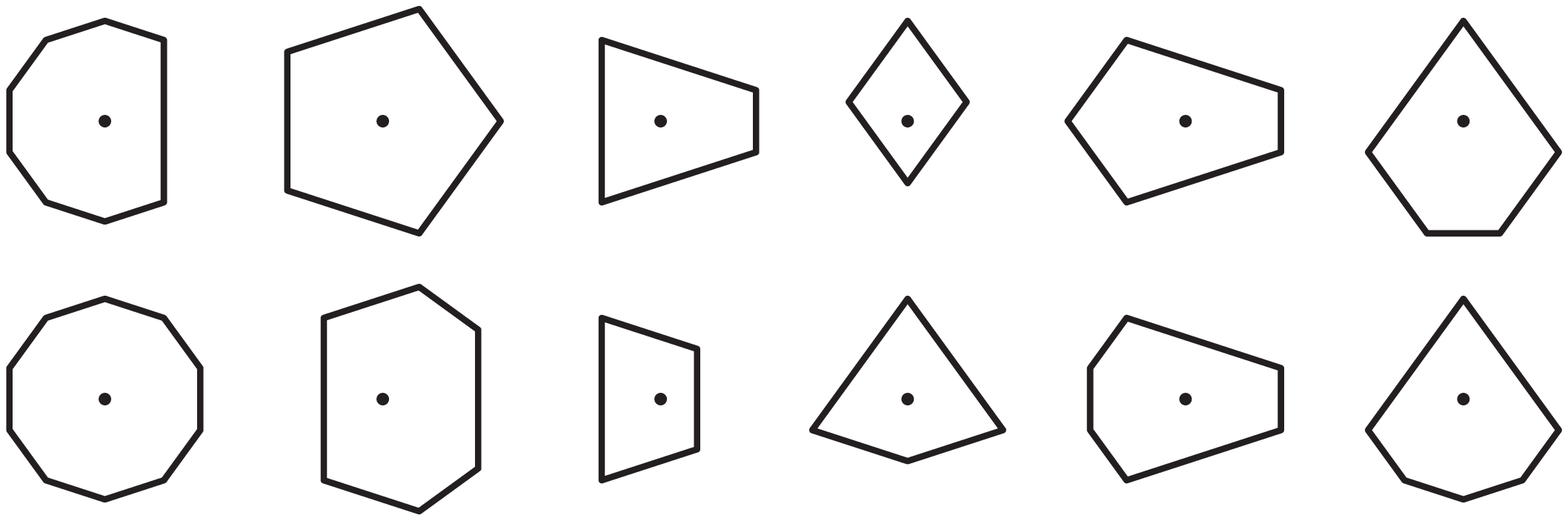}}
\caption{Tiling and acceptance window for $5$-fold symmetry and $\beta=\tau$. The colors of tiles coincide with
the colors of the corresponding regions in the acceptance windows.}
\label{fig:tiles-accw5a}
\end{figure}


\newpage
\subsection{Case $\boldsymbol{n=10}$, $\boldsymbol{\beta=\tau}$}

 Consider $K(\sigma(\tau), \A_{10})$, see Figure \ref{fig:missingquadr}.
 In Section~\ref{sec:cap}, we have seen that setting $\Omega=\mathrm{int}(K(\sigma(\tau),A_{10}))$, there are no missing points, i.e. the spectrum of $\tau$ with decagonal alphabet
 satisfies
 $$
 \Sigma(\Omega)\subseteq X^{\A_{10}}(\tau)\subseteq  \Sigma(\mathrm{cl}(\Omega)).
 $$
 Moreover, since $K(\sigma(\tau),A_{10})=K(1/\tau,A_{10})$, we can derive from Theorem~\ref{thm:ekvivalence} that points of $\sigma(X^{\A_{10}}(\tau))$ do not lie on the boundary of $K(\sigma(\tau),A_{10})$, thus we conclude that
 $$
  X^{\A_{10}}(\tau)=\Sigma(\Omega)\,,
 $$
 i.e. the acceptance window is an open decagon of radius $\tau^2$.

 The tiling of this case of spectra is illustrated in Figure~\ref{fig:tiles-accw10a}. There are 5 prototiles. Let us mention that the tiling for $\Sigma(\Omega)$  for this case was already described in~\cite{MaPaZi}. It can be shown that only 3 of the 5 prototiles appear with non-zero density
in the tiling.

\begin{figure}[ht]
\subfigure[$\sigma(X^{\A_{10}}(\tau))$]{\includegraphics[scale=0.25]
    {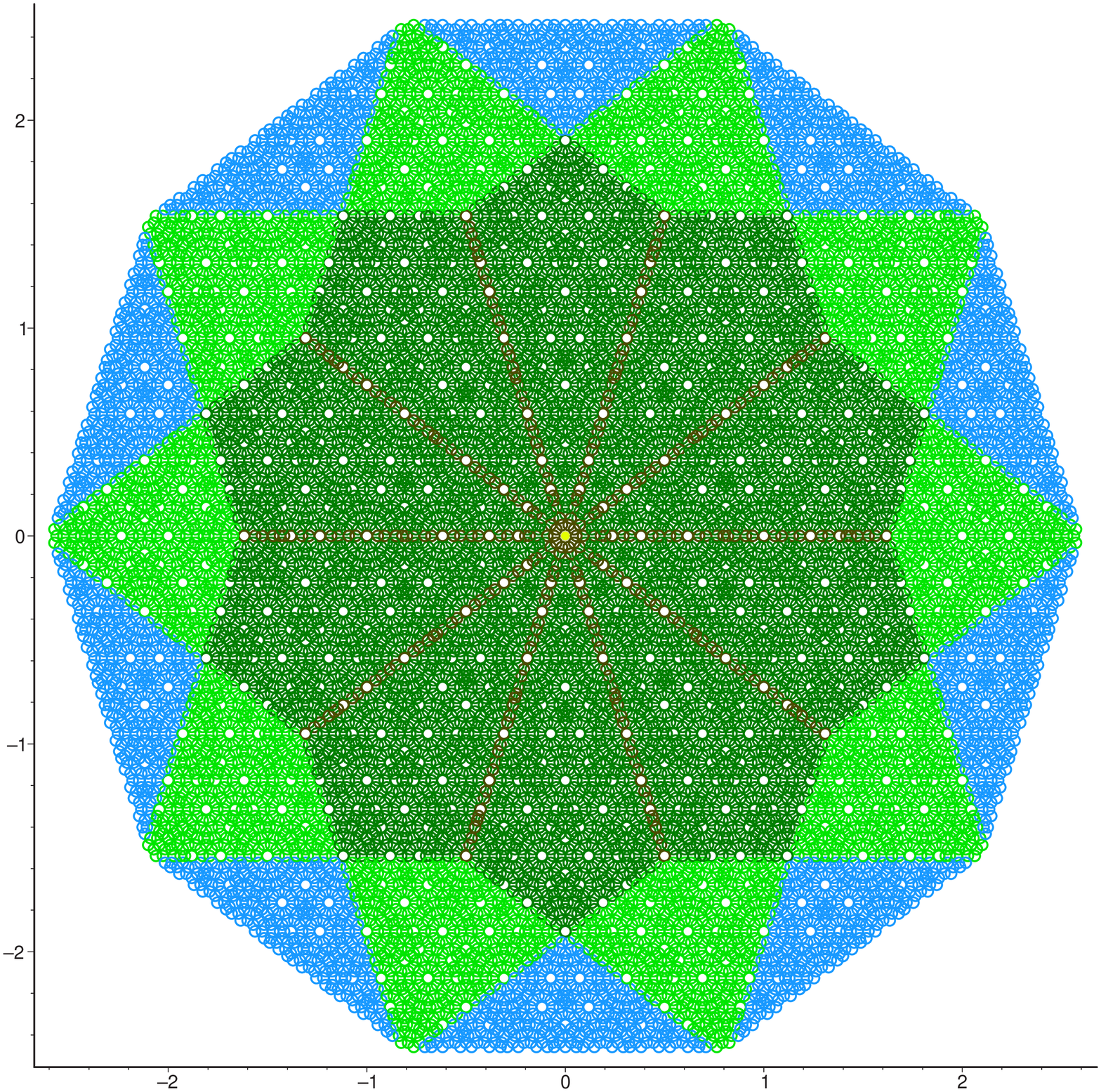}}
\subfigure[Tiling of $X^{\A_{10}}(\tau)$]{\includegraphics[scale=0.25]
    {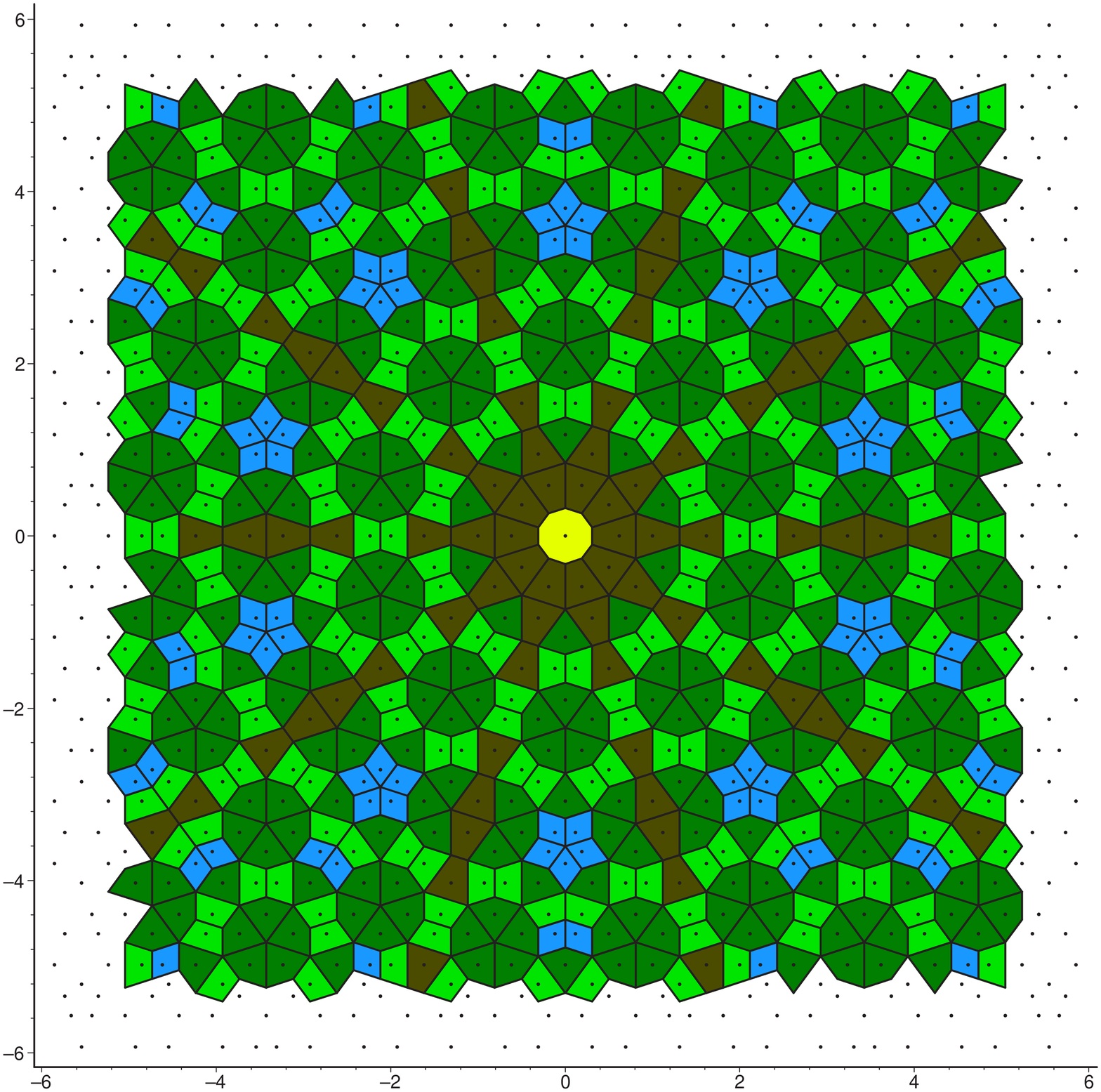}}
\subfigure[The 5 prototiles of $X^{\A_{10}}(\tau)$]{\includegraphics[scale=0.5]{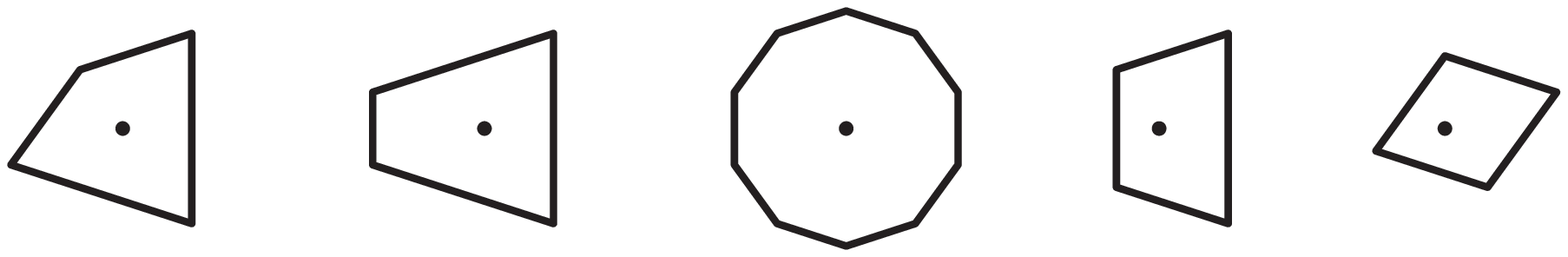}}
\caption{Tiling (a), acceptance window (b), and shapes of prototiles (c) for $10$-fold symmetry and $\beta=\tau$.
The colors of tiles coincide with the colors of the corresponding regions in the acceptance windows.}
\label{fig:tiles-accw10a}
\end{figure}


\newpage
\subsection{Case $\boldsymbol{n=10}$, $\boldsymbol{\beta=\tau^2}$}

Consider $K(\sigma(\tau^2), \A_{10})$, see Figure \ref{fig:missingquadr}.
By a similar argumentation as for $n=10$ and $\beta=\tau$, we can derive that this case of the spectra corresponds to a cut-and-project set with simply connected acceptance window which now has fractal nature. We have
 $$
X^{\A_{10}}(\tau^2)= \Sigma(\Omega),\qquad\text{ where }\quad\Omega=\mathrm{int}(K(\sigma(\tau^2),A_{10}))\,.
 $$
The tiling of this case of spectra is illustrated in Figure~\ref{fig:tiles-accw10a}. Again, there are 5 prototiles, four of which appear with non-zero density in the tiling. The decagonal tile appears only once.

\begin{figure}[ht]
\subfigure[$\sigma(X^{\A_{10}}(\tau^2))$]{\includegraphics[scale=0.25]
    {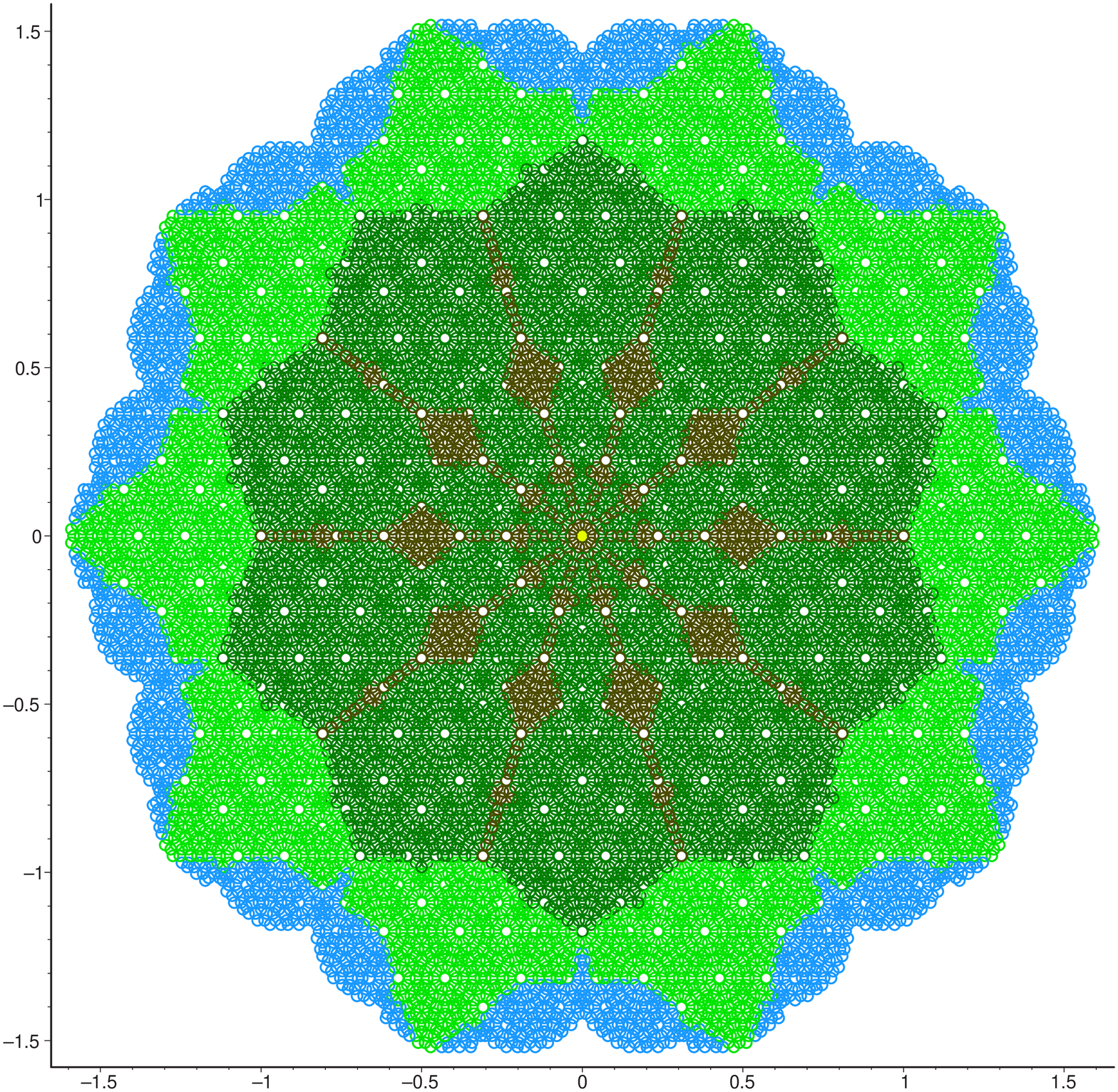}}
\subfigure[Tiling of $X^{\A_{10}}(\tau^2)$]{\includegraphics[scale=0.25]
    {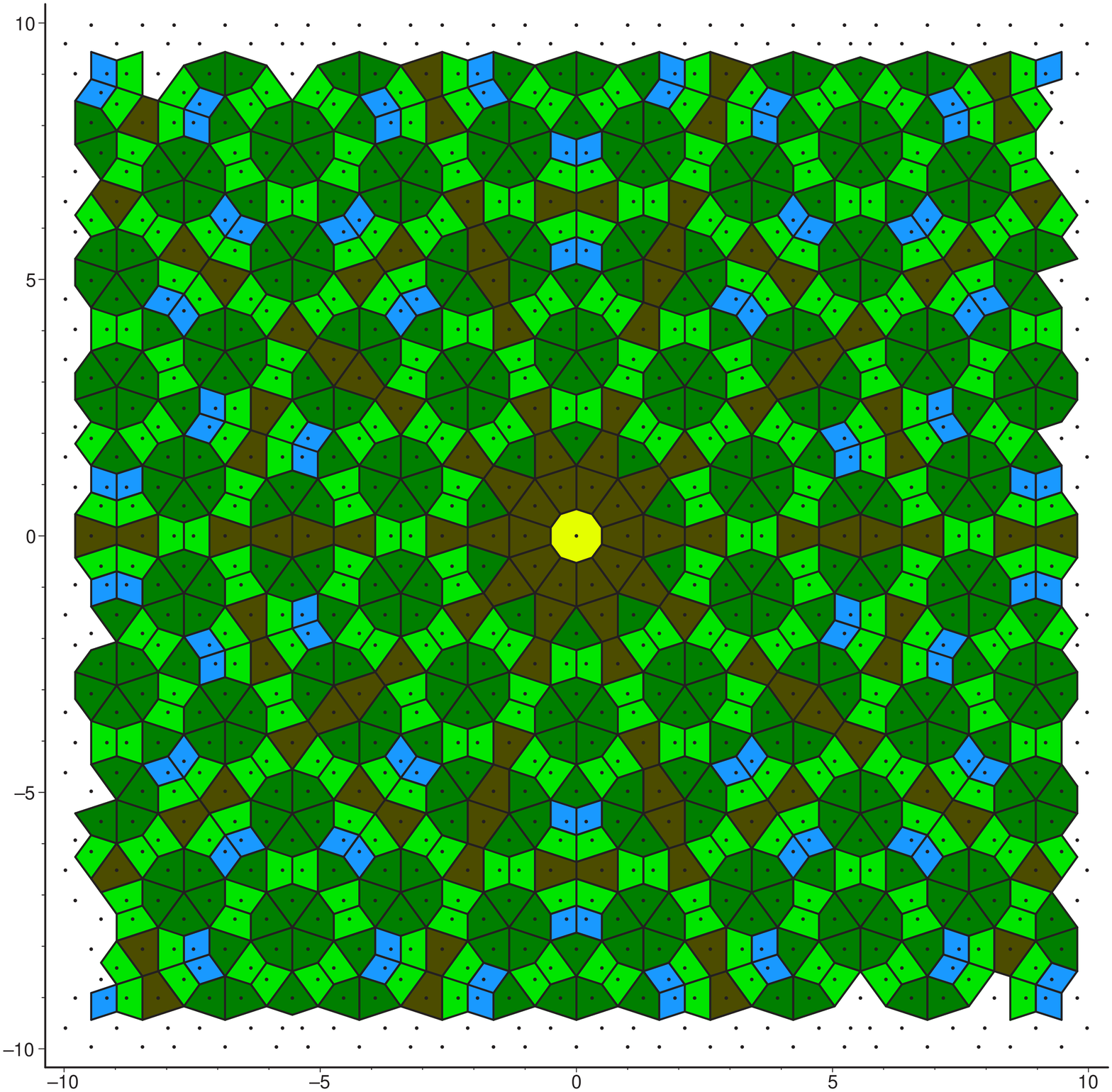}}
\subfigure[The 5 prototiles of $X^{\A_{10}}(\tau^2)$]{\includegraphics[scale=0.5]{pokus10a}}
\caption{Tiling (a), acceptance window (b), and shapes of prototiles (c) for $10$-fold symmetry and $\beta=\tau^2$.
The colors of tiles coincide with the colors of the corresponding regions in the acceptance windows.}
\label{fig:tiles-accw10b}
\end{figure}


\newpage
\subsection{Case $\boldsymbol{n=8}$, $\boldsymbol{\beta=\delta}$}

Consider $K(\sigma(\delta), \A_8)$, see Figure \ref{fig:missingquadr}.
Here,
\begin{equation}\label{eq:8}
X^{\A_8}(\delta) =
\Sigma(\mathrm{int}(K(\sigma(\delta),A_{8}))) = \Sigma(K(\sigma(\delta),A_{8})) \,.
\end{equation}
This follows from
\begin{equation}
X^{\A_8}(\delta) \subseteq
\Sigma(\mathrm{int}(K(\sigma(\delta),A_{8}))) \subseteq \Sigma(K(\sigma(\delta),A_{8})) \,,
\end{equation}
where the first inclusion uses the fact that $K(\sigma(\delta), \A_8) =K(1/\delta, \A_8) $, and item 4 of Theorem~\ref{thm:ekvivalence}. Moreover, we know that $X^{\A_8}(\delta) = \Sigma(K(\sigma(\delta),A_{8})$ due to Corollary~\ref{thm:spectra=cnp}.

It it surprising but true, that the inclusion or exclusion of the
    boundary makes no difference to the cut and project method in this case.
This is because points on the boundary do not correspond to algebraic
    integers in $\Z[\omega]$.
For example, the point on the real line with maximal absolute value
    is equal to $1 - \sigma(\delta) + \sigma(\delta)^2 - \cdots
    = \frac{1}{1+ \sigma(\delta)}\notin \Z[\omega]$.


The tiling of $X^{\A_8}(\delta)$ has exactly 5 distinct types of tiles, all of them appearing with non-zero density.

\begin{figure}[ht]
\subfigure[$\sigma(X^{\A_{8}}(\delta))$]{\includegraphics[scale=0.25]
    {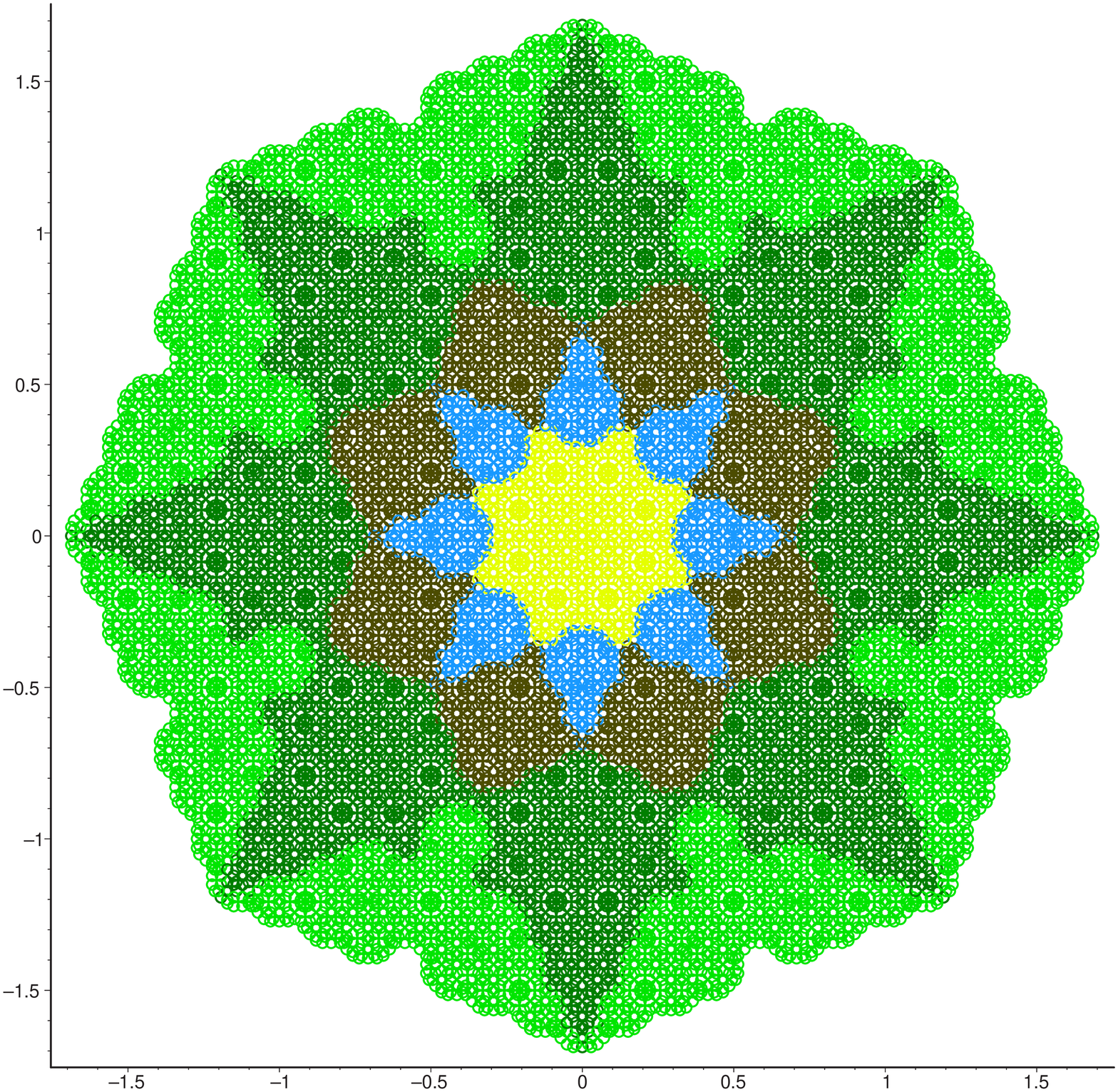}}
\subfigure[Tiling of $X^{\A_{8}}(\delta)$]{\includegraphics[scale=0.25]
    {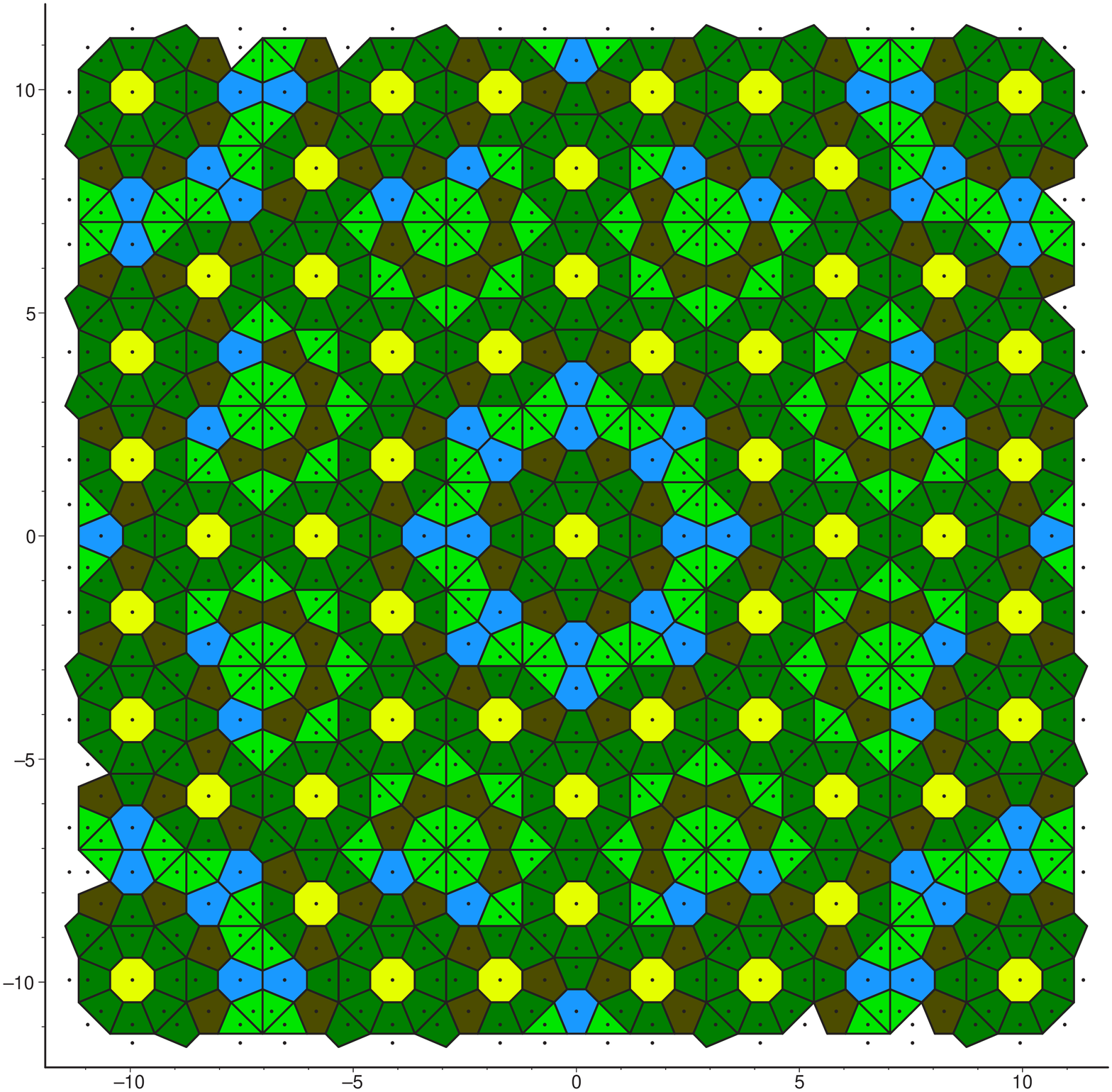}}
\subfigure[The 5 prototiles of $X^{\A_{8}}(\tau)$]{\includegraphics[scale=0.5]{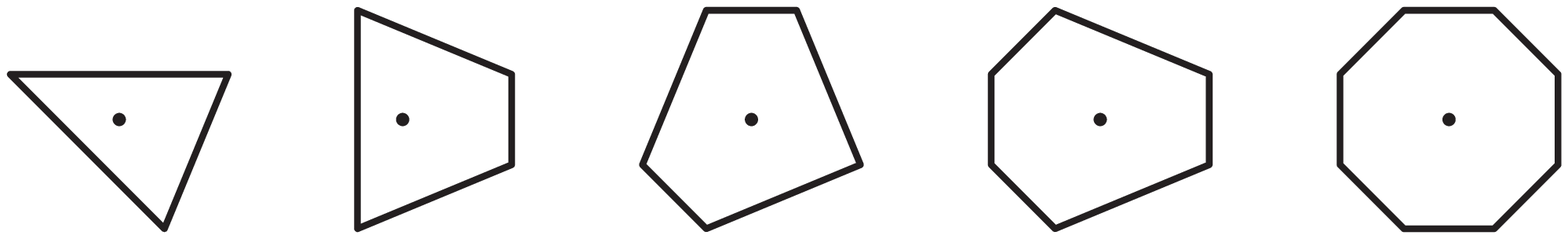}}
\caption{Tiling (a), acceptance window (b), and shapes of prototiles (c) for $8$-fold symmetry and $\beta=\delta$. The colors of tiles coincide with
the colors of the corresponding regions in the acceptance windows.}
\label{fig:tiles-accw8}
\end{figure}

\vspace*{4cm}

\newpage
\subsection{12-fold symmetry}

Consider $K(\sigma(\mu), \A_{12})$  (see Figure \ref{fig:missingquadr}).
In this case $\beta=\mu$ is not a unit, so Corollaries~\ref{thm:spectra=cnp} and~\ref{thm:spectra=cnpint} do not apply, nevertheless, we have found points in $\Sigma(\mathrm{int}(K(\sigma(\delta),A_{12})))\setminus X^{\A_{12}}(\delta)$, which means that this case of the spectra does not correspond to a cut-and-project set with simply connected acceptance window. This fact also reflects in that the tile set of this spectrum contains over hundred distinct types of tiles.


\subsection{7-, 14- and 18-fold symmetry}

Consider the candidates for the acceptance windows, according to Remark~\ref{rem:cap} and Figure \ref{fig:missingcub}. By an analysis similar to that for quadratic cases, we derive that
the spectrum cannot be identified with a cut-and-project set with reasonable acceptance window.
Our computation shows that also the Voronoi tiling is complicated:
\begin{itemize}
\item $X^{\A_{7}}(\lambda)$ has at least 201 distinct types of tiles.
\item $X^{\A_{14}}(\lambda)$ has at least 189 distinct types of tiles.
\item $X^{\A_{18}}(\kappa )$ has at least 154 distinct types of tiles.
\end{itemize}

\begin{figure}[ht]
\subfigure[Tiling for $X^{\A_{12}}(\mu))$.]{\includegraphics[scale=0.41]{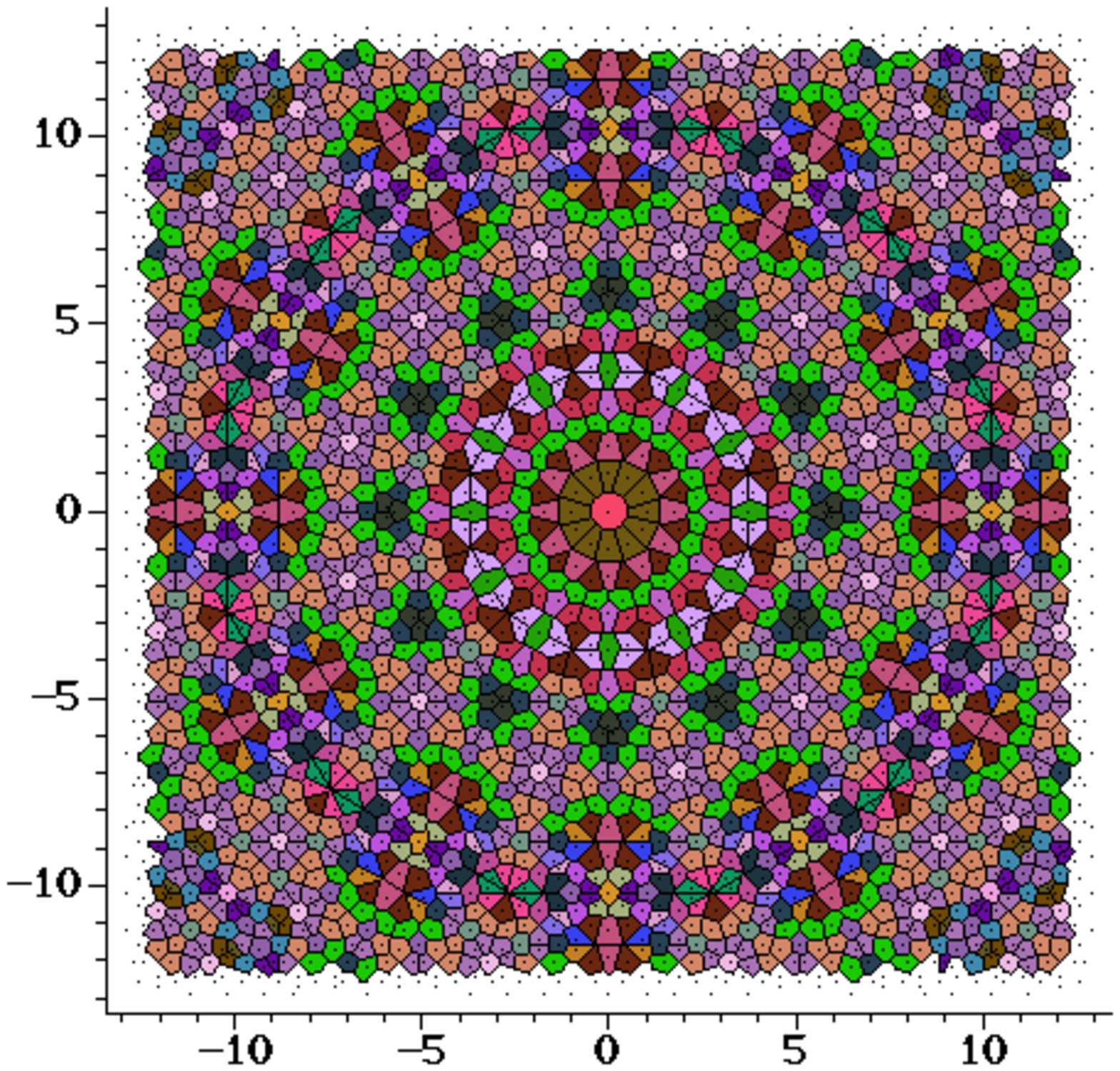}}
\subfigure[Tiling for $X^{\A_{18}}(\kappa))$]{\includegraphics[scale=0.41]{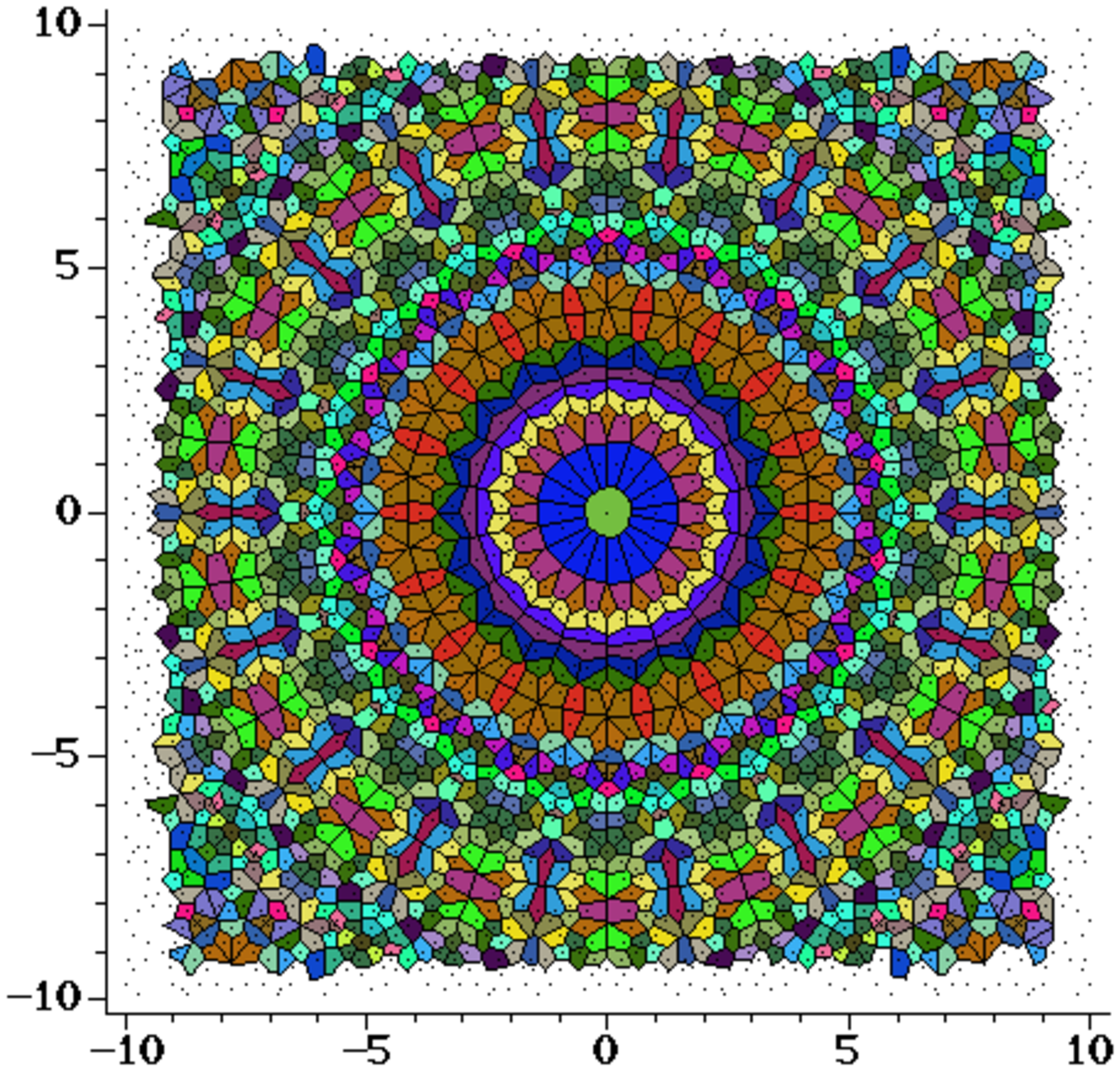}}
\subfigure[Tiling for $X^{\A_{7}}(\lambda))$.]{\includegraphics[scale=0.41]{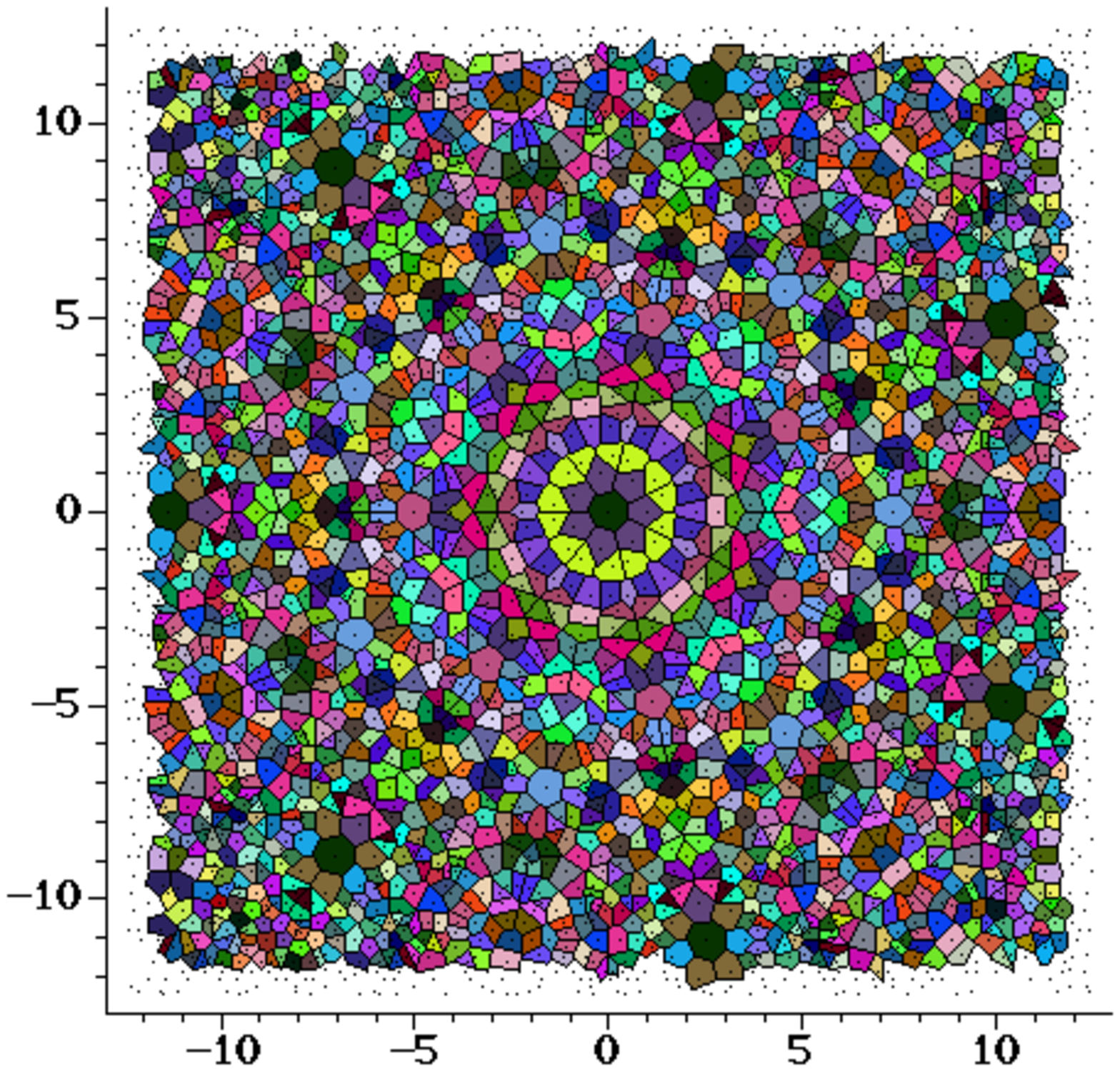}}
\subfigure[Tiling for $X^{\A_{14}}(\lambda))$.]{\includegraphics[scale=0.41]{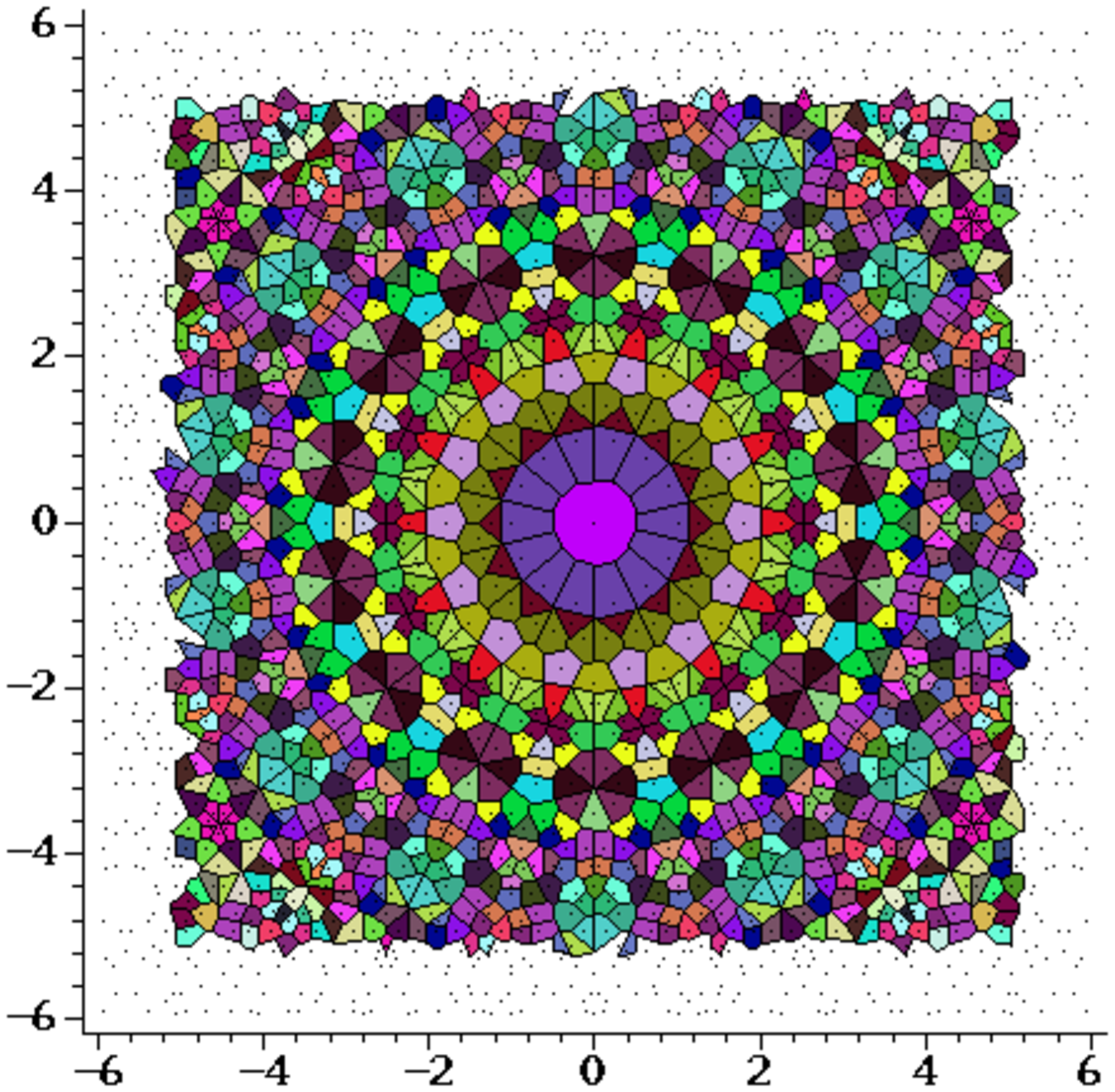}}
\caption{Tilings for the $12$-, $7$-, $14$- and $18$-fold symmetry.}
\end{figure}

\section*{Conclusions}

In this paper we have studied the possibility to model quasicrystals with symmetries of order 5 (10), 7 (14), 8, 9 (18), 12 by spectra of Pisot-cyclotomic numbers of degree two or three. It turned out that in the quadratic case, there are spectra that can be identified with self-similar quasilattices defined by the cut-and-project scheme. Similar structures have been studied in~\cite{niizeki}, however, our approach uses techniques from non-standard number systems and thus brings a new insight into the question. For, the possibility to model certain quasilattices by cut-and-project sets with simply connected acceptance windows is closely related to representability of all complex numbers in an appropriate base with digits in a suitably chosen alphabet.
Our results of number-theoretical nature are collected in Section~\ref{sec:general} and are of independent interest.

\section*{Acknowledgements}

This work was supported by the Czech Science Foundation, grant No.\ 13-03538S. We also acknowledge financial support of
the Grant Agency of the Czech Technical University in Prague, grant No.\ SGS14/205/OHK4/3T/14.
The Research of K. G. Hare was supported by NSERC Grant RGPIN-2014-03154.



\begin{thebibliography}{99}

\bibitem{BaKlSc} M. Baake, R. Klitzing, M. Schlottmann, {\it Fractally shaped acceptance domains of quasiperiodic square-triangle tilings with dedecagonal symmetry},
Physica A: Statistical Mechanics and its Applications
{\bf 191}, Issues 1–4, (1992), 554--558.

\bibitem{BaMaPeVa} S. Baker, Z. Mas\'akov\'a, E. Pelantov\'a, T. V\'avra, {\it On periodic representations in non-Pisot bases}, to appear in Monatshefte fur Mathematik (2016).

\bibitem{BellHare} J.P. Bell, K.G. Hare,  {\it A classification of (some) Pisot-Cyclotomic numbers}, J. Number Theory
{\bf 115} (2005), 215--229.


\bibitem{BuFrGaKr} C. Burd\'ik, Ch. Frougny, J. P. Gazeau, R. Krejcar,
{\it Beta-integers as natural counting systems for quasicrystals},
J. Phys. A: Math. Gen., {\bf 31}, Number 30 (1998), 6449–-6472.


\bibitem{Erdos} P. Erd\H{o}s, I. Jo\'o, V. Komornik, {\it Characterization of the unique expansions
$1 =\sum_{i=1}^\infty q^{-n_i}$ and related problems}, Bull. Soc. Math. France {\bf 118} (3) (1990), 377--390.

\bibitem{coloid}
A. Exner, S. Förster, S. Fischer, et al.
{\it Colloidal quasicrystals with 12-fold and 18-fold diffraction symmetry},
PNAS, 108(5) (2011), 1810--1814.



\bibitem{FengWen02} D.-J. Feng, Z.-Y. Wen, {\it A property of Pisot numbers}, J. Number Theory, {\bf 97} (2) (2002), 305--316.

\bibitem{Garsia61} A. M. Garsia, {\it Arithmetic properties of Bernoulli convolutions}, Trans. Amer. Math.
Soc. {\bf 102} (1962), 409--432.

\bibitem{GuMaPeBordeaux} L.-S. Guimond, Z. Mas\'akov\'a, E. Pelantov\'a, {\it Combinatorial properties of infinite words
associated with cut-and-project sequences},
J. Th\'eor. Nombres Bordeaux {\bf 15} (2003), 697--725.

%

\bibitem{HePe15} T. Hejda, E. Pelantov\'a, {\it Spectral properties of cubic complex Pisot units}, Math. Comp. {\bf 85} (2016), no. 297, 401--421.

\bibitem{herreros} Y. Herreros, {\it Contribution a l'arithmétique des ordinateurs}, Ph.D. thesis, Institut polytechnique de Grenoble, 1991.



\bibitem{MaPaZi} Z. Mas\'akov\'a, J. Patera, J. Zich, {\it Classification of Voronoi and Delone tiles of quasicrystals III: decagonal acceptance window of any size}, J. Phys. A: Math. Gen. {\bf 38} (2005), 1947--1960.

\bibitem{Moody} R.~V. Moody, {\it Model sets: A survey}, in: From
Quasicrystals to More Complex Systems, Centre de Physique des
Houches Volume 13, (2000), 145--166.

\bibitem{Icosians} R.~V. Moody, J. Patera, {\it Quasicrystals and icosians}, J. Phys. A {\bf 26} (1993), no. 12, 2829--2853.

\bibitem{NgaiWang} S.M. Ngai, Y. Wang, {\it Hausdorff dimension of self-similar sets with overlaps},
J. London Math. Soc., {\bf 63} (2001), 655--672.

\bibitem{niizeki} K. Niizeki, {\it Self-similar quasilattices with windows having fractal boundaries},
J. Phys. A: Math. Theor., {\bf 41}, Number 17, (2008), 175208.

\bibitem{niizeki12} K. Niizeki,
{\it A dodecagonal quasiperiodic tiling with a fractal window}
Phil. Mag. {\bf 87}, Iss. 18-21, (2007), 2855--2861.

\bibitem{renyi} A. R\'enyi, {\it Representations for real numbers and their ergodic properties}, Acta Mathematica Academiae Scientiarum Hungaricae, 8: (1957), 477-–493.

\bibitem{schechtman} D. Shechtman, I. Blech, D. Gratias, J. W. Cahn, {\it Metallic Phase with Long-Range Orientational Order and No Translational Symmetry},
Phys. Rev. Lett. 53, (1984), 1951.

\bibitem{steurer7}
W. Steurer, {\it Boron-based quasicrystals with sevenfold symmetry},
Phil. Mag. {\bf 87}, Iss. 18-21, (2007), 2707--2711.
\end{thebibliography}
\end{document}